\documentclass[12pt]{article}
    \usepackage[a4paper, margin=0.75in]{geometry}
\usepackage[utf8]{inputenc}         
\usepackage[english]{babel}
\usepackage{lmodern}
\usepackage{graphicx}

\usepackage[colorlinks=true]{hyperref}

\usepackage[dvipsnames]{xcolor}

\usepackage{amsthm}
\usepackage{amsmath}
\usepackage{amssymb}
\usepackage{amsfonts}
\usepackage{mathrsfs}
\usepackage{mathtools}
\usepackage{verbatim}
\usepackage{footnote}
\usepackage{lineno}
\usepackage{caption}
\usepackage{tikzsymbols}
\usepackage[capitalize, nameinlink]{cleveref}
\usepackage{tikz}
\usetikzlibrary{quantikz2}

\usepackage{icomma}
\usepackage{enumitem}
\usepackage{array}
\usepackage{multirow}
\usepackage{setspace}
\usepackage{euscript}
\usepackage{indentfirst}
\usepackage{epigraph}
\usepackage{fancybox, fancyhdr}
\usepackage{titlesec}
\usepackage{dsfont}
\usepackage{csquotes}

\newtheorem{definition}{Definition}[section]

\newtheorem{assumption}{Assumption}
\newtheorem{convention}{Convention}
\newtheorem{theorem}{Theorem}[section]

\newtheorem{proposition}{Proposition}[section]
\newtheorem{conjecture}{Conjecture}

\newtheorem{corollary}{Corollary}[section]
\newtheorem{lemma}{Lemma}[section]
\newtheorem{observation}{Observation}[section]

\newtheorem{claim}{Claim}[section]

\crefname{conjecture}{\textbf{Conjecture}}{}
\crefname{lemma}{\textbf{Lemma}}{}
\crefname{theorem}{\textbf{Theorem}}{}
\crefname{corollary}{\textbf{Corollary}}{}
\crefname{observation}{\textbf{Observation}}{}
\crefname{proposition}{\textbf{Proposition}}{}
\crefname{assumption}{\textbf{Assumption}}{}
\crefname{convention}{\textbf{Convention}}{}

\newcommand{\overbar}[1]{\mkern 1.5mu\overline{\mkern-1.5mu#1\mkern-1.5mu}\mkern 1.5mu}

\newcommand{\N}{\mathbb{N}}
\renewcommand{\S}{\mathbb{S}}

\newcommand{\R}{\mathbb{R}}

\newcommand{\C}{\mathbb{C}}

\newcommand{\Sign}{\mathrm{Sign}}

\usepackage{xparse}

\DeclarePairedDelimiterX\brackets[1]{(}{)}{
	
	#1
}
\DeclarePairedDelimiterX\squarebrackets[1]{[}{]}{
	
	#1
}
\DeclarePairedDelimiterX\figbrackets[1]{\{}{\}}{
	
	#1
}

\DeclarePairedDelimiterX\ntrs[1]{\langle}{\rangle}{
	
	#1
}
\DeclarePairedDelimiterX\kets[1]{\lvert}{\rangle}{
	
	#1
}
\DeclarePairedDelimiterX\bras[1]{\langle}{\rvert}{
	
	#1
}
\DeclarePairedDelimiterX\brakets[2]{\langle}{\rangle}{
	
	#1
    \delimsize\vert
    
    #2
}

\DeclarePairedDelimiterX\dprs[2]{\langle}{\rangle}{
	
	#1
    ,
    
    #2
}

\DeclarePairedDelimiterX\ketbras[2]{\lvert}{\rvert}{
	
	#1
    \delimsize\rangle\delimsize\langle
    
    #2
}
\DeclarePairedDelimiterX\qvals[3]{\langle}{\rangle}{
	
	#1
    \delimsize\vert
    #2
    \delimsize\vert
    
    #3
}
\DeclarePairedDelimiterX\qexps[2]{\langle}{\rangle}{
	
	#1
    \delimsize\vert
    #2
    \delimsize\vert
    
    #1
}

\let\P\undefined
\let\o\undefined
\let\O\undefined

\NewDocumentCommand{\P}{o g}{
	\IfNoValueTF{#1}{
		\IfNoValueTF{#2}{
			\mathop{\mathds{P}}
		}{
			\mathop{\mathds{P}}\squarebrackets*{#2}
		}
	}{
		\IfNoValueTF{#2}{
			\mathop{\mathds{P}}_{#1}
		}{
			\mathop{\mathds{P}}_{#1}\squarebrackets*{#2}
		}
	}
}

\NewDocumentCommand{\E}{o g}{
	\IfNoValueTF{#1}{
		\IfNoValueTF{#2}{
			\mathop{\mathds{E}}
		}{
			\mathop{\mathds{E}}\squarebrackets*{#2}
		}
	}{
		\IfNoValueTF{#2}{
			\mathop{\mathds{E}}_{#1}
		}{
			\mathop{\mathds{E}}_{#1}\squarebrackets*{#2}
		}
	}
}

\NewDocumentCommand{\D}{o g}{
	\IfNoValueTF{#1}{
		\IfNoValueTF{#2}{
			\mathop{\mathds{D}}
		}{
			\mathop{\mathds{D}}\squarebrackets*{#2}
		}
	}{
		\IfNoValueTF{#2}{
			\mathop{\mathds{D}}_{#1}
		}{
			\mathop{\mathds{D}}_{#1}\squarebrackets*{#2}
		}
	}
}

\NewDocumentCommand{\cov}{o g g}{
	\IfNoValueTF{#1}{
		\IfNoValueTF{#2}{
			\mathop{\mathrm{cov}}
		}{
			\IfNoValueTF{#3}{
				\mathop{\mathrm{cov}}
			}{
				\mathop{\mathrm{cov}}\brackets*{#2, #3}
			}
		}
	}{
		\IfNoValueTF{#2}{
			\mathop{\mathrm{cov}}_{#1}
		}{
			\IfNoValueTF{#3}{
				\mathop{\mathrm{cov}}_{#1}
			}{
				\mathop{\mathrm{cov}}_{#1}\brackets*{#2, #3}
			}
		}
	}
}

\NewDocumentCommand{\I}{g}{
	\IfNoValueTF{#1}{
		\mathop{\mathds{I}}
	}{
		\mathop{\mathds{I}}\figbrackets*{#1}
	}
}

\NewDocumentCommand{\V}{o g}{
	\IfNoValueTF{#1}{
		\IfNoValueTF{#2}{
			\mathop{\mathds{V}}
		}{
			\mathop{\mathds{V}}\squarebrackets*{#2}
		}
	}{
		\IfNoValueTF{#2}{
			\mathop{\mathds{V}}_{#1}
		}{
			\mathop{\mathds{V}}_{#1}\squarebrackets*{#2}
		}
	}
}

\NewDocumentCommand{\Tr}{o g}{
	\IfNoValueTF{#1}{
		\IfNoValueTF{#2}{
			\mathop{\mathrm{Tr}}
		}{
			\mathop{\mathrm{Tr}}\squarebrackets*{#2}
		}
	}{
		\IfNoValueTF{#2}{
			\mathop{\mathrm{Tr}}_{#1}
		}{
			\mathop{\mathrm{Tr}}_{#1}\squarebrackets*{#2}
		}
	}
}

\NewDocumentCommand{\o}{g}{\operatorname{\mathrm{o}}\IfNoValueTF{#1}{}{\brackets*{#1}}}
\NewDocumentCommand{\O}{g}{\operatorname{\mathrm{O}}\IfNoValueTF{#1}{}{\brackets*{#1}}}
\NewDocumentCommand{\Th}{g}{\operatorname{\Theta}\IfNoValueTF{#1}{}{\brackets*{#1}}}
\NewDocumentCommand{\om}{g}{\operatorname{\omega}\IfNoValueTF{#1}{}{\brackets*{#1}}}
\NewDocumentCommand{\Om}{g}{\operatorname{\Omega}\IfNoValueTF{#1}{}{\brackets*{#1}}}

\NewDocumentCommand{\ntr}{g}{\IfNoValueTF{#1}{}{\ntrs*{#1}}}
\let\ket\undefined
\NewDocumentCommand{\ket}{g}{\IfNoValueTF{#1}{}{\kets*{#1}}}
\let\bra\undefined
\NewDocumentCommand{\bra}{g}{\IfNoValueTF{#1}{}{\bras*{#1}}}
\NewDocumentCommand{\dpr}{g g}{\IfNoValueTF{#1}{}{\IfNoValueTF{#2}{}{\dprs*{#1}{#2}}}}
\let\braket\undefined
\NewDocumentCommand{\braket}{g g}{\IfNoValueTF{#1}{}{\IfNoValueTF{#2}{}{\brakets*{#1}{#2}}}}
\NewDocumentCommand{\ketbra}{g g}{\IfNoValueTF{#1}{}{\IfNoValueTF{#2}{}{\ketbras*{#1}{#2}}}}
\NewDocumentCommand{\qval}{g g g}{\IfNoValueTF{#1}{}{\IfNoValueTF{#2}{}{\IfNoValueTF{#3}{}{\qvals*{#1}{#2}{#3}}}}}
\NewDocumentCommand{\qexp}{g g}{\IfNoValueTF{#1}{}{\IfNoValueTF{#2}{}{\qexps*{#1}{#2}}}}

\newcommand{\polylog}{\mathop{\mathrm{polylog}}}

\newcommand{\wpr}{\mathop{\mathrm{w.p.}}}

\newcommand{\spec}{\mathrm{spec}}

\newcommand{\minimize}{\mathop{\mathrm{minimize}}}
\newcommand{\subto}{\mathop{\mathrm{subject\ to}}}
\newcommand{\maximize}{\mathop{\mathrm{maximize}}}

\renewcommand*\d{\mathop{}\!\mathrm{d}}
\newcommand*\dd{\mathop{}\!\partial}

\newcommand{\fe}{\varphi}
\newcommand{\eps}{\varepsilon}

\newcommand{\Bern}{\mathrm{Bern}}

\newcommand{\No}{\mathrm{N}}

\newcommand{\One}{\mathds{1}}

\newcommand{\tr}{\mathrm{tr}}

\newcommand{\cA}{\mathring{A}}

\newcommand{\eqm}{\stackrel{\mathrm{m}}{=}}
\newcommand{\dm}{\mathrm{d}}
\newcommand{\odm}{\mathrm{od}}
\newcommand{\mm}{\mathrm{m}}

\newcommand{\Rho}{\mathrm{P}}

\newcommand{\Aa}{\mathcal{A}}
\newcommand{\Bb}{\mathcal{B}}

\newcommand{\Dd}{\mathcal{D}}

\newcommand{\Xx}{\mathcal{X}}
\newcommand{\Mm}{\mathcal{M}}

\newcommand{\Ww}{\mathcal{W}}

\newcommand{\hx}{\hat{x}}

\newcommand{\supp}{\mathrm{supp}}

\newcommand{\wW}{\widetilde{W}}
\newcommand{\wWw}{\widetilde{\mathcal{W}}}
\newcommand{\wX}{\widetilde{X}}

\newcommand{\oE}{\overline{E}}

\newcommand{\hA}{\hat{A}}

\newcommand{\hX}{\hat{X}}
\newcommand{\hD}{\hat{D}}

\newcommand{\otht}{\overbar{\vartheta}}

\newcommand{\diag}{\mathrm{diag}}
\newcommand{\Diag}{\mathrm{Diag}}

\newcommand{\GOE}{\mathrm{GOE}}
\newcommand{\Ort}{\mathbb{O}}

\date{}
\title{Upper bounds on the theta function of random graphs}
\author{Uriel Feige\thanks{Weizmann Institute of Science, Israel. {\tt uriel.feige@weizmann.ac.il}} \and Vadim Grinberg\thanks{Weizmann Institute of Science, Israel. {\tt vadim.grinberg@weizmann.ac.il}}}
\begin{document}

\maketitle

\begin{abstract}
The theta function of Lovasz is a graph parameter that can be computed up to arbitrary precision in polynomial time. It plays a key role in algorithms that approximate graph parameters such as maximum independent set, maximum clique and chromatic number, or even compute them exactly in some models of random and semi-random graphs. For Erdos-Renyi random $G_{n,1/2}$ graphs, the expected value of the theta function is known to be at most $2\sqrt{n}$ and at least $\sqrt{n}$. These bounds have not been improved in over 40 years.

In this work, we introduce a new class of polynomial time computable graph parameters, where every parameter in this class is an upper bound on the theta function. We also present heuristic arguments for determining the expected values of parameters from this class in random graphs. The values suggested by these heuristic arguments are in agreement with results that we obtain experimentally, by sampling graphs at random and computing the value of the respective parameter. Based on parameters from this new class, we feel safe in conjecturing that for $G_{n,1/2}$, the expected value of the theta function is below $1.55 \sqrt{n}$. Our paper falls short of rigorously proving such an upper bound, because our analysis makes use of unproven assumptions. 
\end{abstract}

\tableofcontents

\section{Introduction}

The {\em theta} function,  denoted by $\vartheta$, is a graph parameter introduced by Lovasz~\cite{Lovasz} as an upper bound on the {\em Shannon capacity} of the graph. For a graph $G$, $\vartheta(G)$ can be computed (up to arbitrary precision) in polynomial time, via semidefinite programming (SDP). The theta function is a very useful function in combinatorial optimization, as its value is sandwiched between two NP-hard parameters. Its value is at least as large as the size $\alpha(G)$ of the maximum independent set in $G$, and at most as large as $\bar{\chi}(G)$, the clique cover number of $G$. Denoting by $\bar{G}$ the complement graph of $G$, we denote $\vartheta(\bar{G})$ by $\bar{\vartheta}(G)$. It follows that $\bar{\vartheta}(G)$ is sandwiched between the maximum clique size, denoted by $\omega(G)$, and the chromatic number, denoted by $\chi(G)$. 

The theta function and its variants play major roles in approximation algorithms for maximum clique, maximum independent set, and minimum coloring (starting with~\cite{KMS98} and~\cite{AK98}, and followed by multiple works). The theta function also plays a major role in exact algorithms for these central graph parameters in random and semirandom models of graphs (starting with~\cite{FK01}, and followed by multiple works). These latter type of results make use of estimates of the probable value of the theta function in random graphs. 

The goal of the current paper is to provide sharper estimates for the value of the theta function on Erdos-Renyi random graphs, $G_{n,1/2}$. We use $G \sim G_{n,1/2}$ to denote a graph $G$ sampled according to the $G_{n,1/2}$ distribution. It is known that $\vartheta(G \sim G_{n,1/2}) = \Theta(\sqrt{n})$ with high probability. Moreover, the leading constant in the $\Theta$ notation is known to lie between~1 and~2 (see Section~\ref{sec:background}). Hence, much is already known about the value of $\vartheta(G_{n,1/2})$. We now explain our motivation for studying $\vartheta(G_{n,1/2})$ in greater precision, and the potential benefits of such a study. 


\subsection{Background on the theta function}
\label{sec:background}

Lovasz presents several equivalent formulations of the theta function, some as a minimization problem, others as a maximization problem. Here, we shall focus on one of these formulations, as a minimization problem. Note that for convenience, we define $\bar{\vartheta}(G)$ rather than $\vartheta(G)$. 

\begin{definition}
    \label{def:theta}
    Given a graph $G(V,E)$, let $n$ denote its number of vertices, and let $\cal{M}_G$ be the class of symmetric matrices of order $n$, where $M \in \cal{M}_G$ if and only if $M_{ii} = 1$ for every $i \in [n]$, and $M_{ij} = 1$ for every $(i,j) \in E$. For a symmetric matrix $M$, let $\lambda_1(M)$ denote its largest eigenvalue. Then $\bar{\vartheta}(G) = \min_{M \in \cal{M}_G} \lambda_1(M)$.
\end{definition}

Recall that we claimed that for every graph $G$, it holds that $\bar{\vartheta}(G) \ge \omega(G)$. We sketch the proof for completeness. Let $K$ be the maximum clique in $G$, and let $k$ denote its size. The rows and columns of $M$ indexed by the vertices of $K$ induce in $M$ a $k$ by $k$  symmetric block all whose entries are~1. The quadratic form $v_K^T M v_K$,  where $v_K$ is the unit vector that has value $\frac{1}{\sqrt{k}}$ on the entries corresponding to vertices of $K$ and~0 elsewhere, has value $k$, implying that $\lambda_1(M) \ge k$.

    Let $G_{n,p}$ denote the distribution over random graphs that have $n$ vertices, and in which every distinct pair of vertices is connected by an edge independently with probability $p$. Juhasz~\cite{Juhasz} proved that for $G \sim G_{n,p}$, it holds with probability $1 - o(1)$ that:
    
    $$\frac{1}{2}\sqrt{\frac{p}{1-p}n} - O(n^{1/3} \log n) \le \bar{\vartheta}(G \sim G_{n,p}) \le {2}\sqrt{\frac{p}{1-p}n} + O(n^{1/3} \log n)$$
    
    For $G_{n,1/2}$, the result of Juhasz implies that:
    
    $$(\frac{1}{2} - o(1))\sqrt{n} \le \bar{\vartheta}(G \sim G_{n,1/2}) \le (2 + o(1))\sqrt{n}$$
    
    For $G_{n,1/2}$, $\vartheta(G \sim G_{n,1/2})$ is distributed like $\bar{\vartheta}(G \sim G_{n,1/2})$, and so the above bounds apply to $\vartheta(G \sim G_{n,1/2})$ as well.

We remark that both the upper bound and the lower bound in the result of Juhasz are immediate consequences of the results of~\cite{FuKo} concerning the eigenvalues of random matrices. In particular, the upper bound on $\bar{\vartheta}(G \sim G_{n,1/2})$ follows by taking $M_{ij}$ in Definition~\ref{def:theta} to have value $-1$ whenever $(i,j) \not\in E$. This $M$ is the $\pm 1$ adjacency matrix of $G$ (except that the entries on the diagonal are~1 rather than~0, which simply shifts all eigenvalues by 1), with $M_{ij} = 1$ if $(i,j) \in E$ and $M_{ij} = -1$ if $(i,j) \not\in E$. The lower bound is as simple. It is proved by considering one of the formulations of the theta function as a maximization problem, and plugging in it the adjacency matrix of $G$. 

The gap between the upper bound and lower bound proved by Juhasz for $G_{n,p}$ is a factor of~4. The gap can be reduced to a factor of~2 in the special case of $G_{n,1/2}$. This is because for every graph $G$ it holds that $\vartheta(G) \cdot \bar{\vartheta}(G) \ge n$~\cite{Lovasz}. Moreover, for random graphs, the value of the theta function is concentrated around its mean~\cite{ACO05}. Recalling that $\vartheta(G_{n,1/2})$ and $\bar{\vartheta}(G_{n,1/2})$ have the same distribution, the combination of the above facts implies that with high probability, $\vartheta(G \sim G_{n,1/2}) \ge (1 - o(1))\sqrt{n}$.

Recall that the theta function can be estimated up to arbitrary precision in polynomial time, using an SDP. This allows us to test experimentally the value of $\vartheta(G \sim G_{n,1/2})$ for quite large values of $n$. We have carried out such experiments, and the results were that the value of the theta function is very close to $\sqrt{n}$. This suggests that the upper bound of $2\sqrt{n}$ can be improved. Attempting to improve this upper bound is the focus of our paper.
\begin{figure}[h]
\includegraphics[width=16cm]{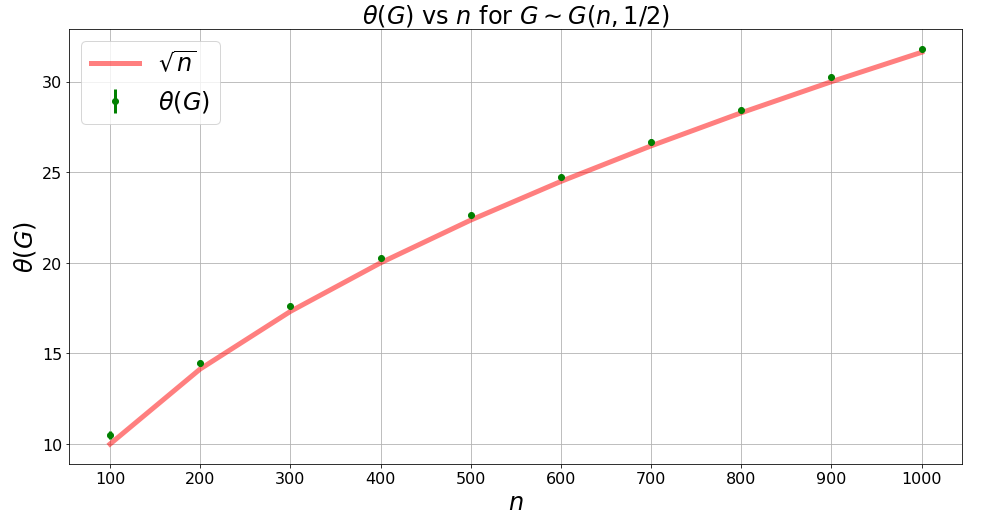}
\label{fig:theta}
\caption{Experimental values of $\bar{\vartheta}(G)$ for $G\sim G(n, 1/2)$ for various $n$.}
\end{figure}

Why do we find improving the upper bound on $\vartheta(G \sim G_{n,1/2})$ to be a strongly motivated question? This is not because of the possible quantitative improvements that might be achieved, as they are relatively minor (at most a factor of~2). Rather, this is because an improved upper bound appears to require a substantially new understanding of the theta function. Recall Definition~\ref{def:theta}, in which the value of  $\bar{\vartheta}$ is a solution to an optimization problem. The known upper bound on $\bar{\vartheta}(G \sim G_{n,1/2})$ simply takes $M \in \cal{M}_G$ to be the $\pm 1$ adjacency matrix of $G$. In other words, it takes the input graph as the solution to the optimization problem, without attempting to optimize anything about it. An improved upper bound can no longer do this -- it should tell us something about what can be optimized in $M$ so as to reduce its largest eigenvalue. In a qualitative sense, an improved upper bound would be substantially different from the existing upper bound, even if quantitatively the difference might be small.

\subsection{Further motivation: the planted clique problem}
\label{sec:platedClique}

The planted clique problem is an extensively studied optimization problem that appears to capture the limitations of current algorithmic techniques. For a given value of $k$, the distribution $G_{n,1/2,k}$ is a distribution over random graphs generated by first sampling a graph $G' \sim G_{n,1/2}$, then selecting uniformly at random a set $K$ of $k$ vertices in $G'$, and finally, the graph $G \sim G_{n,1/2,k}$ is obtained by changing the subgraph induced on $K$ to a clique.  Given $G \sim G_{n,1/2,k}$ as input, the goal is to return the maximum clique in $G$ (with high probability over the choice of $G \sim G_{n,1/2,k}$). If $k$ is sufficiently large (somewhat larger than $2\log n$ suffices), then with high probability this maximum clique is $K$.

When $k > C\sqrt{n\log n}$ for a sufficiently large constant $C$, finding $K$ is trivial, because with high probability, the vertices of $K$ are the $k$ vertices of highest degree in $G$. There are several polynomial time algorithms that are known to find $K$ when $k \ge C\sqrt{n}$ (for example, \cite{AKS98, FK00, FR10, DGP14, DM15}). However, no polynomial time is known for $k \le o(\sqrt{n})$, and it is proved that many known polynomial time algorithms fail in this regime~\cite{Jerrum92, FK03, BarakHKKMP19, FeldmanGRVX17, GZ19}. A conjecture that there are no such polynomial time algorithms implies hardness of many additional problems~\cite{HazanK11, BrennanBH18, ManurangsiRS21}.

A decision problem related to the planted clique search problem is that of distinguishing between $G_{n,1/2,k}$ and $G_{n,1/2}$. That is, given a graph $G$ sampled at random either from $G_{n,1/2,k}$ or from $G_{n,1/2}$, determine from which distribution it was sampled. Even for $k \le o(\sqrt{n})$, there are polynomial time statistical tests that have probability somewhat larger than $\frac{1}{2}$ of guessing the right answer (e.g., answer $G_{n,1/2}$ if the number of edges is smaller than $\frac{1}{2} {n \choose 2} + \frac{1}{4} {k \choose 2}$, and $G_{n,1/2,k}$ otherwise), but this advantage over $\frac{1}{2}$ is small and does not seem to lead to useful algorithms for the search problem. In contrast, it is known~\cite{HiraharaS24} that any polynomial time statistical test that is correct (over choice of random $G$) with probability sufficiently close to~1 can be leveraged to give a polynomial time algorithm for the search problem (there is a small degradation in the value of $k$ in this process). 

One of the polynomial time statistical tests that distinguishes between $G_{n,1/2,k}$ and $G_{n,1/2}$ is the value of $\bar\vartheta(G)$. As we have seen, for $G \sim G_{n,1/2}$ the value is at most $2\sqrt{n}$ (and highly unlikely to exceed this value), whereas for $G \sim G_{n,1/2,k}$, the value is at least $k$. For $k > (2 + \epsilon)\sqrt{n}$,  this gives a statistical test based on comparing $\bar\vartheta(G)$ with $(2 + \frac{\epsilon}{2})\sqrt{n}$.

One may ask whether the theta function might be as effective a statistical test even for much smaller values of $k$, such as $n^{0.4}$. Naively, the answer seems to be negative, because for $G \sim G_{n,1/2}$ we have that $\bar\vartheta(G) \ge \sqrt{n}$, which is larger than $k$. However, this leaves open various unexplored possibilities. Among them:

\begin{enumerate}
    \item It could be that the expected value of $\bar\vartheta(G \sim  G_{n,1/2,k})$ differs from that of $\bar\vartheta(G \sim  G_{n,1/2})$ only by low order terms, but that this small difference is still much larger than the standard deviations of these random variables, and hence suffices in order to obtain a reliable distinguisher.
    \item It could be that the associated matrices $M$ that certify the values of $\bar\vartheta(G \sim  G_{n,1/2,k})$ have polynomial time computable statistics (different from $\lambda_1$, that simply returns $\bar\vartheta$) that make them distinguishable from those $M$ that certify the values of $\bar\vartheta(G \sim  G_{n,1/2})$. For example, it might be that the $k$ rows of highest norm have a noticeably higher average value in the $G_{n,1/2,k}$ case compared to the $G_{n,1/2}$ case.
\end{enumerate}

Testing the above possibilities requires a much better understanding of the theta function than we currently have, both for $G \sim  G_{n,1/2}$ and for $G \sim  G_{n,1/2, k}$. As a first step, we wish to improve our understanding of the theta function for $G \sim  G_{n,1/2}$, and this is part of the motivation for the current paper.

For the goal of breaking the barrier of $k \ge \Omega(\sqrt{n})$ in the planted clique problem, the statistical test may be based on graph parameters different from the theta function. In our attempt to upper bound the theta function, we introduce new graph parameters that serve as upper bounds on the theta function. A detailed understanding of these new graph parameters (understanding their expectation and variance, or the structure of the associated matrices that certify their value) might also lead to distinguishers when $k \le o(\sqrt{n})$.

\subsection{A related problem: optimizing the spectral radius}

In addition to Lovasz theta function, we have also considered another upper bound on the largest clique $\omega(G)$: the spectral radius $\bar{\rho}(G)$.
Recall that for a symmetric matrix $A$, its spectral radius $\rho(A)$ is the largest absolute value of any of its eigenvalues. Hence, $\rho(A) = \max[\lambda_1(A), -\lambda_n(A)]$, and $\bar{\rho}(G)$ corresponds to the minimal value of $\rho(M)$ for matrices $M \in \Mm_G$.
This function is more difficult to optimize than $\otht(G)$, since we have to bound both the largest and the smallest eigenvalues of $M$ at the same time.

For the $\pm 1$ adjacency matrix $A_G$ of a graph $G$, the spectral radius serves as a simultaneous upper bound both on the maximum clique size and on the maximum independent set size (up to a $+1$ term, due to the diagonal). That is:

$$\rho(A_G) + 1 \ge \max[\omega(G),\alpha(G)]$$

In general, $\rho(A_G)$ is less tight than $\lambda_1(A_G)$ as an upper bound on $\omega(G)$. Not only is it the case that $\lambda_1(M) \leq \rho (M)$ for every matrix $M$, but also, $\rho(A_G) \ge \sqrt{n-1}$ for every graph. This follows because the Frobenius norm of a $\pm 1$ adjacency matrix is precisely $n(n-1)$, implying that the square of at least one of its eigenvalues is at least $n-1$. This lower bound on $\rho(A_G)$ is nearly matched by the following family of graphs.
Consider any graph $G$ whose adjacency matrix is a $\pm 1$ symmetric Hadamard matrix (in which every two rows are orthogonal to each other) in which one zeroes the diagonal. For such a graph, $\rho(A_G) \le \sqrt{n} + 1$ (if not for the need to zero the diagonal, $\rho$ would be precisely $\sqrt{n}$).

For $G \sim G_{n,1/2}$, $\lambda_1(A_G) = -\lambda_n(A_G)$ up to low order additive terms, and hence $\rho(A_G)$ is as good an upper bound as $\lambda_1(A_G)$ on $\omega(G)$.
Similar to Definition~\ref{def:theta} that asks for $M \in \cal{M}_G$ that minimizes $\lambda_1(M)$, we may ask for $M \in \cal{M}_G$ that minimizes $\rho(M)$. We denote the resulting minimum value of $\rho$ by $\bar{\rho}(G)$. 

$$\bar{\rho}(G) = \min_{M \in {\cal{M}}_G} \rho(M)$$

{In this work we also try to improve the known upper bound on $\bar{\rho}(G\sim G_{n, 1/2})$, which is a task that is more challenging than improving the upper bound on $\otht(G\sim G_{n, 1/2})$.}

\subsection{Related work}
\label{sec:related}

In Section~\ref{sec:background} we already mentioned the known bounds on the theta function for $G_{n.1/2}$, or $G_{n,p}$. There are other classes of graphs for which the value of the theta function is well understood. Most notably, as pointed out in~\cite{Lovasz}, for the class of {\em perfect graphs} (a graph is perfect if for it and any of its induced subgraphs, the size of the maximum independent set equals the clique cover number), the value of the theta function exactly equals that of the maximum independent set. 

In our work we build upon known results in random matrix theory. Some of the relevant results in this area are briefly surveyed in Section~\ref{sec:randomMatrix}. For a general reference, see~\cite{AGZ09}, for example. In particular, we shall make use of the notion of {\em free convolution}, {developed in the context of {\em free probability}}~\cite{V91, S93}. We are not aware of previous work that used free convolution in the context of estimating the value of the theta function.
{When it comes to estimating other graph parameters, free convolution has been used in the context of constructing expanders.
For example, in the works of \cite{MSS15, MSS18, MSS22} the authors {developed and} used techniques of so-called \textit{finite free probability} to prove that there exist bipartite Ramanujan graphs of every degree and every number
of vertices.
In addition, finite free probability techniques have recently been used to get a tight bound for the spectral expansion of the Zig-Zag product of two graphs, as well as other replacement products \cite{CCM24}.}

As mentioned in Section~\ref{sec:platedClique}, one of our motivations in this work is so as to open up new possibilities for developing new algorithms for the planted clique problem. Results concerning that problem are surveyed in Section~\ref{sec:platedClique}.

\subsection{Relevant background in random matrix theory}
\label{sec:randomMatrix}

We briefly present some background in random matrix theory. More details and more rigorous formalism appear in later sections of our paper.

Some of the random matrices encountered in our study are so called {\em Wigner} matrices. 

\begin{definition}\label{wigdef}
    Consider a family of independent zero mean real random variables $\{Z_{ij}\}_{1\leq i < j \le n}$, with $\E{Z_{ij}^2} = 1$, independent from a family $\{Y_i\}_{i=1}^n$ of i.i.d. zero mean real random variables, with $\E{Y_i^2}$ finite.
    Let $\Xx_n$ be a probability distribution over symmetric $n\times n$ matrix $X_n$ with entries
    \[X_n(i, j) = X_n(j, i) = \begin{cases}
       Z_{ij},&i < j;\\
       Y_i,&i = j.
    \end{cases}\]
    We call such matrix $X_n \sim \Xx_n$ a \textbf{Wigner matrix} with laws $(Z_{ij}, Y_i)$.
\end{definition}

Examples of Wigner matrices include Gaussian Orthogonal Ensemble $\GOE(n)$, where $X_{ij} \sim \No(0, 1)$ for $i < j$ and $X_{ii} \sim \No(0, 2)$, and Bernoulli, where $X_{ij}$ takes values $\pm 1$ with equal probability (and hence $X$ is the $\pm 1$ adjacency matrix of a random graph).

A symmetric matrix $X$ has real valued eigenvalues $\lambda_1(X) \ge \ldots \ge \lambda_n(X)$ and associated eigenvectors $v_1, \ldots, v_n$, where for every $i\in [n]$, $X v_i = \lambda_i v_i$. The set of eigenvalues of a matrix is referred to as its {\em spectrum}, and the set of eigenvectors forms a basis for $R^n$, referred to as the {\em eigenbasis}. Wigner \cite{Wig55,Wig58} established the limiting distribution of the eigenvalues of Wigner matrices $X\sim \Xx_n$ (or more exactly, of matrices sampled at random from a distribution as in Definition~\ref{wigdef}). Scaling $X$ by $1/\sqrt{n}$, its eigenvalues follows a {\em semi-circle} law, with density function $\frac{1}{2\pi}\sqrt{4 - x^2}$,  supported over the interval $x \in [-2, 2]$. Moreover, Wigner matrices enjoy a property referred to as  {\em rigidity} of eigenvalues~\cite{KY13}, saying that with high probability, simultaneously for every $k \in [n]$, the deviation of the $k$th the eigenvalue from its predicted value according to the semi-circle law is a very small low order term.

Much is also understood about the eigenvectors of Wigner matrices. They enjoy the {\em delocalization} property (see~\cite{AGZ09}, for example), which means that if a Wigner matrix is sampled at random, the direction of its eigenvectors is basically random. For Gaussian Wigner matrices, the directions of the eigenvectors are uniformly random, and moreover, independent of the eigenvalues. For other classes of Wigner matrices, they are random in some weaker sense (for every fixed direction, the maximum projection that any of the eigenvectors is expected to have in that direction is not much larger than what one would get if the eigenvectors were completely random)~\cite{BL22}.

In our study we shall also encounter {\em generalized Wigner matrices}. These are symmetric matrices defined in a way similar to Definition~\ref{wigdef}, except that each of the $Z_{ij}$ random variables may have a different variance (in Wigner matrices all these variances are the same and equal to~1). Under certain conditions on these variances, the eigenvalues of generalized Wigner matrices also follow a semi-circle law, where the width of the support of this semi-circle depends on the variances~\cite{GNT15,C23}.

For our analysis, we shall often face a situation in which we have two matrices, $A$ and $B$, each with a known spectrum (set of eigenvalues), and we wish to determine the spectrum of the sum $A + B$. This can easily done in the case of {\em identical eigenbasis}. If $A$ and $B$ have the same eigenvectors (though not necessarily the same eigenvalues), then the eigenvectors of $A + B$ are the same as those of $A$ and $B$, and the eigenvalue corresponding to each eigenvector is the sum of the corresponding eigenvalues of $A$ and $B$. In turns out that the spectrum of the sum can sometimes be estimated even when the individual eigenvectors are not known. At an informal level, this happens in a case that is the opposite of the case of identical eigenbasis, which is when the eigenvectors of $A$ have almost no correlation with those of $B$. This involves the notion of {\em free convolution}~\cite{V91, S93}.
    
Suppose that one of the two matrices, say $B$, is chosen at random from a distribution $\cal{D}$ over matrices, where intuitively, we would like $\cal{D}$ to enjoy the delocalization property for eigenvalues. For some distributions $\cal{D}$, and sometimes also some restrictions on $A$ are required, the matrices $A$ and random $B$ can be shown to be {\em freely independent}, which is a technical condition that is explained in later sections of this paper.  When free independence applies, the  spectrum of $A + B$  is approximated well by a so called {\em free convolution} of the spectra of $A$ and $B$. This allows us to determine the spectrum of $A + B$ with sufficient accuracy.

We alert the reader that determining the {\em bulk} of the spectrum, which includes all but a small fraction of the eigenvalues, is often easier than determining the full spectrum (which includes also $\lambda_1$, the eigenvalue that we are most interested in). For example, this happens when one uses proof techniques that are based on only finitely many moments of the distribution of eigenvectors. Hence, proofs regarding the full spectrum sometimes involve two separate arguments, one determining the bulk of the spectrum, and the other showing that the largest eigenvalue  {\em sticks} to the support of the bulk. For example, the original proof of Wigner's semi-circle law addressed the bulk of the spectrum, and the fact that $\lambda_1$ sticks to the support of the bulk was proved only much later~\cite{FuKo}.

\subsection{Our approach for upper bounding the theta function}
\label{sec:approach}

Recall Definition~\ref{def:theta} of $\bar\vartheta$. Given a graph $G(V,E)$, it asks for the matrix $M \in \cal{M}_G$ that minimizes $\lambda_1(M)$. In this matrix $M$, $M_{ii} = 1$ for every $i$, $M_{ij} = 1$ for every $(i,j) \in E$, and we are left to determine the other entries ($M_{ij}$ for $(i,j) \not\in E$). We refer to these entries as the {\em free} entries. What values should we put in the free entries so as to minimize $\lambda_1(M)$?

One approach is to fix some distribution $\cal{D}$, and for each free entry, select independently a value from $\cal{D}$. When $\cal{D}$ is the distribution that picks the value $-1$ with probability~1 (hence, in this case $\cal{D}$ is deterministic), we get the $\pm 1$ adjacency matrix of $G$ (except that the entries on the diagonal are $1$ instead of~0), which we refer to as $A_G$. By the results of~\cite{FuKo}, $\lambda_1(A_G) \simeq 2\sqrt{n}$ with high probability. This gives the known upper bound of $2\sqrt{n}$ on $\bar\vartheta(G \sim G_{n,1/2})$. 

Is there a choice for $\cal{D}$ for which the expected value of $\lambda_1$ drops below $2\sqrt{n}$? The answer is negative. We show that for every distribution $\cal{D}$, if the free entries of $M$ are chosen independently from $\cal{D}$, {with high probability the value of $\lambda_1(M)$ is at least $(2 - o(1))\sqrt{n}$ for $n$ large enough, up to low order additive terms.
Concretely, we show the following statement, proved in \cref{sec:analysis}.}

\begin{proposition}\label{freeentdist}
    Let $\Dd$ be any distribution over $\R$ with finite fifth moment, and suppose that the free entries of $n\times n$ matrix $M \in \Mm_G$ are chosen independently from $\Dd$.
    Then, for $n$ large enough, $\lambda_1(M) \geq (2 - o(1))\sqrt{n}$ with high probability.
\end{proposition}

It follows that the values of free entries cannot be selected independently from each other -- they need to be correlated in some way. This brings us into the territory of estimating the eigenvalues of matrices whose entries are random but correlated, which is much less understood than that of random matrices with independent random entries.

For simplicity in what follows, we assume that the diagonal entries of $M$ need to be~0, not~1. This affects the eigenvalues only by an additive term of~1, which is negligible for our purpose.

We may decompose the matrix $M$ into $A_G + Y$, where $A_G$ is the $\pm 1$ adjacency matrix of $G$, and $Y = M - A_G$. So the question becomes that of how to choose $Y$ so that we can prove that with high probability, $\lambda_1(A_G + Y) < \lambda_1(A_G)$. This relates to estimating the spectrum (the set of eigenvalues) of sums of matrices ($A_G$ and $Y$, in our case). Recall that in section~\ref{sec:randomMatrix} we reviewed  two cases in which the spectrum of sums of two matrices can be estimated based on the spectra of the individual matrices. One case was that of identical eigenbasis. It is hopeless to try to apply it here, because it seems impossible that $Y$, which is highly constrained by requiring roughly half its entries to be~0, will have the same eigenvectors as $A_G$. The other case was that of free convolution. Using this appears to be equally hopeless. As the trace of $Y$ has to be~0, the sum of its eigenvalues is~0. Under such conditions, free convolution will not give $\lambda_1(A_G + Y) \le \lambda_1(A_G)$.

Given the above difficulties, we now describe our approach. Instead of choosing a symmetric matrix $Y$ that is constrained to have $Y_{ij} = 0$ whenever $(i,j) \in E$, we choose a symmetric matrix $Z$ that has no such constraints. The matrix $Y$  will then be the matrix $Z$ restricted to the free entries. That is,  $Y_{ij} = Z_{ij}$ for $(i,j) \not\in E$, and the rest of the entries of $Y$ (including the diagonal) are~0.

To allow for our method of analysing $\lambda_1(A_G + Y)$, we wish $Z$ to have the following properties:

\begin{enumerate}
    \item {\em Same eigenbasis:} The eigenbasis of $A_G$ is an eigenbasis of $Z$. (All eigenvalues of $A_G$ are likely to be distinct, and hence $A_G$ has $n$ distinct eigenvectors. We allow some of the eigenvalues of $Z$ to have multiplicities, so there might be more than one way of choosing an eigenbasis for $Z$, but one of this choices must be identical to the eigenbasis of $A_G$.)
    \item {\em Concentrated diagonal:} All entries of $D_Z$, the diagonal of $Z$, have roughly the same value. That is, there is some $d$ such that for all $i$, $D_{ii} = (1 \pm o(1))d$.
    \item {\em Small average:} {The absolute value of the average of entries of $Z$ in free locations (those with $(i,j) \not\in E$) is $o(1/\sqrt{n})$, and likewise for the non-free locations (those with $(i,j) \in E$). } 
    \item {\em Pseudorandomness:} This is a collection of properties that we shall not completely specify here, but are required for inner details of some our proofs. They include properties such as that all rows have roughly the same $\ell_2$ norm, and that every entry has absolute value smaller than $\sqrt{n}$.   
\end{enumerate}

Our procedure for selecting $Z$ is based on the eigenvector decomposition of $Z$ {(which will be the same as that for $A_G$)}. Let $\lambda_1(A_G) \ge \ldots \ge \lambda_n(A_G)$ denote the eigenvalues of $A_G$, and let $v_1, \ldots, v_n$ denote the corresponding eigenvectors. We select eigenvalues $\alpha_1, \ldots, \alpha_n$ for $Z$, and then we have that $Z = \sum_{i=1}^n \alpha_i v_i v_i^T$.  The values of $\alpha_1, \ldots, \alpha_n$ are chosen in such a way that it makes it plausible to believe that the resulting $Z$ will have properties as above. More details on this will appear later. 

Suppose that our choice of $\alpha_1, \ldots, \alpha_n$ leads to a matrix $Z$ with the above properties. How do we then analyse $\lambda_1(A_G + Y)$, where $Y$ is derived from $Z$ as explained above? In explaining our approach, we use the following notation. For two matrices $A$ and $B$ of size $n$, $A \circ B$ denotes their entry-wise product. This is the size $n$ matrix $C$ that satisfies $C_{ij} = A_{ij} \cdot B_{ij}$ for every $i$ and $j$. We decompose $Y$ as follows (note that the right hand side indeed equals $Y$ as defined above): 

\begin{equation}\label{eq:Y}
    Y = \frac{1}{2}Z - \frac{1}{2}D_Z - \frac{1}{2} Z \circ A_G
\end{equation}

In (\ref{eq:Y}), $D_Z$ denotes the diagonal matrix whose diagonal is the same as the diagonal of $Z$.

To analyse $\lambda_1(A_G + Y)$, we propose a multi-step approach, analysing the effect of each of the three components of $Y$ separately. In the description below, we use the term {\em spectrum} of a matrix (which usually means the multi-set of all eigenvalues of the matrix) in an approximate sense, basically, as a density function. The information provided by this density function is its support (the largest and smallest eigenvalues, up to low order terms), and the fraction of eigenvalues (up to $o(1)$ terms) that lie in every interval of values. In our approach, whose steps are listed below, we do not analyse only $\lambda_1(A_G + Y)$, but rather the whole spectrum of $A_G + Y$, or equivalently, of $A_G + \frac{1}{2}Z - \frac{1}{2}D_Z - \frac{1}{2} Z \circ A_G$. 

\begin{enumerate}
    \item Determine the spectrum of $A_G + \frac{1}{2}Z$. This can be done because of the {\em same eigenbasis} property. The spectrum of $A_G$ is known to follow Wigner's semi-circle law, because $G$ is a random $G_{n,1/2}$ graph. The spectrum of $\frac{1}{2}Z$ is chosen by us, and is known to be $\frac{1}{2}\alpha_1, \ldots, \frac{1}{2}\alpha_n$. As $A_G$ and $Z$ have the same set of eigenvectors, the spectrum of $A_G + \frac{1}{2}Z$ is the multi-set $\lambda_i(A_G) + \alpha_i$, for $i \in [n]$.
    \item Based on the spectrum of $A_G + \frac{1}{2}Z$, determine the spectrum of $(A_G + \frac{1}{2}Z) - \frac{1}{2}D_Z$. By the {\em concentrated diagonal} property, all diagonal entries of $Z$ are concentrated around a single value of $d$. Thus, the spectrum of $A_G + \frac{1}{2}Z - \frac{1}{2}D_Z$ is essentially the same as the spectrum of $A_G + \frac{1}{2}Z$, with each eigenvalue shifted by $-\frac{d}{2}$.
    \item To proceed, we wish to determine the spectrum of $Z \circ A_G$. The effect of the operation ``$\circ A_G$" on $Z$ is to flip the sign of roughly half of $Z$'s entries (and zero its diagonal). $A_G$ is a random matrix (because $G$ is random), and this 
    suggests that $Z \circ A_G$ should be a {\em generalized Wigner matrix}.
    However, this argument is flawed, because $A_G$ and $Z$ are not chosen independently. To circumvent this dependency, the next step of our analysis does not consider $Z \circ A_G$. Instead we draw a new random graph $G' \simeq G_{n,1/2}$, and consider $Z \circ A_{G'}$ (which is a generalized Wigner matrix) instead of $Z \circ A_G$. If $Z$ has the desired pseudo random properties, then the spectrum of $Z \circ A_{G'}$ can be determined to follow the semi-circle law (where the support of the semi-circle can be computed from the {Frobenius} norm of $Z$, which is known because the spectrum of $Z$ is known). 
    \item We need to show that the spectrum of $Z \circ A_G$ is similar to that of $Z \circ A_{G' \sim G_{n,1/2}}$. The {\em small average} and property makes this more believable, but this property by itself does not suffice. Experimental evidence strongly suggests that the spectra are indeed very similar. However, proving this rigorously seems beyond reach of current techniques. Hence, we shall assume that this is true, and leave the proof of this assumption to future work.
    \item Having determined separately the spectrum of $A_G + \frac{1}{2}Z - \frac{1}{2}D_Z$ and of $-\frac{1}{2}Z \circ A_G$, it is left to determine the spectrum of the sum of these two matrices. Here, we conjecture that this spectrum is the free convolution of the respective spectra (including the largest eigenvalue sticking to the bulk of the spectrum). Again, experimental evidence strongly suggests that this is true. Current techniques at best give us hope to prove that $A_G + \frac{1}{2}Z - \frac{1}{2}D_Z$ and  $-\frac{1}{2}Z \circ A_{G' \sim G_{n,1/2}}$ are freely independent (note the change from $G$ to $G'$). However, proving this for $-\frac{1}{2}Z \circ A_G$ seems to require the development of a deterministic version of free convolution, an open question discussed in Section~\ref{sec:discussion}.
\end{enumerate}

We alert the reader that the approach presented above does not provide at this time rigorous proofs. The main reason is that the tool of free convolution is currently developed only for random matrices, and our approach would need a deterministic version of it. Another reason is that for choices of $Z$ that are most promising (both experimentally and according to the analysis framework presented above), we currently do not know how to prove all properties that we would like $Z$ to have (most notably, pseudorandomness). However, extensive experimentation suggests that the gaps in our proofs are due to lack of proof techniques, and not due to incorrect steps in the approach.

\subsection{Our choice of $Z$, so as to bound the theta function}
\label{sec:results}

The main part of our work analyses the following choice of a matrix $Z$ for the framework described in Section~\ref{sec:approach}. For simplicity, assume that $n$ is even. $Z$ has $n/2$ eigenvalues chosen as $\alpha_i = -\lambda_i(A_G)$ for $1 \le i \le \frac{n}{2}$, and  $n/2$ eigenvalues chosen as $\alpha_i = \lambda_i(A_G)$ for $\frac{n}{2} +1 \le i \le n$. Recall that the eigenvectors of $Z$ are required to be those of $A_G$, that we denote by $v_1, \ldots, v_n$, and that $Z = \sum_i v_i \alpha_i v_i^T$. 

As the eigenvalues of $A_G$ form a semi-circle centered at~0, the eigenvalues of $Z$ are arranged first in increasing order (from roughly $-2\sqrt{n}$ to roughly~0), and then in decreasing order (from roughly~0 to roughly $-2\sqrt{n}$), and form a quarter-circle.

Recall that we derive $Y$ from $Z$  by taking only the free entries of $Z$ (those with $(i,j) \not\in E$), and that our $M \in \cal{M}_G$ that upper bounds $\bar{\vartheta}(G)$ is $M = A_G + Y$. 

Our choice of $Z$ is motivated by the following considerations. The choice of $\alpha_i = -\lambda_i(A_G)$ for $1 \le i \le \frac{n}{2}$ is meant to eliminate the positive part of the spectrum of $A_G$. (See step~1 below.) 
{On average, half of the value of a free entry in $A_G$ comes from the positive part of the spectrum (i.e $\sum_{i = 1}^{n/2}\lambda_i(A_G)v_iv_i^T$), and half from the negative part (i.e $\sum_{i = n/2 + 1}^{n}\lambda_i(A_G)v_iv_i^T$), since the complement of $G \sim G_{n,1/2}$ has the same distribution as $G \sim G_{n,1/2}$, and the same holds for the non-free entries (see \cref{zentries}).
So, if we choose $\alpha_i = 0$ for the remaining eigenvectors $\frac{n}{2} + 1 \leq i \leq n$, then the average value of a free entry in $Z$ would be~$\frac{1}{2}$. 
We want this average value to be~0, so as to satisfy the {\em small average} property. By choosing $\alpha_i = \lambda_i(A_G)$ for $\frac{n}{2} +1 \le i \le n$, the value of the average free entry becomes~0.}

\begin{figure}[h]
\includegraphics[width=16cm]{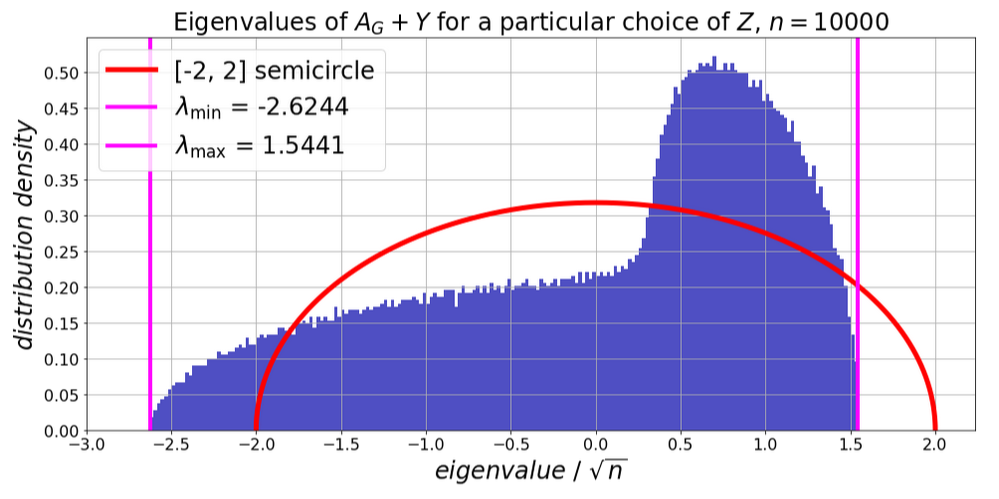}
\caption{Experimental distribution of eigenvalues of $A_G + Y$ for our choice of $Z$}
\label{fig:agz}
\end{figure}

The following is formulated as a conjecture and not as a theorem, because we do not have a full proof for it, but we do have supporting evidence that it is true.

\begin{conjecture}
\label{claim:1.54}
    For $M$ as chosen above, $\lambda_1(M) \le 1.55 \sqrt{n}$, with high probability.
\end{conjecture}

The claim, if true, implies that the expected value of the theta function for $G_{n,1/2}$ is at most $1.55 \sqrt{n}$. We remark that our approach easily reduces the upper bound on the theta function even further, by taking other choices for $Z$. As a simple example, replacing $Z$ above by an alternative $Z' = 1.3 Z$ improves the upper bound to roughly $1.5$. In this presentation we prefer to use the relatively simple $Z$ described above and do not attempt to optimize $Z$, as this makes it easier to describe our proof techniques.

For $M$ chosen as above, experimental results for matrices of sizes up to $n = 10000$ indeed show that its largest eigenvalue is roughly $1.544 \sqrt{n}$. Moreover, the proof approach presented in Section~\ref{sec:approach} strongly suggests that \cref{claim:1.54} is true. We now explain in more details the various steps of this proof approach, as applied to $Z$ chosen as above.

\begin{enumerate}
    \item The spectrum of $A_G + \frac{1}{2}Z$ is $\frac{1}{2}\lambda_1(A_G), \ldots \frac{1}{2}\lambda_{n/2}(A_G), \frac{3}{2}\lambda_{n/2 + 1}(A_G), \ldots \frac{3}{2}\lambda_n(A_G)$. Pictorially, it is composed of two quarter circles, a tall and narrow one between~0 and $\sqrt{n}$, and a shallow and wide one between~0 and $-3\sqrt{n}$. In particular, $\lambda_1(A_G + \frac{1}{2}Z) \simeq \sqrt{n}$, so in this step me significantly decrease $\lambda_1$.
    
    \item The trace of $Z$ equals the sum of $Z$'s eigenvalues and can easily be determined to be $\frac{8}{3\pi}n^{3/2}$ (up to low order terms), and so the average diagonal entry is $d = \frac{8}{3\pi}\sqrt{n}$. Moreover, all diagonal entries have the same distribution (as the distribution $G_{n,1/2}$ does not change if names of vertices are permuted), and thus each one of them has mean $d$. 
    
    Our approach for proving the {\em concentrated diagonal} property is to note that for every $k$, $Z_{kk} \sum_i \alpha_i ((v_i)_k)^2$ (where $(v_i)_k$ is the $k$th entry in eigenvector $v_i$), to show that this latter sum is strongly concentrated around its mean (which we know is equal to $d$), and then use a union bound over the $n$ possible choices for $k$. The $\alpha_i$ are known to great accuracy (from the {\em rigidity} property of eigenvalues~\cite{KY13}), and up to a negligible error term, can be thought of as constants. For every $i$ and $k$, $((v_i)_k)^2$ is a random variable of mean $\frac{1}{n}$ (because eigenvectors are unit vectors, and the symmetry of the $G_{n,1/2}$ distribution). Moreover, the delocalization {property} of eigenvalues of random graph implies  that with overwhelming probability, none of the  $((v_i)_k)^2$ is very large by itself. Under these conditions, sufficiently strong concentration result can be proved if we know that the $((v_i)_k)^2$ entries for different $i$ are independent from each other. For the eigenvectors of Gaussian matrices independence is known to hold. Hence, when $A_G$ is a random Gaussian matrix (rather than a random $\pm 1$ matrix), we can indeed prove the concentrated diagonal property for the corresponding {$Z$}. See Theorem~\ref{diagconc}. For eigenvectors of $\pm1$ matrices such as $A_G$, exact independence is not known to hold, but there are known bounds on the correlation of entries in different eigenvectors~\cite{BL22}. Using these bounds, we can prove that for every $c < 1/2$, for each of the diagonal entries of $Z$, the probability that it deviates from $d$ by more than a low order term is at most $O(n^{-c})$. See Theorem~\ref{diagconc2}. To apply the union bound over all $n$ entries of the diagonal we need $c > 1$, but unfortunately, at this point we cannot prove this. Hence, we come very close to proving the {\em concentrated diagonal} property, but do not quite prove it. 
    
    Assuming that the {\em concentrated diagonal} does hold (an assumption supported by experimental evidence, by the fact that it holds in the Gaussian case, and by the fact that we are nearly able to prove it), we get that the spectrum of $A_G + \frac{1}{2}Z - \frac{1}{2}D_Z$ is the same as that of $A_G + \frac{1}{2}Z$, but shifted by an additive term of $\frac{4}{3\pi}\sqrt{n}$. In particular, $\lambda_1(A_G + \frac{1}{2}Z - \frac{1}{2}D_Z) \simeq (1 + \frac{4}{3\pi})\sqrt{n} < 1.4245 \sqrt{n}$.
   
    \item Our next step is to determine the spectrum of $Z \circ A_{G'}$, where $G'$ is a fresh sample of a graph from $G_{n,1/2}$ (not the original random $G$). $Z \circ A_{G'}$ is a generalized Wigner matrix, and if $Z$ has certain {\em pseudorandom properties}, then its spectrum is a semi-circle whose width is determined by the average squared entry of $Z - D_z$. The average squared entry can be determined from the Frobenius norm of $Z - D_Z$, and this Frobenius norm can easily be computed as it equals the sum of squares of the eigenvalues, which are known. Following these computations, we obtain that the semi-circle is supported in the interval $[-2\alpha, 2\alpha]$ for $\alpha = \sqrt{1 - \frac{64}{9\pi^2}}$. See Corollary~\ref{cond23}.

    It remains to prove that $Z$ has the desired pseudorandomness properties (such as all rows having roughly the same $\ell_2$ norm, no single entry being exceptionally large, and others). As in the case of proving the {\em concentrated diagonal} property, one can approach this by considering the eigenvectors of $A_G$. As in the diagonal case, when $A_G$ is a random Gaussian matrix we can prove the desired pseudorandom properties for the corresponding $Z$ (see Theorem~\ref{genwigmain}), but when $A_G$ is a random $\pm 1$ matrix, the known properties of its corresponding eigenvectors do not suffice in order to complete the proof. Despite lack of a rigorous proof, it seems safe to assume that the spectrum of  $Z \circ A_{G'}$ is a semi-circle, an assumption supported by experimental evidence, and by the fact that full proofs can be given for the Gaussian case.
   
    \item We need to show that the spectrum of $Z \circ A_G$ is similar to that of $Z \circ A_{G' \sim G_{n,1/2}}$. Experimental evidence strongly suggests that the spectra are indeed very similar. However, proving this rigorously seems beyond reach of current techniques. Hence, we assume without proof that this is true.
    
    \item Having determined separately the spectrum of $A_G + \frac{1}{2}Z - \frac{1}{2}D_Z$ and of $-\frac{1}{2}Z \circ A_G$, it is left to determine the spectrum of the sum of these two matrices (though it suffices to determine only the largest eigenvalue of this sum).  We conjecture, but cannot prove, that this spectrum is the free convolution of the respective spectra. We provide two type of experimental evidence to support our conjecture. One is to compute the prediction of the free convolution for the spectrum and compare it with the the spectrum obtained for $A_G + \frac{1}{2}Z - \frac{1}{2}D_Z - \frac{1}{2}Z \circ A_G$ for random graphs with up to 10000 vertices. We get an excellent fit. The other is to verify experimentally that the eigenbasis of $A_G + \frac{1}{2}Z - \frac{1}{2}D_Z$ (or rather, the eigenbasis of $A_G$, as the $-\frac{1}{2}D_Z$ operation can be left to be performed last without changing our analysis) is uncorrelated with that of $Z \circ A_G$. In particular, if they are completely uncorrelated, one would expect the largest inner product between any eigenvector of $A_G$ and any eigenvector of $Z \circ A_G$ to be $O(\sqrt{\frac {\log n}{n}})$. Indeed, experimentally this maximum is roughly $\sqrt{\frac{2\log n}{n}}$, for graphs with 10000 vertices.
    
    For the spectra of $A_G + \frac{1}{2}Z - \frac{1}{2}D_Z$ and $-\frac{1}{2}Z \circ A_G$ as computed earlier, the free convolution predicts that the largest eigenvalue of the sum is roughly~$1.544\sqrt{n}$, in agreement with experimental results.
\end{enumerate}

\subsection{Optimizing the spectral radius}

In addition to bounding the theta function $\otht(G)$, we also consider the problem of optimizing the spectral radius $\bar{\rho(G)}$.
That is, finding $M \in \Mm_G$ with minimal $\rho(M) = \max[\lambda_1(M), -\lambda_n(M)]$.

\begin{figure}[h]
\includegraphics[width=16cm]{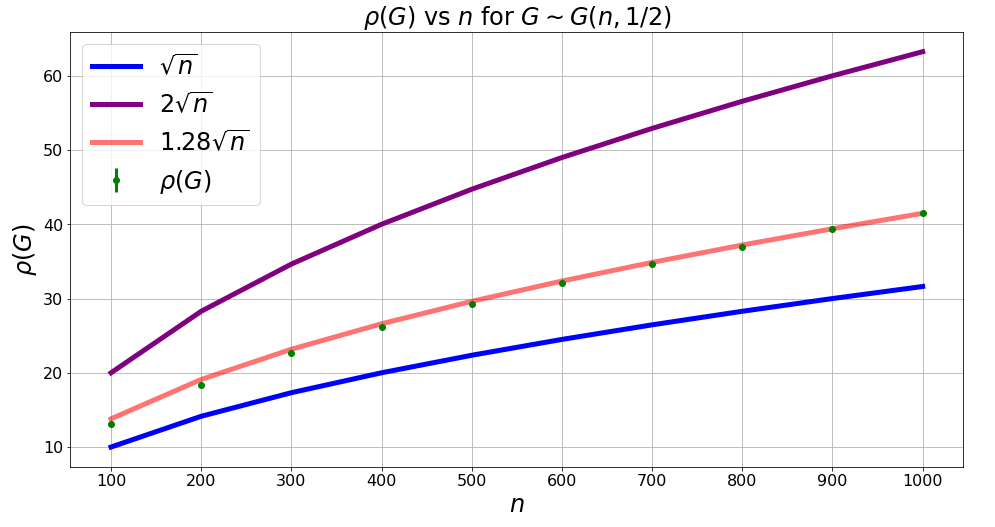}
\label{fig:sigma}
\caption{Experimental value of $\bar{\rho}(G)$ for $G\sim G(n, 1/2)$ for various $n$}
\end{figure}

{Results of experiments that we conducted suggest that for $G \sim G_{n,1/2}$,  $\bar{\rho}(G) \simeq 1.28\sqrt{n}$.} 
An approach similar to that outlined in Section~\ref{sec:approach} can also be used to potentially bound $\bar{\rho}(G)$ and to convince ourselves that indeed, $\bar{\rho}(G)$ is significantly smaller than $2\sqrt{n}$ (with high probability over choice of $G \sim G_{n,1/2}$). 
More specifically, we use the {approach} from \cref{sec:approach} with a different choice of matrix $Z$.
{More specifically, if $A_G = \sum_{k = 1}^n\lambda_ku_ku_k^T$ is the spectral decomposition of the adjacency matrix $A_G$, our choice is $Z = \sum_{k = 1}^n\alpha_k u_ku_k^T$, where for $k \leq n/2$, we set $\alpha_k = \frac{3\pi}{8}\sqrt{n} - \lambda_k$, and for $k > n/2$ we set $\alpha_k = -\frac{3\pi}{8}\sqrt{n} - \lambda_k$.}
{With this choice of $Z$, we} perform analogous approximate analysis, {which suggests} that the largest \textit{singular} value of the sum $A_G + Y(Z)$ for this new choice of $Z$ is roughly $1.75\sqrt{n}$. {Experimental result also support this conclusion.}
{See \cref{sec:specradius} for more details.}

\begin{conjecture}
\label{claim:spec}
For $M = A_G + Y(Z)$ with $Z$ as above, $\rho(M) \le 1.75 \sqrt{n}$, with high probability.
\end{conjecture}

As to showing that $\bar{\rho}(G)$ is significantly higher than $\sqrt{n}$, for this we have a rigorous proof.

\begin{theorem}
    \label{thm:radius}
    With high probability over choice of $G \sim G_{n,1/2}$ it holds that $$\bar{\rho}(G) \ge  \left(3\pi/8 - \o{1}\right)\sqrt{n} \simeq 1.178\sqrt{n}$$
\end{theorem}

The basic idea of the proof is as follows. The average entry in the free locations of $M$ ($M_{ij}$ with $(i,j)\not\in E$) needs to be roughly $-1$, so as to keep the average entry of $M$ close to~0. Otherwise, the quadratic form ${\bf 1}^T M {\bf 1}$ certifies that the spectral radius is large (here ${\bf 1}$ is the all~1 vector). The free entries cannot remain $-1$, as then the spectral radius remains $2\sqrt{n}$. Hence their average must be $-1$, but they should be distributed with some variance. This forces the Frobenius norm of $M$ to increase to $(1 + \Omega(1))n^2$, which implies a spectral radius of $(1 + \Omega(1))\sqrt{n}$. For a detailed proof, see Section~\ref{sec:radius}.

\subsection{Discussion}
\label{sec:discussion}

We view \cref{claim:1.54} {and \cref{claim:spec}},  together with the evidence presented to support them, as the main contributions of our paper. Our proof approach provides what we view as a convincing argument that the expected value of the theta function for $G_{n,1/2}$ is significantly below the ``trivial" upper bound implied by the $\pm 1$ adjacency matrix of $G$ (which stands as the only known upper for over 40 years). Though our proof currently has significant gaps, these gaps suggest new research directions that may be of interest beyond the current work. Perhaps the most interesting of them is to extend the applicability of free convolution from settings in which at least one of the two matrices is chosen at random, to settings in which both matrices are chosen deterministically, subject to some limitations on the correlations between their eigenvalues. 

{We present here two independent modifications to our approach for upper bounding $\bar{\vartheta}(G)$. Experimentally, each of them by itself appears to provide upper bounds on  $\bar{\vartheta}(G)$ that are tighter than those of \cref{claim:1.54}, but we have not attempted to analyze them.

\begin{enumerate}
    \item In Section~\ref{sec:results} we selected $M = A_G + Y$ for a certain choice of $Y$. For the same choice of $Y$, one can introduce a parameter $\tau > 0$, and optimize $\tau$ so that $M = A_G + \tau \cdot Y$ minimizes  $\lambda_1(M)$. Experimentally, it appears that for $\tau \simeq 1.3$, one gets $\lambda_1(M) \simeq 1.5$.

    \item A recursive procedure. In Section~\ref{sec:results} we derived from $A_G$ a matrix $Z$ with the same eigenvectors as those of $A_G$ and with negative eigenvalues (except possibly for some negligible positive eigenvalues), we derived from $Z$ a matrix $Y$ by taking only the free entries of $Z$ (those with $(i,j) \not\in E$), and chose $M = A_G + Y$. One can take this $M$ and derive from it a matrix $Z'= 2M^-$, where $M^-$ is the contribution of the $n/2$ smallest (most negative) eigenvalues to the spectral decomposition of $M$, and then use the corresponding $Y'$ to obtain $M' = A_G + Y'$. Experimentally, it appears that $\lambda_1(M') \simeq 1.45$. (One may then use a similar construction starting from $M'$ instead of $M$, but experimentally this does not appear to help in further reducing $\lambda_1$.) 
\end{enumerate}
}

In the process of our study, we considered graph parameters other than the theta function. We explicitly mentioned $\bar{\rho}$, which is a version of the theta function based on spectral radius instead of largest eigenvalue. Other parameters were mentioned only implicitly. Each deterministic choice of $Z$ under the framework of Section~\ref{sec:approach} (e.g., the choice of $Z$ used in Section~\ref{sec:results}) gives a graph parameter that serves as an upper bound on $\omega(G)$. Each such graph parameter comes together with an associated matrix $M$. The distribution of these $M$ over random choice of $G \in G_{n,1/2}$ is statistically distinguishable from that over  random choice of $G \in G_{n,1/2, k}$ (for $k \gg 2\log n$). In remains open whether there is any such graph parameter and associated matrix $M$, for which these two distributions are also distinguishable in polynomial time, for some value of $k \le o(\sqrt{n})$.

\subsection*{Acknowledgements}

This research was supported in part by the Israel Science Foundation (grant No. 1122/22). We thank Ofer Zeitouni for helpful discussions and for directing us to relevant literature.

\bibliographystyle{alpha}
\bibliography{main}

\appendix

\section{Analysis of the algorithm}\label{sec:analysis}

\begin{definition}\label{wigdefApp}
    Consider a family of independent zero mean real random variables $\{Z_{ij}\}_{1\leq i < j}^n$, with $\E{Z_{ij}^2} = 1$, independent from a family $\{Y_i\}_{i=1}^n$ of i.i.d. zero mean real random variables, with $\E{Y_i^2}$ finite.
    Let $\Xx_n$ be a probability distribution over symmetric $n\times n$ matrix $X_n$ with entries
    \[X_n(i, j) = X_n(j, i) = \begin{cases}
       Z_{ij},&i < j;\\
       Y_i,&i = j.
    \end{cases}\]
    We call such matrix $X_n \sim \Xx_n$ a \textbf{Wigner matrix} with laws $(Z_{ij}, Y_i)$.
\end{definition}
Examples of distributions $X\sim \Xx_n$ include Gaussian Orthogonal Ensemble $\GOE(n)$, where $X_{ij} \sim \No(0, 1)$, and Bernoulli, where $X_{ij}$ takes values $\pm 1$ with equal probability.
We are going to additionally require the distributions of the entries of $X\sim \Xx_n$ to be \textbf{symmetric}.

\begin{convention}\label{symdist}
    Let $X\sim \Xx_n$ be a Wigner matrix with laws $(Z_{ij}, Y_i)$ as in \cref{wigdefApp}.
    We are going to require that for all $i, j \in [n]$, $\P{X_{ij} > 0} = \P{X_{ij} < 0} = 1/2$.
\end{convention}

A classic result of Wigner \cite{Wig55}, \cite{Wig58} establishes the limiting distribution of the eigenvalues of Wigner matrices $X\sim \Xx_n$.
\begin{definition}
    Let $X_n \sim \Xx_n$ be a Wigner matrix, let $\lambda_1(X_n)\geq \ldots \lambda_n(X_n)$ be its eigenvalues.
    The \textbf{empirical spectral distribution measure} of the eigenvalues of $X_n/\sqrt{n}$ is a function
    \[L_{X_n}(x) := \frac{1}{n}\sum_{i = 1}^n\I{\lambda_{i}(X_n/\sqrt{n}) \leq x},\]
    denote its density by $l_{X_n}(x)$.
    In particular, for any $a \leq b$,
    \[L_{X_n}([a, b]) = \frac{1}{n}|\{i : \lambda_i(X_n/\sqrt{n}) \in [a, b]\}|,\]
    and for any continuous $f$, the integral $\int f\d L_{X_n}(x)$ (as a function of $X_n$) is the random variable 
    \[\int f\d L_{X_n}(x) = \frac{1}{n}\sum_{i = 1}^nf(\lambda_i(X_n/\sqrt{n})).\]
\end{definition}
\begin{theorem}[Wigner Semi-Circle Law, \cite{Wig55,AGZ09}]\label{semicircle}
   Define the \textbf{standard semi-circle law} on $\R$ as probability distribution $\Rho$ with density
   \[\rho(x) = \d\Rho(x) := \frac{1}{2\pi}\sqrt{4 - x^2}\I{x \in [-2, 2]}\d x.\]
   Let $X_n \sim \Xx_n$ be a Wigner matrix, let $L_{X_n}(x)$ be the empirical spectral distribution measure of $X_n/\sqrt{n}$.
   For any continuous bounded function $f$, and each $\eps > 0$,
   \[\lim_{n\to\infty}\P{\left|\int f(x)\d L_{X_n}(x) - \int f(x)\d \Rho(x)\right| > \eps} = 0.\]
   That is, $L_{X_n} \xrightarrow[n\to\infty]{\mathrm{a.s.}} \Rho$, i.e the distribution of the eigenvalues of $X_n/\sqrt{n}$ converges almost surely to a standard semi-circle law.
\end{theorem}

Before we start the analysis of our approach from \cref{sec:approach}, we prove \cref{freeentdist}, stated earlier.
That is, choosing the free entries of matrix $M \in \Mm_G$ independently from some fixed distribution $\Dd$ will not allow us to push $\lambda_1(M)$ below $2\sqrt{n}$ by any significant margin.
\begin{theorem}[\cref{freeentdist} restated]\label{freeentdistproof}
    Let $M$ be a random symmetric matrix, where $M_{ii} = 1$ for all $i \in [n]$, and each entry $M_{ij}$, $i < j \in [n]$, is sampled independently as follows: with probability $1/2$, $M_{ij} = 1$, and with probability $1/2$, $M_{ij}$ is sampled from distribution $\Dd$ with mean $\fe \in \R$, and variance $\psi^2 \geq 0$, such that $\Dd$ has finite fifth moment.
    {Then, for every constant $\eps > 0$, for $n$ large enough, $\lambda_1(M) \geq (2 - \eps)\sqrt{n}$ with high probability.} 
\end{theorem}
\begin{proof}
    First, observe that
    \[\E{\lambda_1(M)} \geq \E{\frac{\One^TM\One}{\|\One\|^2}} = \frac{1}{n}\sum_{i, j = 1}^n\E{M_{ij}} = 1 + (n^2 - n)\left(\frac{1}{2}\cdot 1 + \frac{1}{2}\cdot \fe\right) = 1 + \frac{n^2 - n}{2}\cdot (1 + \fe).\]
    Then, if $\fe > -1 + 4/\sqrt{n}$, the expectation of $\lambda_1(M)$ will already be at least $(2 - o(1))\sqrt{n} > (2 - \eps)\sqrt{n}$ for every constant $\eps > 0$.
    Hence, we assume that $\fe < -1 + 4/\sqrt{n}$. 
    
    For all $i < j \in [n]$, 
    \[\E{M_{ij}} = \frac{1}{2}\cdot 1 + \frac{1}{2}\cdot \fe = \frac{1 + \fe}{2}.\]
    {Let $B$ be a matrix such that $M = B + \frac{1 + \fe}{2}J + \frac{1 - \fe}{2}I$,} where $J$ is an all-one matrix of size $n$.
    It is easy to see that {$B_{ii} = 0$ for all $i$, and that $\E{B_{ij}} = 0$ for all $i < j$.}
    In addition, the variances of $B_{ij}$, $i < j$, are {the same as the variances of $M_{ij}$, as the entries of $B$ are a linear shift of $M$'s.
    So,
    \begin{multline*}
        \V{B_{ij}} = \V{M_{ij}} = \E{M_{ij}^2} - \E{M_{ij}}^2 \\
        = \frac{1}{2}\cdot 1 + \frac{1}{2}\cdot \left(\V{\Dd} + \E{\Dd}^2\right) - \frac{(1 + \fe)^2}{4} 
        = \frac{1}{2} + \frac{\psi^2}{2} + \frac{\fe^2}{2} - \frac{(1 + \fe)^2}{4} \\
        = \frac{\psi^2}{2} + \frac{1}{4} + \frac{\fe^2}{4} - \frac{\fe}{2} = \frac{\psi^2}{2} + \frac{(1 - \fe)^2}{4}.
    \end{multline*}
    }
    Clearly, $B$ is a Wigner matrix.
    By \cref{semicircle}, the spectrum of $\frac{1}{\sqrt{n}}B$ has a semicircle distribution on $[-2\alpha, 2\alpha]$ where $\alpha = \sqrt{\frac{\psi^2}{2} + \frac{(1 - \fe)^2}{4}}$.
    {By Weyl's theorem (see, for example, \cite{HJ12}), $\lambda_1(M) = \lambda_1(B + \frac{1 + \fe}{2}J + \frac{1 - \fe}{2}I) \geq \lambda_2(B) + \lambda_{n - 1}(\frac{1 + \fe}{2}J) + \frac{1 - \fe}{2} = \lambda_2(B) + \frac{1 - \fe}{2}$. 
    Since the support of the distribution of the eigenvalues of $\frac{1}{\sqrt{n}}B$ is the segment $[-2\alpha, 2\alpha]$, with high probability for every constant $\eps > 0$ there will be a lot of eigenvalues of $\frac{1}{\sqrt{n}}B$ that are at least $2\alpha - \eps$.
    In particular, $\lambda_2(\frac{1}{\sqrt{n}}B) \geq 2\alpha - \eps$.}
    Now, for $2\alpha$ to be smaller than $2$, while $\fe < -1 + 4/\sqrt{n}$ and $\psi^2 \geq 0$, one needs $\fe > - 1$.
    {For $\fe$ in this range, $\frac{1-\fe}{2} > 1-\frac{2}{\sqrt{n}}$.}
    Furthermore, as $n$ grows, the upper bound $\fe < -1 + 4/\sqrt{n}$ becomes arbitrarily close to $-1$, {hence for $n$ large enough we would have $\lambda_2(\frac{1}{\sqrt{n}}B) \geq 2 - \eps$.
    It follows that with high probability $\lambda_1(M) \geq \lambda_2(B) + \frac{1-\fe}{2} \geq (2 - \eps)\sqrt{n}$.}
\end{proof}

Recall that in \cref{sec:approach} we consider matrix sum $A_G + (1/2)Z - (1/2)D_Z - (1/2)Z\circ A_G$ for a particular choice of matrix $Z$.
Since $A_G$ is $\pm1$-adjacency matrix of $G\sim G(n, 1/2)$, $A_G \sim \Bern_{\pm1}(n, 1/2)$, where $\Bern_{\pm1}(n, 1/2)$ is a distribution over symmetric matrices of size $n$ with diagonal $0$ and off-diagonal entries taking values $\{\pm 1\}$ with probability $1/2$.
So, $A_G$ is a Wigner matrix according to definition \cref{wigdefApp} for $\Xx_n = \Bern_{\pm1}(n, 1/2)$.
From now on, we are going to denote $X := A_G$ for convenience.

Let $X = U\Lambda U^T = \sum_{i = 1}^n\lambda_iu_iu_i^T$ be the eigendecomposition of $X$, where $U = [u_1,\ldots, u_n]$ is an orthonormal eigenbasis of $X$ and $\lambda_1 \geq \ldots \lambda_n$ are its eigenvalues.
Consider the following two matrices, $X^+$ and $X^-$, defined as
\[X^+ = \sum_{i = 1}^{n/2}\lambda_iu_iu_i^T,\qquad\qquad X^- = \sum_{i = n/2 + 1}^n\lambda_iu_iu_i^T.\]
Matrices $X^+$ and $X^-$ are interesting for us, due to the fact that with high probability, almost all eigenvalues of $X^+$ are positive, and almost all eigenvalues of $X^-$ are negative.
This property is called \textbf{rigidity} of eigenvalues of Wigner matrices (\cite{KY13}) and follows from the local semicircle law (\cite{ESY09}).
More specifically, the following holds for $X \sim \Xx_n$.
\begin{theorem}[Rigidity of eigenvalues, \cite{KY13}]\label{rigeig0}
    For $i \in [n]$, denote by $\gamma_i$ the \textbf{classical location} of the $i$-th eigenvalue under the semicircle law.
    Specifically, $\Rho(\gamma_i) = \int_{-\infty}^{\gamma_i}\rho(x)\d x = \frac{i}{n}$.
    Let $X\sim \Xx_n$, let $\lambda_1 \geq \ldots \geq \lambda_n$ denote the eigenvalues of $X/\sqrt{n}$.
    There exist positive constants $A > 1, c, C$ and $\tau < 1$ such that
    \[\P{\exists i : |\lambda_i - \gamma_i| \geq \frac{(\log n)^{A\log\log n}}{n^{2/3}}} \leq \exp\left(-c\left(\log n\right)^{\tau A\log\log n}\right).\]
\end{theorem}

We are going to pick matrix $Z := X^- - X^+$ as our choice in \cref{sec:approach}.
Clearly, $Z = X^- - X^+$ has the same eigenbasis as $X$.
On top of that, the average entries in non-free and free locations of matrices $X^-$ and $X^+$ are the same with high probability, and as a result the average entry in these locations respectively of the matrix $Z = X^- - X^+$ is going to be close to $0$.
To see this, we are going to prove the following claim.
\begin{theorem}\label{zentries}
    Let $G \sim G(n, 1/2)$, $G = (V, E)$, be a random graph, let $X = \sum_{k = 1}^n\lambda_ku_ku_k^T$ be the spectral decomposition of its $\pm 1$ adjacency matrix $X$.
    Let $X^+ = \sum_{k = 1}^{n/2}\lambda_ku_ku_k^T$, $X^- = \sum_{k = n/2 + 1}^n\lambda_ku_ku_k^T$, $Z = X^- - X^+$.
    For any $\eps \in (0, 1/4)$, with probability at least $1 - O(n^{-1}\polylog n)$,
    \begin{gather*}
        \frac{1}{|E|}\sum_{(i, j) \in E}X^+_{ij} = \frac{1}{2} \pm O(n^{2\eps - 1/2}) \quad \text{and} \quad \frac{1}{|E|}\sum_{(i, j) \in E}X^-_{ij} = \frac{1}{2} \pm O(n^{2\eps - 1/2});\\
        \frac{1}{|\oE|}\sum_{(i, j) \notin E}X^+_{ij} = -\frac{1}{2} \pm O(n^{2\eps - 1/2}) \quad \text{and}\quad \frac{1}{|\oE|}\sum_{(i, j) \notin E}X^-_{ij}  =  -\frac{1}{2} \pm O(n^{2\eps - 1/2}).
    \end{gather*}
    {As a corollary}, $\left|\frac{1}{|E|}\sum_{(i, j) \in E}Z_{ij}\right| 
        \leq O(n^{2\eps - 1/2})$ and $\left|\frac{1}{|\oE|}\sum_{(i, j) \notin E}Z_{ij}\right| 
        \leq O(n^{2\eps - 1/2})$.
\end{theorem}

\begin{proof}
    For a fixed $k \in [n]$, we consider the contribution of the eigenvector $u_k$ to the sums $\sum_{(i, j) \in E}X^+_{ij}$ and $\sum_{(i, j) \notin E}X^+_{ij}$ (and similar sums with $X^-$).
    Define
    \[f(k) := \sum_{(i, j)\in E}u_{ki}u_{kj};\qquad g(k) := \sum_{(i, j)\notin E}u_{ki}u_{kj}.\]
    From the definition of $X^+$, we have
    \[\sum_{(i, j) \in E}X^+_{ij} = \sum_{(i, j) \in E}\sum_{k = 1}^{n/2}\lambda_ku_{ki}u_{kj} = \sum_{k = 1}^{n/2}\lambda_k\sum_{(i, j)\in E}u_{ki}u_{kj} 
        = \sum_{k = 1}^{n/2}\lambda_kf(k),\]
    and $\sum_{(i, j) \notin E}X^+_{ij} = \sum_{k = 1}^{n/2}\lambda_k g(k)$, $\sum_{(i, j) \in E}X^-_{ij} = \sum_{k = n/2 + 1}^{n}\lambda_kf(k)$, $\sum_{(i, j) \notin E}X^-_{ij} = \sum_{k = n/2 + 1}^{n}\lambda_k g(k)$.
    Next, we establish that for every $k \in [n]$ we have that 1) $f_k \simeq -g(k)$ and 2) the sum of $f(k)$ and $-g(k)$ is equal to $\lambda_k$.
Precisely:
    \begin{lemma}\label{fgrel}
        For $k \in [n]$, let $f(k), g(k)$ be as above.
        For every $k \in [n]$, 1) $f(k) - g(k) = \lambda_k$; and 2) for any constant $\eps \in (0, 1/2)$, it holds that $f(k) - (-g(k)) = f(k) + g(k) \in [-n^{\eps}, n^{\eps}]$ with probability at least $1 - O(n^{-10})$ (we are going to write $f(k) + g(k) = \pm n^{\eps}$).
    \end{lemma}
    \begin{proof}
        To see 1), recall that the matrix $X = \sum_{k = 1}^n\lambda_ku_ku_k^T$ is the $\pm 1$ adjacency matrix of the random graph $G = (V, E)$.
        In particular, for every $(i, j) \in E$, $X_{ij} = 1$, and for every $(i, j) \notin E$, $X_{ij} = -1$.
        It immediately follows that
        \begin{multline*}
            f(k) - g(k) = \sum_{(i, j)\in E}u_{ki}u_{kj}-\sum_{(i, j)\notin E}u_{ki}u_{kj} = \sum_{(i, j)\in E}X_{ij}u_{ki}u_{kj} + \sum_{(i, j)\notin E}X_{ij}u_{ki}u_{kj}\\
            = \sum_{i, j = 1}^nX_{ij}u_{ki}u_{kj} = u_k^TXu_k = \lambda_k.
        \end{multline*}
        For 2), consider the sum $f(k) + g(k)$.
        It is easy to see that
        \[f(k) + g(k) = \sum_{(i, j) \in E}u_{ki}u_{kj} + \sum_{(i, j)\notin E}u_{ki}u_{kj} \\
            = \sum_{i, j = 1}^nu_{ki}u_{kj} = \left(\sum_{i = 1}^nu_{ki}\right)^2 = \langle \One, u_k\rangle^2.\]
         To bound $\langle \One, u_k\rangle^2$, we are going to rely on the following theorem by \cite{BEKYY14}.
        \begin{theorem}[\cite{BEKYY14}]\label{vecbound}
        Let $X$ be an $n\times n$ Wigner matrix, where for all $i, j \in [n]$, $\E{X_{ij}^2}$ and $\E{X_{ii}^2}$ are finite.
        Let $u_1,\ldots, u_n$ be the eigenvectors of $X$.
        For any constants $c > 0$, $0 < \eps < 1/2$, there exists constant $C = C(c, \eps)$ (also depending on the variances of $X_{ij}$ and $X_{ii}$), such that
        \[\sup_{k \in [n]}|\langle s, u_k\rangle| \leq n^{\eps - 1/2},\]
        for any fixed unit vector $s \in S^{n - 1}$, with probability at least $1 - Cn^{-c}$.
    \end{theorem}
        Given $\eps \in (0, 1/2)$, we apply \cref{vecbound} with $c = 10$ and $\eps/2$.
        Then, for any $k \in [n]$ we have
        $f(k) + g(k) = \langle \One, u_k\rangle^2 \leq (n^{\eps/2})^2 = n^\eps$ with probability at least $1 - O(n^{-10})$. 
    \end{proof}
    From \cref{fgrel} above, for all $k \in [n]$ with probability at least $1 - O(n^{-9})$ we get
    \[\lambda_k = f(k) - g(k) =\begin{cases}
        2f(k) - \langle \One, v_k\rangle^2 = 2f(k) \pm n^{2\eps};\\
        \langle \One, v_k\rangle^2 - 2g(k) = -2g(k) \pm n^{2\eps}.
    \end{cases}\]
    Therefore,
    \[\sum_{(i, j) \in E}X^+_{ij} = 
        \sum_{k = 1}^{n/2}\lambda_kf(k)=\sum_{k = 1}^{n/2}\lambda_k \cdot \frac{\lambda_k \pm n^{\eps}}{2} = \frac{1}{2}\sum_{k = 1}^{n/2}\lambda_k ^2 \pm  \frac{n^{\eps}}{2}\sum_{k = 1}^{n/2}\lambda_k.\]
    Similarly,
    \[\sum_{(i, j) \notin E}X^+_{ij}  = \sum_{k = 1}^{n/2}\lambda_k g(k) =  \sum_{k = 1}^{n/2}\lambda_k\cdot \frac{-\lambda_k \pm n^{\eps}}{2} = -\frac{1}{2}\sum_{k = 1}^{n/2}\lambda_k^2 \pm \frac{n^{\eps} }{2}\sum_{k = 1}^{n/2}\lambda_k, \]
    and 
    \[
        \sum_{(i, j) \in E}X^-_{ij} = \frac{1}{2}\sum_{k = n/2 + 1}^{n}\lambda_k^2 \pm  \frac{n^{\eps}}{2}\sum_{k =n/2 + 1}^{n}\lambda_k, \qquad 
        \sum_{(i, j) \notin E}X^-_{ij} = -\frac{1}{2}\sum_{k = n/2 + 1}^{n}\lambda_k^2 \pm \frac{n^{\eps} }{2}\sum_{k = n/2 + 1}^{n}\lambda_k.
    \]
    In order to determine the values of the sums above, we are going to use the fact that $\sum_k\lambda_k$ and $\sum_k\lambda_k^2$ are concentrated.
    The concentration bounds of \cref{PlusEigBound} below regarding eigenvalues of Wigner matrices are well-known, but we still provide a proof for completeness.
    \begin{lemma}\label{PlusEigBound}
        There exist constants $C, \alpha > 0$ such that with probability at least $1 - Cn^{-1}\log^\alpha n$,
        \[\sum_{k : \lambda_k \geq 0}\lambda_k = \frac{4}{3\pi}n^{3/2} \pm O(n^{1/2}\polylog n)\quad \text{and}\quad \sum_{k : \lambda_k \leq 0}\lambda_k = -\frac{4}{3\pi}n^{3/2} \pm O(n^{1/2}\polylog n),\]
        and
        \[\sum_{k : \lambda_k \geq 0}\lambda_k^2 = \frac{1}{2}n^{2} \pm O(n\polylog n)\quad\text{and}\quad \sum_{k : \lambda_k \leq 0}\lambda_k^2 = \frac{1}{2}n^{2} \pm O(n\polylog n).\]
        Combined with \cref{rigeig0}, we get that with probability at least $1 - O(n^{-1}\polylog n)$,
        \[\sum_{k = 1}^{n/2} \lambda_k = \frac{4}{3\pi}n^{3/2} \pm O(n^{1/2}\polylog n)\quad\text{and}\quad \sum_{k = n/2 + 1}^{n}\lambda_k = -\frac{4}{3\pi}n^{3/2} \pm O(n^{1/2}\polylog n)\]
        and
        \[\sum_{k = 1}^{n/2} \lambda_k^2 = \frac{1}{2}n^{2} \pm O(n\polylog n)\quad\text{and}\quad \sum_{k = n/2 + 1}^{n}\lambda_k^2 = \frac{1}{2}n^{2} \pm O(n\polylog n).\]
    \end{lemma}
    \begin{proof}
        We are going to use the following result on the convergence of the Wigner semicircle law.
        \begin{theorem}[Convergence Rate in Wigner Semi-Circle Law,\cite{GT13}]\label{ExpWigRate}
        Let $X_n$ be a Wigner matrix with laws $(Z_{ij}, Y_i)$, and let $A, \kappa > 0$ be such that for any $i, j \in [n]$, for any $t \geq 1$, $\P{|X_{ij}| > t}\leq A\exp(-t^\kappa)$.
        Let $\beta_n := \left(\log n\, (\log\log n)^\alpha\right)^{1 + 1/\kappa}$, let $L_n(x) := \frac{1}{n}\sum_{k = 1}^n\I{\lambda_{k}(X_n) \leq x}$ denote the empirical spectral distribution of $X_n$, and let $\Rho$ denote the standard semicircle law. 
        Then for any $\alpha > 0$ there exist constants $c, C > 0$ depending on $A, \kappa, \alpha$ only such that
        \[\P{\sup_x|L_n(x) - \Rho(x)| > \frac{\beta_n^4\log n}{n}} \leq C\exp\left(-c\log n\, (\log\log n)^\alpha\right).\]
        \end{theorem}
    
        Consider a particular function $f(x) = x$ on $[-2, 2]$, then according to the definition above for Wigner matrix $X_n := n^{-1/2}M$, we have
        $\frac{1}{n}\sum_{0 \leq\lambda_i \leq 2}\lambda_i(X_n) = \int_{0}^2x\d L_n(x)$.
        By \cref{ExpWigRate}, there exist constants $C, \alpha > 0$ such that with probability least $1 - Cn^{-1}\log^\alpha n$,
        \[\sup_{0 \leq x \leq 2} \left|\int_{0}^x\d L_n(x) - \int_{0}^x\d \Rho(x)\right|=\sup_{0 \leq x \leq 2}|L_n(x) - \Rho(x)| \leq \frac{\beta_n^4\log n}{n},\]
        then, since $f(x) = x$ is non-negative, monotone non-decreasing function on $[0, 2]$ we get
        \[\int_{0}^2x\d L_n(x) \leq \int_{0}^2x\d \Rho(x) + \frac{4\beta_n^4\log n}{n} = \frac{1}{2\pi}\int_{0}^2x\sqrt{4 - x^2}\d x +  \frac{4\beta_n^4\log n}{n} = \frac{4}{3\pi} +  \frac{4\beta_n^4\log n}{n},\]
        and $\int_{0}^2x\d L_n(x) \geq \int_{0}^2x\d \Rho(x) - \frac{4\beta_n^4\log n}{n} = \frac{4}{3\pi} -  \frac{4\beta_n^4\log n}{n}$.
        Since our matrix $X$ is essentially $X_n\sqrt{n}$, we conclude that with probability least $1 - Cn^{-1}\log^\alpha n$,
        \[\sum_{k : \lambda_k \geq 0}\lambda_k= n\cdot\left(\sum_{k : \lambda_k \geq 0}\lambda_k(X_n\sqrt{n})\right) = \frac{4}{3\pi}n^{3/2} \pm O(n^{1/2}\beta_n^4\log n) = \frac{4}{3\pi}n^{3/2} \pm O(n^{1/2}\polylog n).\]
        The proof for $\sum_{k : \lambda_k \leq 0}\lambda_k$ is completely analogous.

        Similarly, in order to show that $\sum_{k:\lambda_k \geq 0}\lambda_k^2 = \frac{1}{2}n^2 \pm O(n\polylog n)$ with high probability, we consider a function $f(x) = x^2$ on $[-2, 2]$, which is non-negative, monotone and non-decreasing on $[0, 2]$.
        Combining this with the fact that $\int_0^2x^2\d\Rho(x) = \frac{1}{2\pi}\int_0^2x^2\sqrt{4 - x^2}\d x = \frac{1}{2}$, we obtain the desired bound.
        The proof for $\sum_{k : \lambda_k \leq 0}\lambda_k^2$ is completely analogous.
    \end{proof}
    Applying \cref{PlusEigBound}, we get that with probability at least $1 - O(n^{-1}\polylog n)$
    \[\sum_{(i, j) \in E}X^+_{ij} = \frac{1}{2}\sum_{k = 1}^{n/2}\lambda_k ^2 \pm  \frac{n^{\eps}}{2}\sum_{k = 1}^{n/2}\lambda_k = \frac{n^2}{4} \pm \frac{n^{\eps}}{2}\cdot O(n^{3/2}).\]
    Now, since $G \sim G(n, 1/2)$, $\E{|E|} = \frac{n(n - 1)}{2}$ (the diagonal of adjacency matrix $X$ is $0$).
    And, by Chernoff, $\P{|E| - \frac{n^2 - n}{2} > cn\log n} = \P{|E| - \E{|E|} > cn\log n} \leq O(n^{-c})$ for a constant $c > 0$.
    Therefore, with probability at least $1- O(n^{-1}\polylog n)$, for any $\eps \in (0, 1/2)$,
    \[\frac{1}{|E|}\sum_{(i, j) \in E}X^+_{ij} = \frac{n^2}{4|E|} \pm \frac{n^{\eps}}{2|E|}\cdot O(n^{3/2}) = \frac{1}{2} \pm O(n^{{\eps} - 1/2}).\]
    Since $|\oE| = \frac{n^2 - n}{2}$ with probability at least $1 - O(n^{-c})$ as well, applying \cref{PlusEigBound} to the remaining sums we get that with probability at least $1- O(n^{-1}\polylog n)$, for any $\eps \in (0, 1/2)$,
    $
    \frac{1}{|\oE|}\sum_{(i, j) \notin E}X^+_{ij} = -\frac{1}{2} \pm O(n^{{2\eps} - 1/2})$, $\frac{1}{|E|}\sum_{(i, j) \in E}X^-_{ij} = \frac{1}{2} \pm O(n^{{\eps} - 1/2})$, $\frac{1}{|\oE|}\sum_{(i, j) \notin E}X^-_{ij} = -\frac{1}{2} \pm O(n^{{\eps} - 1/2})$.
    As a corollary, with probability at least $1 - O(n^{-1}\polylog n)$, for any $\eps \in (0, 1/2)$,
    \[\left|\frac{1}{|E|}\sum_{(i, j) \in E}^nZ_{ij}\right| = \left|\frac{1}{|E|}\sum_{(i, j) \in E}X^-_{ij} - \frac{1}{|E|}\sum_{(i, j) \in E}X^+_{ij}\right| \leq O(n^{\eps - 1/2}),\]
    and similarly $\left|\frac{1}{|\oE|}\sum_{(i, j) \notin E}^nZ_{ij}\right| \leq O(n^{\eps - 1/2})$.
    
    Note that the proof still works if we replace the $\pm 1$-Wigner matrix $X$ with a Gaussian one, i.e $X \sim \GOE(n)$.
    The only difference is that the average value in the non-free/free locations instead of $\pm\frac{1}{2}$ would be $\pm\frac{1}{2}\sqrt{\frac{2}{\pi}}$, the value of $\E{|\xi|}$ for $\xi \sim \No(0, 1)$.
\end{proof}

In addition to \cref{zentries}, one can potentially hope that $Z = X^- - X^+$ has concentrated diagonal and satisfies certain pseudorandom properties (which we will cover later).
We observe the following by directly considering each entry $X_{ij}$, and using the fact that $A_G = X = X^+ + X^-$.
\begin{observation}\label{represent}
    Denote $\hX := X^- - X^+$, let $\hD:= \Diag(\hX)$, i.e a matrix corresponding to the diagonal of $X^- - X^+$.
    Define symmetric matrix $\wX$ as:
    \[\forall i \neq j \in [n],\qquad \wX_{ii} = 0; \qquad \text{and}\qquad \wX_{ij} = \begin{cases}
        X^+_{ij} - X^-_{ij} = -\hX_{ij},& X_{ij} \geq 0;\\
        X^-_{ij} - X^+_{ij} = \hX_{ij},& X_{ij} < 0.
    \end{cases}\]
    Then $A_G + (1/2)Z - (1/2)D_Z - (1/2)Z\circ A_G = (3/2)X^- + (1/2)X^+ -(1/2)\hD + (1/2)\wX$.
\end{observation}
In order to understand why taking $Z = A_G^- - A_G^+ = X^- - X^+ = \hX$ manages (according to the experimental results, at least) to get the largest eigenvalue below $2\sqrt{n}$, we need to understand the properties of the spectrums of the summands $(3/2)X^- + (1/2)X^+$, $\hD$ and $\wX$.

\subsection{Spectrum of $X + (1/2)\hX = (3/2)X^- + (1/2)X^+$}

Recall that $X^+ = \sum_{i = 1}^{n/2}\lambda_iu_iu_i^T$ and $X^- = \sum_{i = n/2 + 1}^n\lambda_iu_iu_i^T$.
Then, {$\hX = X^- - X^+ = \sum_{i = n/2 + 1}^n\lambda_iu_iu_i^T-\sum_{i = 1}^{n/2}\lambda_iu_iu_i^T$ and $(3/2)X^- + (1/2)X^+ = \frac{3}{2}\sum_{i = n/2 + 1}^n\lambda_iu_iu_i^T+ \frac{1}{2}\sum_{i = 1}^{n/2}\lambda_iu_iu_i^T$, and eigenvalues of $X + (1/2)\hX$ are exactly $\lambda_1/2,\ldots, \lambda_{n/2}/2, 3\lambda_{n/2 + 1}/2, \ldots, 3\lambda_n/2$.}
It seems natural to assume that, since $\lambda_1,\ldots, \lambda_n$ are distributed according to the standard semi-circle, the eigenvalues of $\hX$ will be distributed as two rescaled quartercircles ``merged'' together, as with high probability, all eigenvalues of $(3/2)X^- + (1/2)X^+$ will be close to their classical locations.
\begin{figure}[h]
\includegraphics[width=16cm]{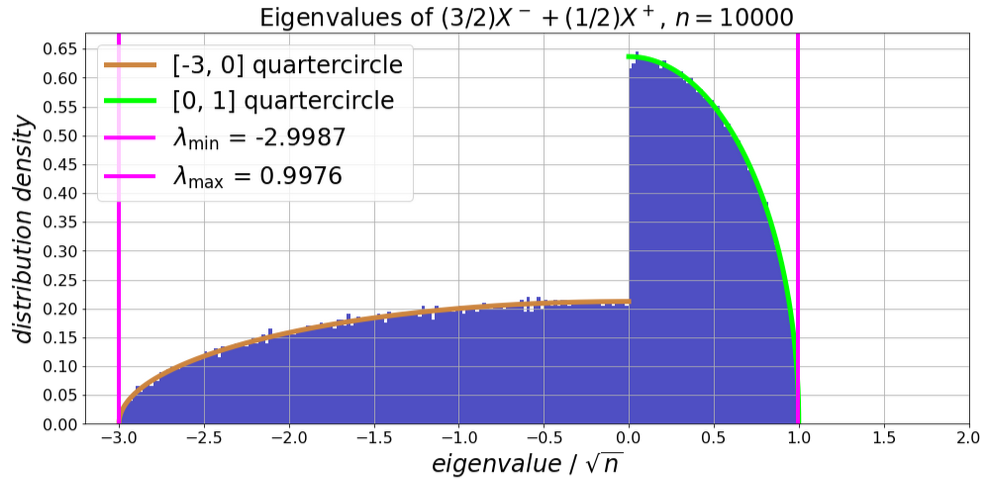}
\label{fig:qc}
\caption{Experimental distribution of eigenvalues of $(3/2)X^- + (1/2)X^+$ for $X \sim \Bern_{\pm1}(n, 1/2)$}
\end{figure}

This property does indeed hold, and is captured in the following theorem.
\begin{theorem}\label{hxdist}
    Let $\alpha, \beta > 0$, let $\Rho_{\alpha, \beta}$ be a probability distribution on $\R$ with density
    \[\rho_{\alpha, \beta}(x) = \d\Rho_{\alpha, \beta}(x):= \frac{1}{2\alpha^2\pi}\sqrt{4\alpha^2 - x^2}\I{x \in [-2\alpha, 0]}\d x + \frac{1}{2\beta^2\pi}\sqrt{4\beta^2 - x^2}\I{x \in [0, 2\beta]}\d x.\]
   For $X \sim \Xx_n$, let $L_{\frac{3}{2}X^- + \frac{1}{2}X^+}(x)$ be the empirical spectral distribution measure of $(\frac{3}{2}X^- + \frac{1}{2}X^+)/\sqrt{n}$.
   For any continuous bounded function $f$, and each $\eps > 0$,
   \[\lim_{n\to\infty}\P{\left|\int f(x)\d L_{\frac{3}{2}X^- + \frac{1}{2}X^+}(x) - \int f(x)\d \Rho_{\frac{3}{2}, \frac{1}{2}}(x)\right| > \eps} = 0.\]
   That is, $L_{\frac{3}{2}X^- + \frac{1}{2}X^+} \xrightarrow[n\to\infty]{\mathrm{a.s.}} \Rho_{\frac{3}{2}, \frac{1}{2}}$, i.e the distribution of the eigenvalues of $(\frac{3}{2}X^- + \frac{1}{2}X^+)/\sqrt{n}$ converges almost surely to $\Rho_{\frac{3}{2}, \frac{1}{2}}$.
\end{theorem}

\begin{proof}
    Let $L_X$ denote the empirical spectral distribution measure of matrix $X/\sqrt{n}$.
    By \cref{semicircle}, since $X$ is a Wigner matrix (as an adjacency matrix of a random graph $G\sim G(n, 1/2)$), the measure $L_X$ converges almost surely to a standard semicircle law $\Rho$.
    Therefore, it holds for any continuous bounded function $f$, and each $\eps > 0$, that
    \[\lim_{n\to\infty}\P{\left|\int f(x)\d L_{X}(x) - \int f(x)\d \Rho(x)\right| > \eps} = 0.\]
    
    Let $\lambda_1 \geq \ldots \geq \lambda_n$ be the eigenvalues of $X/\sqrt{n}$, and let $X = \sum_{i = 1}^n\lambda_iu_iu_i^T$ be the spectral decomposition of $X/\sqrt{n}$.
    Let $\gamma_1, \ldots, \gamma_n$ be the classical locations of eigenvalues w.r.t the semicircle distribution, so for every $i \in [n]$, $\Rho(\gamma_i) = \int_{-\infty}^{\gamma_i}\rho(x)\d x = i/n$.

    Take any continuous bounded function $f$, and let $L_{\frac{1}{2}X^+}$ be the empirical spectral measure of the matrix $\frac{1}{2}X^+/\sqrt{n} = \frac{1}{2}\sum_{i = 1}^{n/2}\lambda_iu_iu_i^T$.
    Observe that
    \[\int f(x)\d L_{\frac{1}{2}X^+} = \frac{1}{n}\sum_{i = 1}^nf\left(\lambda_i\left(\frac{1}{2}X^+/\sqrt{n}\right)\right) = \frac{1}{n}\sum_{i = 1}^{n/2}f\left(\frac{1}{2}\lambda_i\right).\]
    By \cref{rigeig0}, with probability at least $1 - O(n^{-1}\polylog n)$, for all $i \in [n]$ it holds that $|\lambda_i - \gamma_i|\leq \delta_{A, n}$, where $\delta_{A, n} := n^{-2/3}(\log n)^{A\log \log n}$ for a constant $A > 1$.
    At the same time, it holds for classical locations that $\gamma_i \in [-2, 2]$ for all $i \in [n]$, $\gamma_i > 0$ for $i \leq n/2$ and $\gamma_i < 0$ for $i > n/2$.
   {Let $Y := \sum_{i = 1}^{n}\gamma_iu_iu_i^T$ and $Y^+ := \sum_{i = 1}^{n/2}\gamma_iu_iu_i^T$.
    Then, if we replace $\lambda_i$-s with corresponding $\gamma_i$-s in the expression above, $\int f(x)\d L_{\frac{1}{2}Y^+} = \frac{1}{n}\sum_{i = 1}^{n/2}f\left(\frac{1}{2}\gamma_i\right) = \frac{1}{n}\sum_{i = 1}^{n}f\left(\frac{1}{2}\gamma_i\right)\I{\gamma_i \geq 0} = \frac{1}{n}\sum_{i = 1}^{n}f_+(\gamma_i)$, where $f_+(x) := f(x/2)\I{x > 0}$. 
    Now, $\frac{1}{n}\sum_{i = 1}^{n}f_+(\gamma_i) = \int f_+(x)\d L_{Y}$, and simultaneously, since $\gamma_i$-s are classical locations, $\frac{1}{n}\sum_{i = 1}^{n}f_+(\gamma_i) = \int f_+(x)\d \Rho(x)$.
    Observe that, if the expressions $\frac{1}{n}\sum_{i = 1}^{n/2}f\left(\frac{1}{2}\lambda_i\right)$ and $\frac{1}{n}\sum_{i = 1}^{n}f_+(\lambda_i)$ are close, then similar to the above we get $\int f(x)\d L_{\frac{1}{2}X^+} \simeq \int f_+(x)\d L_X(x)$.
    But then by \cref{semicircle}, $\int f_+(x)\d L_X(x)$ is close to $\int f_+(x) \d \Rho(x) = \int f(x)\d \Rho_{\frac{3}{2}, \frac{1}{2}}(x)\I{x \geq 0}$, which gives us the desired convergence.
    Intuitively this must be the case, since $\gamma_i$-s and $\lambda_i$-s are close and $f$ is bounded and continuous, so the chain of equalities above holds for $\lambda_i$-s too, up to a negligible error term. 
    We present a relatively straightforward rigorous proof of this fact below, for the sake of completeness.}
    
    {Observe that $\int f(x)\d L_{\frac{1}{2}X^+}  = \frac{1}{n}\sum_{i = 1}^{n/2}f\left(\frac{1}{2}\lambda_i\right)\I{\gamma_i > 0} + \frac{1}{n}\sum_{i = n/2 + 1}^{n}f\left(\frac{1}{2}\lambda_i\right)\I{\gamma_i > 0}$ (multiplier $\I{\gamma_i > 0}$ is equal to $1$ and $0$ for the first and second summands respectively).
    Consider the segment $S_0 = [-\delta_{A, n}, \delta_{A, n}]$, it is clear that if $\lambda_i \notin S_0$, then by \cref{rigeig0} with probability at least $1 - O(n^{-1}\polylog n)$ the signs of $\lambda_i$ and corresponding $\gamma_i$ for $i \in [n]$ are the same, so $f\left(\frac{1}{2}\lambda_i\right)\I{\gamma_i > 0} = f\left(\frac{1}{2}\lambda_i\right)\I{\lambda_i > 0}$ when $\lambda_i \notin S_0$.
    Hence, with probability at least $1 - O(n^{-1}\polylog n)$ it holds
    \[\int f(x)\d L_{\frac{1}{2}X^+} = \frac{1}{n}\sum_{i = 1}^{n/2}f\left(\frac{1}{2}\lambda_i\right)  = \frac{1}{n}\sum_{i = 1}^{n}f\left(\frac{1}{2}\lambda_i\right)\I{\lambda_i \geq 0} \pm \frac{1}{n}\sum_{\lambda_i \in S_0}\left|f\left(\frac{1}{2}\lambda_i\right)\right|.\]
    As was shown in \cite{Tao11} and \cite{ESY10}, the number of indices $i$ with $\lambda_i \in S_0$ is small with high probability.
    \begin{theorem}[\cite{Tao11}, \cite{ESY10}]\label{levrep}
        Let $X \sim \Xx_n$, let $\lambda_1\geq \ldots \lambda_n$ be the eigenvalues of $X$.
        There exist constants $C, c, \alpha$ such that for any $0 < \eps \leq \alpha$ there exists constant $K_\eps$ such that
        \[\P{\left|\frac{|\{i : - \eta/2 \leq \lambda_i \leq \eta/2\}|}{n\eta} - \frac{1}{\pi}\right| > \eps}\leq C\exp\left(-c\eps\sqrt{n\eta}\right)\]
        holds for all $K_\eps/n \leq \eta \leq 1$ and all $n \geq 2$.
    \end{theorem}
    Using \cref{levrep} with $\eta = 2\delta_{A, n}$, with probability at least $1 - \exp(-n^{\Omega(1)})$, it holds that $|\{i \in [n] : \lambda_i \in S_0\}| \leq O(n\delta_{A, n})$.
    By assumption, $f$ is continuous and bounded, so for all $\lambda_i \in S_0$, $|f(\frac{1}{2}\lambda_i)| \leq C_0$ for some constant $C_0 > 0$.
    Hence, with probability at least $1 - O(n^{-1}\polylog n)$
    \begin{multline*}
        \int f(x)\d L_{\frac{1}{2}X^+}  = \frac{1}{n}\sum_{i = 1}^{n}f\left(\frac{1}{2}\lambda_i\right)\I{\lambda_i \geq 0} \pm \frac{1}{n}\sum_{\lambda_i \in S_0}\left|f\left(\frac{1}{2}\lambda_i\right)\right| \\
        = \frac{1}{n}\sum_{i = 1}^{n}f\left(\frac{1}{2}\lambda_i\right)\I{\lambda_i \geq 0} \pm \frac{1}{n}\cdot O(n\delta_{A, n}) = \frac{1}{n}\sum_{i = 1}^{n}f_+(\lambda_i) \pm O(\delta_{A, n}) = \int f_+(x)\d L_X \pm o(1).
    \end{multline*}
    Next, let $L_{\frac{3}{2}X^-}$ be the empirical spectral measure of the matrix $\frac{3}{2}X^-/\sqrt{n} = \frac{3}{2}\sum_{i = n/2 + 1}^{n}\lambda_iu_iu_i^T$.
    Similar to the above, using \cref{rigeig0} and \cref{levrep}, we get that with probability at least $1 - O(n^{-1}\polylog n)$ for any continuous bounded function $f$ and $f_-(x) := f(3x/2)\I{x < 0}$ it holds $\int f(x)\d L_{\frac{3}{2}X^-} = \int f_-(x)\d L_X \pm o(1)$.}

    Finally, consider the empirical spectral distribution measure $L_{\frac{3}{2}X^- + \frac{1}{2}X^+}(x)$ of the matrix $(\frac{3}{2}X^- + \frac{1}{2}X^+)/\sqrt{n}$.
    For any continuous bounded $f$ with probability at least $1 - O(n^{-1}\polylog n)$:
    \begin{multline*}
        \int f(x)\d L_{\frac{3}{2}X^- + \frac{1}{2}X^+} = \frac{1}{n}\sum_{i = 1}^nf\left(
        \lambda_i\left(\left(\frac{3}{2}X^- + \frac{1}{2}X^+\right)/\sqrt{n}\right)\right) \\
        = \frac{1}{n}\sum_{i = n/2 + 1}^{n}f\left(\lambda_i\left(\frac{3}{2}X^-/\sqrt{n}\right)\right) + \frac{1}{n}\sum_{i = 1}^{n/2}f\left(\lambda_i\left(\frac{1}{2}X^+/\sqrt{n}\right)\right) \\
        = \int f(x)\d L_{\frac{3}{2}X^-} + \int f(x)\d L_{\frac{1}{2}X^+} = 
        \int f_-(x)\d L_{X} + \int f_+(x)\d L_{X} \pm o(1).
    \end{multline*}
    Now, by construction, $f_-(x)$ is non-zero only when $x < 0$, and $f_+(x)$ is non-zero only when $x > 0$.
    Meanwhile, a straightforward calculation shows that for any $x > 0$, $f(x)\d\Rho_{\frac{3}{2}, \frac{1}{2}}(x) = f_+(x)\d\Rho(x)$, and for any $x < 0$, $f(x)\d\Rho_{\frac{3}{2}, \frac{1}{2}}(x) = f_-(x)\d\Rho(x)$.
    Therefore,
    \begin{multline*}
        \int_{-\infty}^\infty f(x)\d L_{\frac{3}{2}, \frac{1}{2}}(x) - \int_{-\infty}^\infty f(x)\d\Rho_{\frac{3}{2}, \frac{1}{2}}(x)\\
        = \left(\int_{-\infty}^0 f(x)\d L_{\frac{3}{2}, \frac{1}{2}}(x) - \int_{-\infty}^0 f(x)\d\Rho_{\frac{3}{2}, \frac{1}{2}}(x)\right)
        + \left(\int_{0}^\infty f(x)\d L_{\frac{3}{2}, \frac{1}{2}}(x) - \int_{0}^\infty f(x)\d\Rho_{\frac{3}{2}, \frac{1}{2}}(x)\right)\\
        = \left(\int_{-\infty}^\infty f_-(x)\d L_{X}(x) - \int_{-\infty}^\infty f_-(x)\d\Rho(x)\right) +
        \left(\int_{-\infty}^\infty f_+(x)\d L_{X}(x) - \int_{-\infty}^\infty f_+(x)\d\Rho(x)\right) \pm o(1).
    \end{multline*}
    As a result, since empirical spectral measure $L_X$ converges almost surely to a standard semicircle law $\Rho$, applying the definition of almost sure convergence to $f_-$ and $f_+$ gives that for each $\eps > 0$
    \begin{multline*}
        \lim_{n\to\infty}\P{\left|\int_{-\infty}^\infty f(x)\d L_{\frac{3}{2}X^- + \frac{1}{2}X^+}(x) - \int_{-\infty}^\infty f(x)\d \Rho_{\frac{3}{2}, \frac{1}{2}}(x)\right| > \eps} \\
        \leq \lim_{n\to\infty}\P{\left|\int_{-\infty}^\infty f_-(x)\d L_{X} - \int_{-\infty}^\infty f_{-}(x)\d \Rho(x)  \pm o(1)\right| > \eps}\\ + \lim_{n\to\infty}\P{\left|\int_{-\infty}^\infty f_{+}(x)\d L_{X} - \int_{-\infty}^\infty f_+(x)\d \Rho(x) \pm o(1)\right| > \eps} = 0.
    \end{multline*}
    We conclude that by definition, $L_{\frac{3}{2}X^- + \frac{1}{2}X^+}$ converges almost surely to $\Rho_{\frac{3}{2}, \frac{1}{2}}$.
\end{proof}

\subsection{Spectrum of $\hD = \Diag(X^- - X^+)$}

Having determined the limiting distribution of the eigenvalues of $\hX$, we would like to do the same for the diagonal matrix $\hD = \Diag(\hX) = \Diag(X^- - X^+)$, where for all $i \in [n]$, $\hD_{ii} = \hX_{ii} = X^-_{ii} - X^+_{ii}$.
Note that if all diagonal entries of $\hD$ were exactly the same with high probability, i.e $\hD = \nu I$ for some $\nu \in \R$, adding matrix $\hD$ to any other symmetric matrix $M \in \S^n$ shifts all the eigenvalues of $M$ by exactly $\nu$. That is, $\lambda_i(M +\nu I) = \lambda_i(M) + \nu$ for all $i \in [n]$.
Observe that, due to symmetry of the distribution $\Xx_n$, for all $i, j \in [n]$, $\E{\hD_{ii}} = \E{\hD_{jj}}$.
If one could, on top of that, guarantee that for every $i \in [n]$ entries $\hD_{ii}$ are well-concentrated around their mean, then $\hD$ is equal to $(\nu + o(1))I$ for some $\nu \in \R$ with high probability, and we can easily analyze its affect on the spectrum when summed up with other matrices.
As observed in numerical experiments, the diagonal entries of $\hX = X^- - X^+$ do indeed seem concentrated around the same value, which is more or less expected from a Wigner matrix.

We are going to assume that our Wigner matrix $X\sim \Xx_n$ is taken from a distribution $\Xx_n$ for which the concentration of the diagonal entries of $\hX$ does hold.
\begin{assumption}\label{assumdiag}
    Let $X \sim \Xx_n$, let $\hX = X^- - X^+$, and let $\hD = \Diag(\hX)$.
    There exists $\nu \in \R$ and constant {$K > 0$}  such that for every $i \in [n]$, $\hD_{ii} = (1 + o(1))\nu$ with probabiltiy at least $1 - O(n^{-K})$.
\end{assumption}

For arbitrary Wigner matrices $\Xx_n$, it is not clear whether \cref{assumdiag} holds, and whether diagonal entries of $\hX = X^- - X^+$ are simultaneously well-concentrated around their means.
However, we can prove that the desired simultaneous concentration holds when $\Xx_n = \GOE(n)$. 

\begin{theorem}\label{diagconcbase}
     Let $X \sim \GOE(n)$, let $\hX = X^- - X^+$, and let $\hD = \Diag(\hX)$.
     {Then,
     $\E{\hD_{ii}} = (1 + o(1))\frac{8}{3\pi}\sqrt{n}$,
     }
     and for all $i \in [n]$, for any constant {$K > 0$}, $\P{\left|\hD_{ii} - {\E{\hD_{ii}}}\right| \geq n^{1/6}} = O(n^{-K})$.
\end{theorem}

{This theorem is proved later (\cref{diagconc}).
When $X\sim \Bern_{\pm 1}(n, 1/2)$, we show in \cref{diagconc2} that for every $i \in [n]$, $\E{\hD_{ii}} = (1 + o(1))(8/3\pi)\sqrt{n}$ as well. However, we manage to prove simultaneous concentration of only $o(\sqrt{n})$ diagonal entries at once.}

\subsection{Spectrum of $\wX$}

Recall the definition of matrix $\wX$.
Initially, we are given $X\sim \Xx_n$, where for all $i < j \in [n]$, $X_{ii} = 0$ and $X_{ij} \in \{\pm 1\}$ with equal probability.
Then, the matrix $\wX$ is constructed as follows: 
\[\forall i \neq j \in [n],\qquad \wX_{ii} = 0; \qquad \text{and}\qquad \wX_{ij} = \begin{cases}
        X^+_{ij} - X^-_{ij} = -\hX_{ij},& X_{ij} > 0;\\
        X^-_{ij} - X^+_{ij} = \hX_{ij},& X_{ij} < 0.
    \end{cases}\]

The matrix $\wX$, being a function of a Wigner random matrix $X \sim \Xx_n$, is a random matrix, {but its} entries are very much dependent on each other.
Already the entries of the matrix $\hX = X^- - X^+$ are not independent, and in $\wX$ we additionally change the signs in a deterministic way based of the signs of the original Wigner matrix $X$. Standard techniques for analysing the spectrum of random matrices are not applicable for $\wX$, as they assume that the entries are distributed independently from each other, or that dependencies, if exist, are very minor.

{Interestingly, in numerical experiments (\cref{fig:xeigs}) the distribution of eigenvalues of $\wX$ strongly resembles a semi-circle law on $[-2\alpha, 2\alpha]$, with $\alpha = \sqrt{1 - \frac{64}{9\pi^2}}$.}
The fact that the spectrum of $\wX$ resembles a scaled semicircle is not that surprising.
Since $X \sim \Xx_n$ is a Wigner matrix satisfying symmetricity  \cref{symdist}, for every $i \neq j \in [n]$, $X_{ij}$ can be positive or negative with equal probability.
But then, since $\wX_{ij} = -\hX$ if {$X_{ij} > 0$}, and $\wX_{ij} =\hX$ if $X_{ij} < 0$, one can view $\wX_{ij}$ for $i \neq j$ taking values $\{\pm \hX_{ij}\}$ with equal probability.
Combined with the fact that $\wX_{ii} = 0$ for all $i \in [n]$, matrix $\wX$ resembles a \textit{generalized} Wigner matrix.
\begin{figure}[h]
\includegraphics[width=16cm]{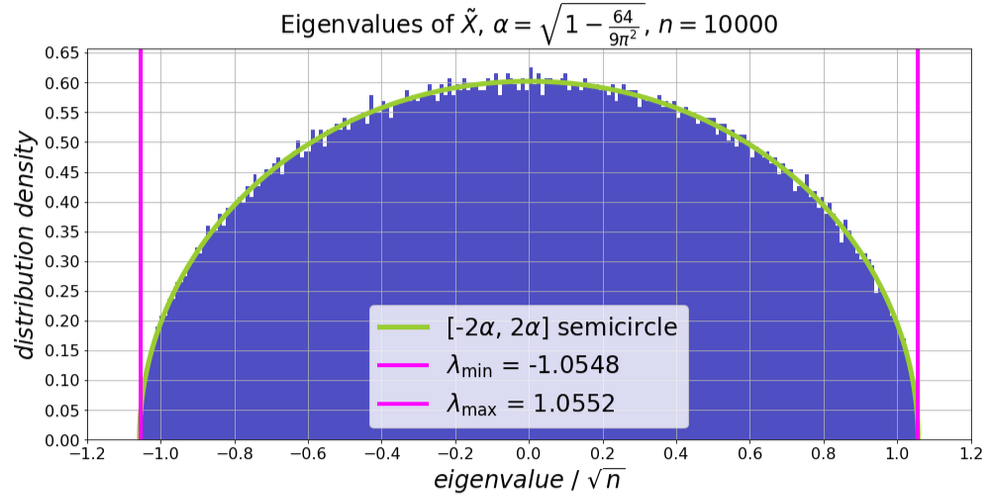}
\caption{Experimental distribution of eigenvalues of $\wX$}
\label{fig:xeigs}
\end{figure}

\begin{definition}\label{genwigdef}
    Let $\{Z_{ij}\}_{1\leq i \leq j}^n$ be independent zero mean real random variables, with $\E{Z_{ij}^2}$ possibly different for different $i, j$.
    Let $\Ww_n$ be a distribution over size-$n$ matrix $W_n$ s.t $(W_n)_{ij} = (W_n)_{ji} = Z_{ij}$.
    We call such $W_n \sim \Ww_n$ a \textbf{generalized Wigner matrix} with laws $\{Z_{ij}\}_{1\leq i \leq j}^n$.
\end{definition}
As mentioned in the approach overview (\cref{sec:approach}), instead of $\wX = Z\circ A_G$, we would like to consider $Z\circ A_{G'}$ for a new random graph $G'$, to be able to analyze spectrum of the sum.
Now we define the resulting distribution over matrices $Z\circ A_{G'}$ for $G'\sim G(n, 1/2)$ formally.
For $X\sim \Xx_n$, the distribution $\Ww_n := \Ww_n(X)$ over symmetric matrices is as follows.
For $W \sim \Ww_n$,
\[\forall i < j \in [n] :\quad W_{ii} = 0\quad\text{and}\quad W_{ij} = W_{ji} = \begin{cases}
    X^+_{ij} - X^-_{ij},& \wpr 1/2;\\
    X^-_{ij} - X^+_{ij},& \wpr 1/2;
\end{cases}\quad \text{independently}.\]
Then $W\sim \Ww_n$ is a generalized Wigner matrix, where for all $i < j \in [n]$, $\E{W_{ii}} = \E{W_{ij}} = 0$, $\E{W_{ij}^2} = \sigma_{ij}^2 = (X^-_{ij} - X^+_{ij})^2$ and $\E{W_{ii}^2} = 0$.
Note that in the definition of $\Ww_n$, we assume that $X \sim \Xx_n$ is already given, so the randomness in $W\sim \Ww_n$ is independent of the randomness used to obtain $X$, though the supports of the distributions of $W_{ij}$ are determined by the matrix $X$.

Results of \cite{GNT15} and \cite{C23} show that, under mild conditions on the variances $\sigma_{ij}^2$, the spectrum of a generalized Wigner matrix converges almost surely to a standard semi-circle.
\begin{theorem}[\cite{GNT15}, \cite{C23}]\label{gensemicirc}
    Let $W \in \R^{n\times n}$ be a generalized Wigner matrix with variances $\E{W_{ij}^2} = \sigma^2_{ij} < \infty$ for $i, j \in [n]$.
    Let $L_{W}(x) := \frac{1}{n}\sum_{i = 1}^n\I{\lambda_{i}(W/\sqrt{n}) \leq x}$ be the empirical spectral distribution measure of $W/\sqrt{n}$.
    Suppose that
    \begin{enumerate}
        \item for any constant $\eta > 0$, $\lim_{n\to\infty}\frac{1}{n^2}\sum_{i, j = 1}^n\E{W_{ij}^2\I{|W_{ij}| > \eta\sqrt{n}}} = 0$;
        \item there exists global constant $C$ such that for every $i \in [n]$,
        $\frac{1}{n}\sum_{j = 1}^n\sigma_{ij}^2 \leq C$;
        \item $\frac{1}{n}\sum_{i = 1}^n\left|\frac{1}{n}\sum_{j = 1}^n\sigma_{ij}^2 - 1\right|\xrightarrow{n\to\infty}0$.
    \end{enumerate}
    Then, $L_{W} \xrightarrow[n\to\infty]{\mathrm{a.s.}} \Rho$, {where $\Rho$ is the standard semi-circle law.}
\end{theorem}

Supporting \cref{gensemicirc}, numerical experiments with $X \sim \Bern_{\pm 1}(n, 1/2)$ show (\cref{fig:weigs}) that the spectrum of $W\sim \Ww_n(X)$ is distributed like a scaled standard semi-circle, with a scaling parameter essentially the same as the one we had for the matrix $\wX$.

\begin{figure}[h]
\includegraphics[width=16cm]{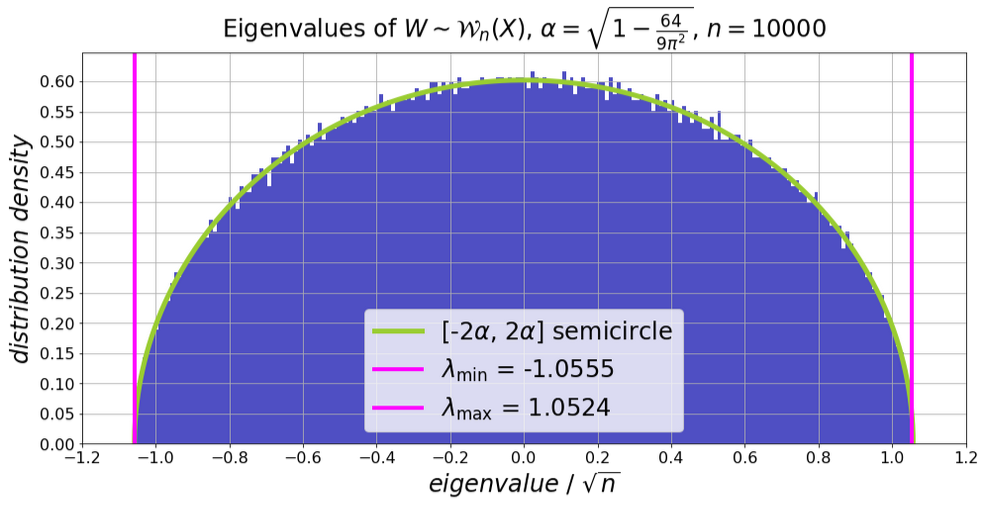}
\caption{Experimental distribution of eigenvalues of $W$}
\label{fig:weigs}
\end{figure}

For $W$, the scaling parameter $\alpha = \sqrt{1 - \frac{64}{9\pi^2}}$ arrives naturally from the following observation.
For every $i < j \in [n]$, $\E{W_{ij}^2} = \E{(X^- - X^+)_{ij}^2} = \E{\hX^2_{ij}}$.
At the same time, by \cref{PlusEigBound}, $\sum_{i, j = 1}^n\hX^2_{ij} = \sum_{i = 1}^n\lambda_i(\hX)^2 = \sum_{i = 1}^n\lambda_i(X)^2 = (1 + o(1))n^2$, and by \cref{assumdiag} and \cref{diagconc2}, $\E{\hX^2_{ii}} = (1 + o(1))\frac{64}{9\pi^2}n$ for every $i \in [n]$.
Then, the average variance of all non-diagonal entries of $\hX$ (and hence $W$) is {roughly $\frac{1}{n(n-1)}(\sum_{i, j = 1}^n\E{\hX^2_{ij}} - \sum_{i = 1}^n\E{\hx_{ii}^2}) \simeq 1 - \frac{64}{9\pi^2}$.}

These observations suggest that, even though the entries of the matrix $\wX$ by construction are heavily dependent on each other, and ``signs'' of the entries are chosen in a deterministic way, the resulting matrix $\wX$ actually resembles a generalized Wigner matrix $W\sim \Ww_n$.
So, in order to better understand why exactly the sum of the matrices $(3/2)X^- + (1/2)X^+$, $\hD$ and $\wX$ pushes the largest eigenvalue below $2\sqrt{n}$, one could try performing the analysis with the assumption that matrix $\wX$ is actually a random matrix $W$ taken from distribution $\Ww_n(X)$.
\begin{assumption}\label{wxdist}
    Let $X \sim \Xx_n$ {with $X_{ij} \in \{\pm 1\}$ with probability $1/2$ each}, and define $\wX$ as:
    \[\forall i \neq j \in [n],\quad \wX_{ii} = 0; \quad \text{and}\quad \wX_{ij} = \begin{cases}
         X^+_{ij} - X^-_{ij},& X_{ij} > 0;\\
         X^-_{ij} - X^+_{ij},& X_{ij} < 0.
    \end{cases}\]
    {Let $\Ww_n = \Ww_n(X)$ be a distribution over $n\times n$ matrices s.t for $W\sim \Ww_n$, $\forall i \in [n]$, $W_{ii} = 0$ and
    \[\forall i < j \in [n] :\quad W_{ij} = W_{ji} = \begin{cases}
    X^+_{ij} - X^-_{ij},& \wpr 1/2;\\
    X^-_{ij} - X^+_{ij},& \wpr 1/2;
\end{cases}\quad \text{independently}.\]
    Let $L_W$ denote the empirical spectral distribution measure of matrix $W\sim \Ww_n$, and let $\mu$ be a probability measure such that $L_W \xrightarrow{n\to\infty}\mu$.
    Then, if $L_{\wX}$ is the empirical spectral distribution measure of matrix $\wX$, it holds that $L_{\wX} \xrightarrow{n\to\infty} \mu$.
    That is, the distribution of eigenvalues of $\wX$ converges to the same probability measuire as the distribution of eigenvalues of $W\sim \Ww_n$.}
\end{assumption}

\cref{wxdist} allows us to study the spectrum of $\wX$, as {matrices $W \sim \Ww_n(X)$ introduce additional randomness independent of $\wX$, making it easier to analyze distribution of its eigenvalues and eigenvectors.
For example,} in the case when $X\sim \GOE(n)$, we show that the variances of entries $W\sim \Ww_n(X)$ satisfy the conditions of \cref{gensemicirc} with high probability, and as a result the eigenvalues of $W$ have a scaled semi-circle distribution.
\begin{theorem}\label{genwig}
    Let $X\sim \GOE(n)$ and let $W\sim \Ww_n(X)$, where $\Ww_n(X)$ is defined in \cref{wxdist}.
    For $\alpha > 0$, let $L_{\alpha W}(x) := \frac{1}{n}\sum_{i = 1}^n\I{\lambda_{i}(\alpha W/\sqrt{n}) \leq x}$ be the empirical spectral distribution measure of the eigenvalues of $\alpha W/\sqrt{n}$.
    There exists constant $\alpha > 0$ such that $L_{\alpha W} \xrightarrow[n\to\infty]{\mathrm{a.s.}} \Rho$, i.e the distribution of the eigenvalues of $\alpha W/\sqrt{n}$ converges almost surely to a standard semi-circle law.
\end{theorem}
\cref{genwig} is proved in \cref{genwigmain}.
When $X \sim \Bern_{\pm 1}(n, 1/2)$, i.e $X = A_G$, it is not clear whether analogous statement holds.
However, due to universality property of eigenvectors and eigenvalues of Wigner matrices (\cite{ER11}), and the fact that the distribution of $W$ depends only on the variances of the entries of $X^- - X^+$,
it seems reasonable to assume that a statement similar to \cref{genwig} holds also when $X \sim \Bern_{\pm 1}(n, 1/2)$.
Based on \cref{genwig}, we make the following assumption about the spectrum of $W\sim \Ww_n(X)$ when $X \sim \Bern_{\pm 1}(n, 1/2)$.
\begin{assumption}\label{genwigbern}
    Let $X \sim \Xx_n$, and let $W\sim \Ww_n(X)$, where $\Ww_n(X)$ is defined in \cref{wxdist}.
    There exists constant $\alpha > 0$ such that $L_{\alpha W} \xrightarrow[n\to\infty]{\mathrm{a.s.}} \Rho$.
\end{assumption}

\subsection{Spectrum of $(3/2)X^- + (1/2)X^+ + (1/2)\wX$}

{Equipped with \cref{hxdist}, \cref{assumdiag}, \cref{wxdist} and \cref{genwigbern}, we ``know'' individual distributions of the eigenvalues of the matrices $(3/2)X^- + (1/2)X^+$, $\hD$, $W$ and $\wX$ respectively.}
$\hD$ is a diagonal matrix designed to zero-out the diagonal of the sum $(3/2)X^- + (1/2)X^+$, so the change to the spectrum achieved by adding $\hD$ is clear.
Next step is to determine the spectrum of the other matrices' sum, i.e $(3/2)X^- + (1/2)X^+ + (1/2)\wX$ and $(3/2)X^- + (1/2)X^+ + (1/2)W$.
Matrix $W \sim \Ww_n$ is more convenient to work with, and by \cref{wxdist}, the distributions of spectra of $\wX$ and $W$ converge to the same measure. However, this does not guarantee that the spectra of $\wX$ and $W$ behave similarly when summed up with other matrices (in particular with $(3/2)X^- + (1/2)X^+$), {because the spectrum of a sum of matrices depends not only on the spectra of the individual matrices, but also on their eigenvectors.}

Empirical evidence suggests that it is actually the case that the spectra 
of the matrices $(3/2)X^- + (1/2)X^+ + (1/2)\wX$ for $X\sim \Bern_{\pm 1}(n, 1/2)$, and $(3/2)X^- + (1/2)X^+ + (1/2)W$ for $W \sim \Ww_n(X)$, do behave quite similarly.
In the pictures below, we show the distribution of the eigenvalues of $(3/2)X^- + (1/2)X^+ + (1/2)\wX$ in \cref{fig:sumx}, and the distribution of the eigenvalues of $(3/2)X^- + (1/2)X^+ + (1/2)W$ in \cref{fig:sumw}.
One can observe that the plots are essentially the same, with largest eigenvalues being very close to each other.

\begin{figure}[h]
\includegraphics[width=16cm]{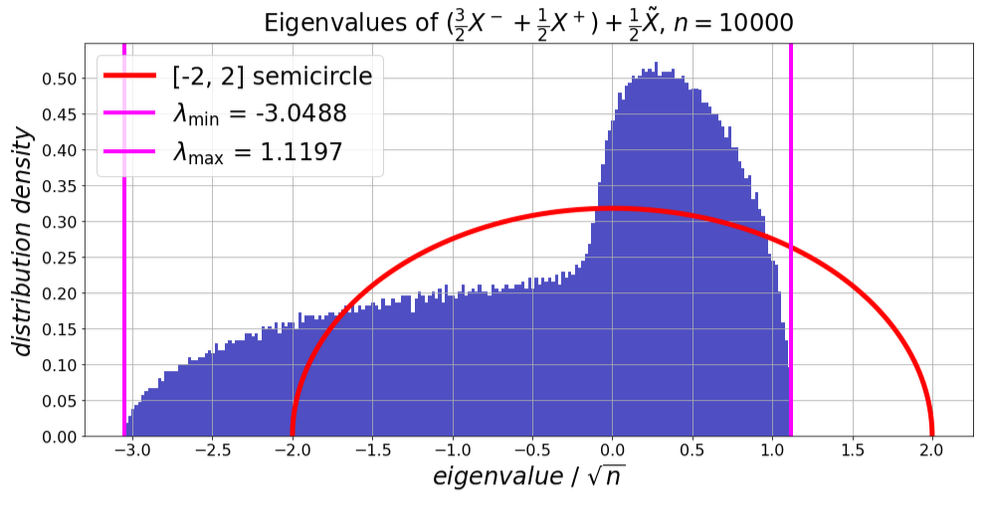}
\caption{Experimental distribution of eigenvalues of $(3/2)X^- + (1/2)X^+ + (1/2)\wX$}
\label{fig:sumx}
\end{figure}

\begin{figure}[h]
\includegraphics[width=16cm]{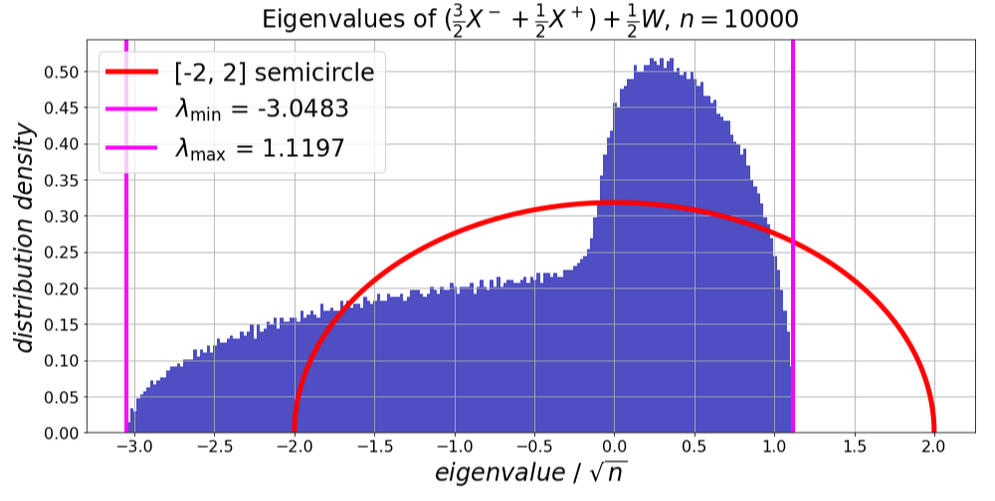}
\caption{Experimental distribution of eigenvalues of $(3/2)X^- + (1/2)X^+ + (1/2)W$}
\label{fig:sumw}
\end{figure}

{These observations suggest that it is reasonable to assume that not only the spectra of $X\sim \Xx_n$ and $W\sim \Ww_n(X)$ are close, but also the spectra of the sums $(3/2)X^- + (1/2)X^+ + (1/2)\wX$ and $(3/2)X^- + (1/2)X^+ + (1/2)W$ remain similar.
\begin{assumption}\label{wxsums}
    Let $X \sim \Xx_n$ and let $W \sim \Ww_n(X)$, where $\Ww_n(X)$ is defined in \cref{wxdist}.
    Let $L_{+ W}$ denote the empirical spectral distribution measure of the sum $(3/2)X^- + (1/2)X^+ + (1/2)W$, and let $\nu$ be a probability measure such that $L_{+ W} \xrightarrow{n\to\infty}\nu$.
    Then, if $L_{+ \wX}$ is the empirical spectral distribution measure of the sum $(3/2)X^- + (1/2)X^+ + (1/2)\wX$, it holds that $L_{+ \wX} \xrightarrow{n\to\infty}\nu$.
\end{assumption}

Given \cref{wxsums}, the next natural step is to understand the distribution of the spectrum of the sum $(3/2)X^- + (1/2)X^+ + (1/2)W$ for $W \sim \Ww_n(X)$, and estimate its largest eigenvalue.
We {shall} use the fact that, given some $X\sim \Xx_n$, the randomness introduced in the matrix $W\sim \Ww_n(X)$ is \textit{independent} of the randomness of $X$ (though the variances of the entries of $W$ still depend on $X$).
Therefore, one could treat the sum $(3/2)X^- + (1/2)X^+ + (1/2)W$ as a sum of a \textit{deterministic} matrix $(3/2)X^- + (1/2)X^+$ and a \textit{random} matrix $(1/2)W$.
The eigenvalue distributions (and other properties) of the sums of random and deterministic matrices is studied by the \textbf{free probability theory}, which {is the topic of} the next section.
}

\newcommand{\orho}{\bar{\rho}}
\section{Free Probability and the spectrum of $X + (1/2)\hX + (1/2)\wX$}
\label{sec:free}

\subsection{Background on free probability}

The area of free probability theory was first introduced by Voiculescu in \cite{V86, V87}, applying it to study additions and multiplications of non-commutative random variables.
In the subsequent works of \cite{V91, VB92, V92}, Voiculescu et.al managed to apply the free probability theory to random matrices, showing that, given a family of independent Wigner matrices $\{W_i\}_{i = 1}^m$, the mixed moments of these matrices, i.e expressions of the form $\frac{1}{n}\tr(W_{i_1}\cdot \ldots\cdot  W_{i_k})$ for $i_1,\ldots, i_k \in [m]$ and $k \geq 1$ arbitrary, converge almost surely and in expectation to corresponding mixed moments $\sigma^{(m)}(W_{i_1},\ldots, W_{i_k})$ of $m$ free semi-circular random variables.
In other words, independent Wigner matrices $\{W_i\}_{i = 1}^m$ are \textbf{asymptotically free}, which allows to apply the theory of non-commutative probability, specifically regarding the distribution of the sums and products of free random variables, to study the distribution of the eigenvalues of sums and products of independent Wigner matrices from the asymptotically free family $\{W_i\}_{i = 1}^m$.

We provide a brief overview of the main definitions and results on free probability theory, relevant to random matrices and their spectrums.
\begin{definition}[\cite{V86}]\label{freeness}
    Let $(\Aa, \fe)$ be a non-commutative probability space, where $\Aa$ is a unital associative algebra (over $\C)$, and $\fe : \Aa \to \C$ is a unital ($\fe(1) = 1$) linear functional.
    Let $\Aa_i\subseteq \Aa$, $i \in I$, be unital subalgebras of $(\Aa, \fe)$.
    We say that $\{\Aa_i\}_{i \in I}$ are \textbf{free} (or \textbf{freely independent}), if for every $k \geq 1$, for every $a_1,\ldots, a_k$ such that
    \begin{itemize}
        \item for all $j \in [k]$, $a_{j} \in \Aa_{i(j)}$ for some $i(j) \in I$;
        \item for all $j \in [k - 1]$, $i(j) \neq i(j + 1)$ (but we allow same non-consecutive subalgebras);
        \item for all $j \in [k]$, $\fe(a_j) = 0$;
    \end{itemize}
    it holds that $\fe(a_1\ldots a_k) = 0$.
    In this case random variables $a_1,\ldots, a_k$ are also called \textbf{free}.
\end{definition}
The function $\fe$ in \cref{freeness} serves the role of the expectation operator, and the definition can be interpreted as a rule of calculating joint moments of free random variables in a non-commutative space.
For example, let $a_1 \in A_1$, $a_2 \in A_2$ where $A_1$ and $A_2$ are free subalgebras of $(\Aa, \fe)$.
Then, by linearity of $\fe$, $\fe(a_1 - \fe(a_1)\One) = 0$ and $\fe(a_2 - \fe(a_2)\One) = 0$, hence it must hold
\[\fe\left[(a_1 - \fe(a_1)\One)(a_2 - \fe(a_2)\One)\right] = 0 
    \iff \fe(a_1a_2) = \fe(a_1)\fe(a_2).\]
As we see, in case of two non-commutative random variables $a_1, a_2$, freeness gives us the exact same expression as for calculating expectation of a product of two independent random variables in the standard probability theory.
However, because $(\Aa, \fe)$ is a non-commutative probability space, expressions for moments of higher orders will look quite different.
For example, if $a_1, a_1' \in A_1$ and $a_2, a_2'\in A_2$, where $A_1$ and $A_2$ are free, \cref{freeness} implies that $\fe(a_1a_2a_1') = \fe(a_1a_1')\fe(a_2)$
and $\fe(a_1a_2a_1'a_2') = \fe(a_1a_1')\fe(a_2)\fe(a_2') + \fe(a_1)\fe(a_1')\fe(a_2a_2') - \fe(a_1)\fe(a_1')\fe(a_2)\fe(a_2')$.
In order to derive these equalities, we rely on the fact from \cref{freeness} that the mixed moments of centered random variances must be equal to $0$.
For example, the first equality is obtained as follows:
\begin{multline*}
    \fe\left[(a_1 - \fe(a_1)\One)(a_2 - \fe(a_2)\One)(a_1' - \fe(a_1')\right] = 0 \\
    \implies \fe\left[a_1a_2a_1' - \fe(a_1)a_2a_1' - a_1\fe(a_2)a_1' + \fe(a_1)\fe(a_2)a_1' \right. \\
    \left. - a_1a_2\fe(a_1') + \fe(a_1)a_2\fe(a_1') + a_1\fe(a_2)\fe(a_1') - \fe(a_1)\fe(a_2)\fe(a_1')\One\right] = 0\\
    \implies \fe(a_1a_2a_1') - \fe(a_1)\fe(a_2a_1') - \fe(a_2)\fe(a_1a_1') + \fe(a_1)\fe(a_2)\fe(a_1') -\\
    \fe(a_1a_2)\fe(a_1') + \fe(a_1)\fe(a_2)\fe(a_1') + \fe(a_1)\fe(a_2)\fe(a_1') - \fe(a_1)\fe(a_2)\fe(a_1') = 0.
\end{multline*}
Now, using the fact that $a_2$ is free from $a_1$ and $a_1'$, we have $\fe(a_1a_2) = \fe(a_1)\fe(a_2)$ and $\fe(a_2a_1') = \fe(a_2)\fe(a_1')$.
Then, after cancellations, the only terms left are $\fe(a_1a_2a_1') - \fe(a_1a_1')\fe(a_2) = 0$.

Even though free independence was developed as an analogue to classical, commutative independence, it is not a generalization.
Classical commuting random variables $a_1, a_2$ are free only in the trivial case (when either $a_1$ or $a_2$ is a constant).
Indeed, in commutative case we would have $\fe(a_1a_1a_2a_2) = \fe(a_1a_2a_1a_2)$, while the formulas above give us $\fe(a_1a_1a_2a_2) = \fe(a_1^2)\fe(a_2^2)$, but $\fe(a_1a_2a_1a_2) = \fe(a_1^2)\fe(a_2)^2 + \fe(a_1)^2\fe(a_2^2) - \fe(a_1)^2\fe(a_2)^2$, 
and together they imply $\fe\left[(a_1 - \fe(a_1)\One)^2\right]\cdot \fe\left[(a_2 - \fe(a_2)\One)^2\right] = 0$, so either $a_1$ or $a_2$ has variance $0$.
As a corollary, for any $a_1, a_2 \in \Aa$, if $a_1$ and $a_2$ commute w.r.t $\fe$ and are nontrivial, then $a_1$ and $a_2$ are not free.

A natural question in probability theory is determining the limiting distribution of a given sequence of random variables, for example, a limiting distribution of the empirical spectral measure of some random matrix (or sums of random matrices).
In order to apply the toolkit of free probability to this subject, we need to define a notion of \textbf{asymptotic freeness}.
\begin{definition}\label{distconv}
    Let $(\Aa_n, \fe_n)_{n = 1}^\infty$, and $(\Aa, \fe)$ be non-commutative probability spaces, and let $(a_n)_{n = 1}^\infty$, where $a_n \in \Aa_n$, be a sequence of non-commutative random variables, and let $a \in \Aa$.
    We say that $a_n$ \textbf{converges in distribution} to $a$, denoted by $a_n \xrightarrow[n\to\infty]{\d} a$, if for any fixed $k \in \N$,
    \[\lim_{n\to\infty}\fe_n(a_n^k) = \fe(a^k).\]
    More generally, for an index set $I$, and variables $\{a_n^{(i)}\}_{i \in I}$ and $\{a^{(i)}\}_{i \in I}$, where $a_n^{(i)} \in \Aa_n$ and $a^{(i)} \in \Aa$, we say that $\{a_n^{(i)}\}_{i \in I} \xrightarrow[n\to\infty]{\d} \{a^{(i)}\}_{i \in I}$ if for any fixed $k \in \N$ and any $i_1,\ldots, i_k \in I$,
    \[\lim_{n\to\infty}\fe_n(a_n^{(i_1)}\ldots a_n^{(i_k)}) = \fe(a^{(i_1)}\ldots a^{(i_k)}).\]
\end{definition}
As an example, let $\Aa^{(1)}_n, \Aa^{(2)}_n \subseteq \Aa_n$ be two subalgebras, $\Aa^{(1)}, \Aa^{(2)} \subseteq \Aa$, and let $a^{(1)}_n \in \Aa^{(1)}_n$, $a^{(2)}_n \in\Aa^{(2)}_n$ and $a^{(1)} \in \Aa^{(1)}$, $a^{(2)}\in \Aa^{(2)}$.
Then $a^{(1)}_n, a^{(2)}_n \xrightarrow[n\to\infty]{\d} a^{(1)}, a^{(2)}$ if for any $k \in \N$ and any $i_1,\ldots, i_k \in \{1, 2\}$, mixed moments $\fe_n(a_n^{(i_1)}\ldots a_n^{(i_k)})$ of $a^{(1)}_n, a^{(2)}_n$ converge to corresponding mixed moments $\fe(a^{(i_1)}\ldots a^{(i_k)})$ of $a^{(1)}, a^{(2)}$.
\begin{definition}\label{asfree}
    Let $(\Aa_n, \fe_n)$ be a non-commutative probability space, and let $a^{(1)}_n,\ldots, a^{(m)}_n \in \Aa_n$ for some $m \geq 1$.
    We say that $a^{(1)}_n,\ldots, a^{(m)}_n$ are \textbf{asymptotically free} if there exists a non-commutative probability space $(\Aa, \fe)$ and free random variables $a^{(1)},\ldots, a^{(m)} \in \Aa$, such that $a^{(1)}_n,\ldots, a^{(m)}_n \xrightarrow[n\to\infty]{\d} a^{(1)},\ldots, a^{(m)}$.
    We say that $a^{(1)}_n,\ldots, a^{(m)}_n$ are \textbf{asymptotically almost sure free}, if there exists a non-commutative probability space $(\Aa, \fe)$ and free random variables $a^{(1)},\ldots, a^{(m)} \in \Aa$, such that for any $k \in \N$, any $i_1,\ldots, i_k \in [m]$, any $\eps > 0$, 
    \[\lim_{n\to\infty}\P{\left|\fe_n(a_n^{(i_1)}\ldots a_n^{(i_k)}) - \fe(a^{(i_1)}\ldots a^{(i_k)})\right| > \eps} = 0.\]
\end{definition}
In addition to asymptotic freeness, we will need to formally define a semi-circular element.
\begin{definition}\label{ssrv}
    Let $(\Aa, \fe)$ be a non-commutative probability space, and let $s \in \Aa$ be a random variable with moments 
    \[\fe(s^{2m + 1}) = 0\qquad \text{and} \qquad \fe(s^{2m}) = \sigma^{2m}\cdot \frac{1}{m + 1}\binom{2m}{m},\]
    where $\sigma > 0$ is a constant.
    Then $s$ is called a \textbf{semi-circular element} of variance $\sigma^2$.
    When $\sigma = 1$, it is called a \textbf{standard} semi-circular element.
\end{definition}
When it comes to random matrices, asymptotic freeness is defined with respect to expected normalized trace $\E{\frac{1}{n}\tr(\cdot)}$ serving as an expectation operator.
This is a natural expectation operator to consider, as, for example, the moment-method based proof of the Wigner semicircle result \cref{semicircle} shows that for a Wigner matrix $W_n$ and any fixed $k$, $\E{\frac{1}{n}\tr(n^{-k/2}W_n^k)}$ converges to $0$ if $k$ is odd, and if $k = 2m$ it converges to $\frac{1}{m + 1}\binom{2m}{m}$, which is the $2m$-th moment of a standard semi-circular distribution on $\R$.
But in the language of \cref{distconv} it means that $W_n$ converges in distribution to a standard semi-circular element.
\begin{definition}\label{freemat}
    Let $\{A_n\}_{n = 1}^\infty$ and $\{B_n\}_{n = 1}^\infty$ be two sequences of random matrices of size $n$, such that for each $n$, $A_n$ and $B_n$ are defined on the same probability space, and let $\E[n]$ denote the expectation on that probability space.
    Let $\Mm_n = \Mm_n(A_n, B_n)$ denote the algebra defined by $A_n, B_n$, and consider the non-commutative probability space $(\Mm_n, \fe_n)$ where $\fe_n(\cdot) := \E[n]{\frac{1}{n}\tr(\cdot)}$.
    
    We say that $A_n, B_n$ are \textbf{asymptotically free}, if there exists a non-commutative probability space $(\Mm, \fe)$ and free random variables $a, b \in \Mm$, such that $A_n, B_n \xrightarrow[n\to\infty]{\d} a, b$, which by definition means that for any $k \in \N$ and any $p_1, q_1, \ldots, p_k, q_k \geq 0$,
    \[\E[n]{\frac{1}{n}\tr\left(A_n^{p_1}B_n^{q_1}\ldots A_n^{p_k}B_n^{q_k}\right)} \xrightarrow{n\to\infty} \fe(a^{p_1}b^{q_1}\ldots a^{p_k}b^{q_k}).\]
\end{definition}
As we see, asymptotic freeness for random matrices is just a particular case of a general asymptotic freeness.
Similarly, \textbf{asymptotically almost sure freeness} definition for $\{A_n\}_{n = 1}^\infty$ and $\{B_n\}_{n = 1}^\infty$ is analogous to the one in \cref{asfree}.
Voiculescu in the work of \cite{V91} discovered that the free probability theory can be applied to empirical spectral distributions of random matrices, and proved the following fundamental result.
\begin{theorem}[\cite{V91, V98}]\label{goefree}
    For arbitrary $t \geq 1$, let $\{A_n^{(1)}, \ldots, A_n^{(t)}\}_{n = 1}^\infty$ be $t$ independent matrices from $\GOE(n)$.
    Then, $A_n^{(1)}, \ldots, A_n^{(t)}\xrightarrow[n\to\infty]{\d} s_1,\ldots, s_t$ where $s_1, \ldots, s_t$ are free standard semi-circular elements.
    That is, $A_n^{(1)}, \ldots, A_n^{(t)}$ are asymptotically free, furthermore, the convergence holds almost surely.
\end{theorem}
The result can be strengthened to the case where some of the matrices $A_n^{(1)}, \ldots, A_n^{(t)}$ are not random, but arbitrary deterministic matrices $D_n^{(r)}$ for $r \geq 1$ such that $D_n^{(r)} \xrightarrow[n\to\infty]{\d}d_r$, i.e we know the limiting distribution of the eigenvalues of $D_n^{(r)}$.
\begin{theorem}[\cite{V91, V98}]\label{goedetfree}
    For arbitrary $t, r \geq 1$, let $\{A_n^{(1)}, \ldots, A_n^{(t)}\}_{n = 1}^\infty$ be $t$ independent matrices from $\GOE(n)$, and let $\{D_n^{(1)}, \ldots, D_n^{(r)}\}_{n = 1}^\infty$ be $r$ deterministic matrices of size $n$ such that for $j \in [r]$, $D_n^{(j)} \xrightarrow[n\to\infty]{\d}d_j$.
    Then, $A_n^{(1)}, \ldots, A_n^{(t)}, D_n^{(1)}, \ldots, D_n^{(r)} \xrightarrow[n\to\infty]{\d} s_1,\ldots, s_t, d_1,\ldots, d_r$ where $s_1, \ldots, s_t$ are standard semi-circular elements, and {for every $j \in [r]$, $s_1,\ldots, s_t, d_j$ are free.}
    That is, {for every $j \in [r]$, $A_n^{(1)}, \ldots, A_n^{(t)}, D_n^{(j)}$ }are asymptotically free, and the convergence holds almost surely.
\end{theorem}
Later, the result was extended from just $\GOE(n)$ to arbitrary Wigner matrices, with additional requirement on deterministic matrices to have bounded spectrum.

\begin{theorem}[\cite{AGZ09, MS17}]\label{wigdetfree}
    For arbitrary $t, r \geq 1$, let $\{A_n^{(1)}, \ldots, A_n^{(t)}\}_{n = 1}^\infty$ be $t$ independent Wigner matrices of size $n$ (with off-diagonal variance $1$), and let $\{D_n^{(1)}, \ldots, D_n^{(r)}\}_{n = 1}^\infty$ be $r$ deterministic matrices  of size $n$ such that for $j \in [r]$, $\sup_{n}\|D_n^{(j)}\|_{\mathrm{op}}< \infty$ and $D_n^{(j)} \xrightarrow[n\to\infty]{\d}d_j$. 
    Then, $A_n^{(1)}, \ldots, A_n^{(t)}, D_n^{(1)}, \ldots, D_n^{(r)} \xrightarrow[n\to\infty]{\d} s_1,\ldots, s_t, d_1,\ldots, d_r$ where $s_1, \ldots, s_t$ are standard semi-circular elements, and {for every $j \in [r]$, $s_1,\ldots, s_t, d_j$ are free.}
    That is, {for every $j \in [r]$, $A_n^{(1)}, \ldots, A_n^{(t)}, D_n^{(j)}$ are asymptotically free}, and the convergence holds almost surely.
\end{theorem}

Freeness of random matrices is a very important property and is {helpful} when it comes to determining the spectrum of matrix sums and products.
Recall that in classical probability, if two random variables $a, b$ with measures $\mu_a, \mu_b$ respectively are independent, we can determine the distribution of $a + b$ via convolution operation, that is, $\mu_{a + b} = \mu_a * \mu_b$.
In other words, if two random variables are independent, the distribution of their sum depends only on the individual distributions of the summands.
In free probability theory, a similar statement holds, i.e if two non-commutative random variables are free, then one can determine the distribution of their sum from individual distributions of the summands, via an operation called \textbf{free convolution}.
However, unlike in classical probability, the notion of free convolution is more complicated, and requires introduction of additional terminology.

\subsection{Free convolution}

\begin{definition}
    Let $(\Aa, \fe)$ be a non-commutative probability space, and let $a, b \in \Aa$ be random variables with distributions $\mu_a, \mu_b$ respectively.
    If $a$ and $b$ are free, and $\mu_a, \mu_b$ have compact support, $a + b$ has distribution $\mu_{a \boxplus b} := \mu_a \boxplus \mu_b$, where $\boxplus$ is called a \textbf{free additive convolution}.
    The measure $\mu_{a \boxplus b}$ is uniquely determined by its moments $\fe((a + b)^k) = \int x^k\d \mu_{a\boxplus b}(x)$, and depends only on individual distributions $\mu_a, \mu_b$ and their freeness, and not on realizations $a, b$.
\end{definition}
Combining this definition with \cref{freemat}, we get that if two random symmetric matrices are asymptotically free, the empirical spectral distribution of their sum converges to the free convolution of the individual spectrums' distributions.
\begin{theorem}[\cite{S93}, \cite{AGZ09}, \cite{MS17}]\label{freeconvspec}
    Let $\{A_n\}_{n = 1}^\infty$ and $\{B_n\}_{n = 1}^\infty$ be two sequences of random symmetric matrices of size $n$, let $L_{A_n}, L_{B_n}$ be the empirical spectral distribution measures of $A_n, B_n$ respectively.
    Suppose that $L_{A_n}\xrightarrow{n\to\infty}\mu_A$ and $L_{B_n}\xrightarrow{n\to\infty}\mu_B$ for some probability measures $\mu_A, \mu_B$ with compact support.

    Let $L_{A_n + B_n}$ denote the empirical spectral distribution measure of $A_n + B_n$.
    If $A_n$ and $B_n$ are asymptotically free, then $L_{A_n + B_n} \xrightarrow{n\to\infty}\mu_A \boxplus \mu_B$.
\end{theorem}
As a corollary, by \cref{wigdetfree} we can determine limiting spectral distribution of the sum of a Wigner matrix $A_n$ and any deterministic matrix $D_n$ by taking the individual spectrums and computing their convolution, that is a free convolution of a semicircular distribution $s$ and limiting distribution $d$ of the eigenvalues of $D_n$.
However, in applications, we are mostly interested in the extremal eigenvalues of matrices, while asymptotic freeness, as well as limiting spectral distribution, characterize the bulk of the spectrum of a random matrix.
So, even if the eigenvalues of the free sum $A + B$ are distributed according to the free convolution, $\lambda_1(A + B)$ might {be an outlier and} deviate from the {top endpoint of the} support of the sum with non-negligible probability.
The result of \cite{CDFF10} {shows that when $A_n$ is a Wigner matrix, and $B_n$ is a uniformly bounded deterministic matrix (so by \cref{wigdetfree} they are free), then for $n$ large enough the spectrum of $A_n + B_n$ is included in a neighborhood of the support of $\mu_A \boxplus \mu_B$ with probability tending to $1$ as $n$ grows, and the largest eigenvalue of $A_n + B_n$ converges almost surely to the top endpoint of the support of $\mu_A \boxplus \mu_B$.}
\begin{theorem}[\cite{CDFF10}]\label{freestick}
    Let $\{A_n\}_{n = 1}^\infty$ and $\{B_n\}_{n = 1}^\infty$ be two sequences of random symmetric matrices of size $n$, and suppose that $A_n$ and $B_n$ are asymptotically free, with $L_{A_n} \xrightarrow{n\to\infty}\mu_A$ and $L_{B_n}\xrightarrow{n\to\infty}\mu_B$. 
    If for every $\eps > 0$,
        \[\P{\text{for all large $n$, }\spec(A_n + B_n) \subset \supp(\mu_A \boxplus \mu_B) + (-\eps, \eps)}\xrightarrow{n\to\infty} 1,\]
    then $\lambda_1(A + B)$ (resp. $\lambda_n(A + B)$) converges almost surely to the top (resp. bottom) endpoint of the support of $\mu_A \boxplus \mu_B$.

    {Furthermore, when $A_n$ is a Wigner matrix with symmetrically distributed entries (so $\mu_A$ is a scaled semicircle distribution), and $B_n$ is a uniformly bounded deterministic matrix, i.e $\sup_n\|B_n\|_{\mathrm{op}} < \infty$, it does hold for every $\eps > 0$ that
    \[\P{\text{for all large $n$, }\spec(A_n + B_n) \subset \supp(\mu_A \boxplus \mu_B) + (-\eps, \eps)}\xrightarrow{n\to\infty} 1.\]}
\end{theorem}
By \cref{freestick}, in order to determine largest eigenvalue of the sum of a Wigner matrix and a bounded deterministic matrix, it suffices to find the top endpoint of the support of their free convolution.
{However, it is not known whether analogous claim holds when $A_n$ is a generalized Wigner matrix.}

\subsection{Application to $X + (1/2)\hX + (1/2)W$}

Recall that we are interested in determining the largest eigenvalue of the sum $(3/2)X^- + (1/2)X^+ + (1/2)W = X + (1/2)\hX  + (1/2)W$, where $\hX = X^- - X^+$ and {matrix $W \sim \Ww_n(X)$ is defined as \[\forall i \neq j \in [n],\qquad W_{ii} = 0; \qquad \text{and}\qquad W_{ij} = W_{ji} = \begin{cases}
    X^+_{ij} - X^-_{ij},& \wpr 1/2;\\
    X^-_{ij} - X^+_{ij},& \wpr 1/2;
\end{cases}\quad \text{independently}.\]}
From \cref{hxdist} we know that the empirical spectral distribution of $(X + (1/2)\hX)/\sqrt{n}$ converges almost surely to $\Rho_{\frac{3}{2}, \frac{1}{2}}$, and {under \cref{genwigbern}, the empirical spectral distribution $L_W$ of $W/\sqrt{n}$, converges almost surely to a scaled semi-circle $\Rho_\alpha$ with some parameter $\alpha > 0$.
If $W$ was a standard, not generalized, Wigner matrix, then by \cref{freestick} the spectrum of $X + (1/2)\hX + (1/2)W$ would converge to the free convolution of $\Rho_{\frac{3}{2},\frac{1}{2}}$ and $\Rho_\alpha$, and, most importantly, the largest eigenvalue of $X + (1/2)\hX + (1/2)W$ converges to the top endpoint of $\supp(\Rho_{\frac{3}{2}, \frac{1}{2}} \boxplus \Rho_\alpha)$. 
Since $W$ is a generalized Wigner matrix, it is not know whether a claim similar to \cref{freestick} holds for matrices $X + (1/2)\hX$ and $(1/2)W$.
However, given individual distributions $\Rho_{\frac{3}{2},\frac{1}{2}}$ and $\Rho_\alpha$, one could try to recover the density and the support of their free convolution $\Rho_{\frac{3}{2}, \frac{1}{2}} \boxplus \Rho_\alpha$ and check whether its density curve (and the top endpoint of the support) matches the one of the experimental distribution of the eigenvalues of $X + (1/2)\hX + (1/2)W$.
We are going to briefly cover below the existing tools that allow us to estimate $\supp(\Rho_{\frac{3}{2}, \frac{1}{2}} \boxplus \Rho_\alpha)$.}

\begin{definition}[\cite{V91}]\label{functionals}
    For a probability distribution function $\mu$ on $\R$, the \textbf{Cauchy transform} of a point $z \in \C_+$ (complex numbers with non-negative imaginary part) is 
    \[G_\mu(z) := \int_\R\frac{1}{z - x}\d \mu(x).\]
    The \textbf{$R$-transform} of $\mu$ for $z \in \C_+$ is
    \[R_\mu(z) = G_\mu^{-1}(z) - \frac{1}{z},\]
    where $G_\mu^{-1}$ is the inverse of $G_\mu$ w.r.t composition, i.e $G^{-1}(G(z)) = z$.
\end{definition}
The $R$-transform in particular is one of the central notions in free probability theory, due to its direct connection to the free convolution.
\begin{theorem}[\cite{V91}]\label{rfree}
    For any two freely independent non-commutative variables $a, b \in (\Mm, \fe)$, with individual measures $\mu_a, \mu_b$, for all $z \in \C_+$,
    \[R_{\mu_a\boxplus \mu_b}(z) = R_{\mu_a}(z) + R_{\mu_b}(z).\]
    As an immediate corollary, by definition of the $R$-transform,
    \[G^{-1}_{\mu_a \boxplus \mu_b}(z) = G^{-1}_{\mu_a}(z) + G^{-1}_{\mu_b}(z) - \frac{1}{z}.\]
\end{theorem}
Using \cref{rfree}, one can recover the support of $\mu_a \boxplus \mu_b$ directly from $G^{-1}_{\mu_a \boxplus \mu_b}(z)$ provided that measures $\mu_a, \mu_b$ are ``nice enough''.
\begin{theorem}[\cite{ON13, CY23}]\label{freealgo}
    Let $\mu_a, \mu_b$ be two probability measures of free non-commutative random variables $a, b$, such that
    \begin{enumerate}
        \item $\mu_a$ has compact support $[s_a, t_a]$, $\mu_b$ has compact support $[s_b, t_b]$;
        \item $\mu_a$ has sqrt-behaviour, $\mu_b$ is a Jacobi measure;
        \item Cauchy transforms $G_{\mu_a}, G_{\mu_n}$ are invertable.
    \end{enumerate}
    Let $g(z) := G^{-1}_{\mu_a \boxplus \mu_b}(z) = G^{-1}_{\mu_a}(z) + G^{-1}_{\mu_b}(z) - \frac{1}{z}$.
    Then the support of $\mu_a \boxplus \mu_b$ is contained in the interval $[s, t] = [g_s, g_t]$ where $g_s, g_t$ are unique zeroes of the derivative $g'(z)$ in the intervals 
    \[\big(\max\{G_{\mu_a}(s_a), G_{\mu_b}(s_b)\}, 0\big)\qquad \text{and}\qquad \big(0, \min\{G_{\mu_a}(t_a), G_{\mu_b}(t_b)\}\big)\]
    respectively.
\end{theorem}
Having sqrt-behavior means that $\d \mu_a(x)$ is of the form $\psi(x)\sqrt{x - c_1}\sqrt{c_2 - x} \d x$ for some $c_1, c_2$ and function $\psi$ continuously differentiable on $[c_1, c_2]$, a well-known example of such measure is the semicircle distribution.
Jacobi measure is a generalization of sqrt-behaviour, replacing $\sqrt{x - c_1}\sqrt{c_2 - x}$ with $(x - c_1)^{\gamma_1}(c_2 - x)^{\gamma_2}$ for $\gamma_1, \gamma_2 > -1$.
In case of $X + (1/2)\hX + (1/2)W$, matrix $X + (1/2)\hX$ is a sum of two quarter-circles, while $(1/2)W$ {by \cref{genwigbern} is a generalized Wigner with semicircle spectrum.
So, if $X + (1/2)\hX$ and $(1/2)W$ were free and a claim analogous to \cref{freestick} was true for these matrices}, one could apply \cref{freealgo}, compute Cauchy transforms and $R$-transforms of individual distributions, and then recover the support of the limiting distribution of the eigenvalues of $X + (1/2)\hX + (1/2)\wX$, in particular, recovering the largest eigenvalue.

Interestingly enough, numerical experiments suggest that the distribution of the spectrum of $X + (1/2)\hX + (1/2)W$ matches closely the free convolution of {(assumed)} individual limiting spectral distributions of $X + (1/2)\hX$ and $(1/2)W$, i.e $\Rho_{\frac{3}{2}, \frac{1}{2}}$ and $\Rho_\alpha$ respectively, and also that $\lambda_1(X + (1/2)\hX + (1/2)W)$ sticks to the support of $\Rho_{\frac{3}{2}, \frac{1}{2}} \boxplus \Rho_\alpha$.
For $ \mu = \Rho_\alpha$, the closed form for $G_{\mu}$ and $R_{\mu}$ is known, however, obtaining a closed-form expression of $G_\nu$ and $R_\nu$ for $\nu = \Rho_{\frac{3}{2}, \frac{1}{2}}$ is pretty challenging. {
For example, directly computing the Cauchy-transform integral via standard techniques, we derived a closed-form expression of $G_{\Rho_\beta}$ for a \textit{scaled quartercircle} distribution 
\[\rho_\beta(x) = \d \Rho_\beta(x) := \frac{1}{\beta^2\pi}\sqrt{4\beta^2 - x^2}\I{x \in [0, 2\beta]},\]
which looks as follows (we omit the calculations):  
\[G_{\Rho_\beta}(z) = - \frac{\sqrt{4\beta^2 - z^2}}{\beta^2\pi}\ln\left(-\frac{\sqrt{4\beta^2-z^2} + 2\beta}{z}\right) + \frac{2}{\beta \pi } + \frac{z}{2\beta^2}.\]
The expression for $\beta = 1$ (in which case $\beta = \beta^2$) can be found in \cite{RE08}.
As we see, the closed form for $G_{\Rho_\beta}$ is quite complicated already by itself, and computing the $R$-transform for the sum of two differently scaled quartercircle distributions is even more complex.
So, in order to compute the Cauchy and $R$ transforms of $\Rho_{\frac{3}{2}, \frac{1}{2}}$, as well as its free convolution with $\Rho_\alpha$, we used the approximation algorithm for computing the support in \cref{freealgo}, provided by the authors of \cite{CY23}.
The algorithm represents Cauchy and $R$ transforms of corresponding distributions as certain power series (depending on moments and other parameters of distributions), and then approximates their values at given points using various high-precision numerical methods.
This algorithm allowed us to approximately compute the density and the support of $\Rho_{\frac{3}{2}, \frac{1}{2}} \boxplus \Rho_\alpha$, which we then compared with the experimental distribution of eigenvalues of $X + (1/2)\hX + (1/2)W = (3/2)X^- + (1/2)X^+ + (1/2)W$ for $W\sim \Ww_n(X)$.
The reader can see the results in \cref{fig:freesumw} below.}
\begin{figure}[h]
\includegraphics[width=16cm]{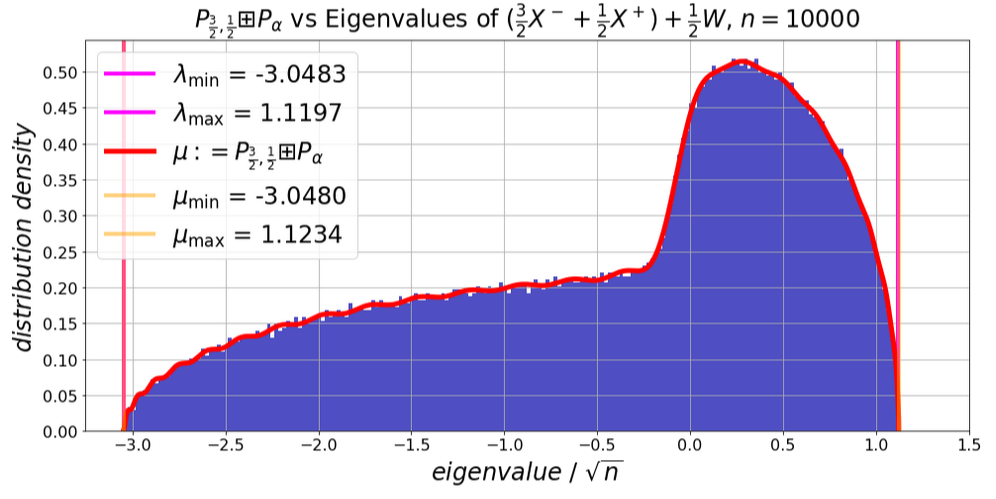}
\caption{Comparison of numerically computed distribution $\Rho_{\frac{3}{2}, \frac{1}{2}} \boxplus \Rho_\alpha$ with experimental distribution of eigenvalues of $(3/2)X^- + (1/2)X^+ + (1/2)W$}
\label{fig:freesumw}
\end{figure}
{As one can observe, the density curve of $\Rho_{\frac{3}{2}, \frac{1}{2}} \boxplus \Rho_\alpha$ matches the contour of the eigenvalue distribution of the matrix $(3/2)X^- + (1/2)X^+ + (1/2)W$, and the endpoints of the support of $\Rho_{\frac{3}{2}, \frac{1}{2}} \boxplus \Rho_\alpha$ essentially coincide with the top and bottom eigenvalues of the sum.
In addition, in support of \cref{wxsums}, we observe the same phenomenon when comparing the free convolution $\Rho_{\frac{3}{2}, \frac{1}{2}} \boxplus \Rho_\alpha$ with the original sum $(3/2)X^- + (1/2)X^+ + (1/2)\wX$, as one can see below in \cref{fig:freesumx}.
That is, not only the eigenvalue distributions of $\wX$ and $W\sim \Ww_n(X)$ are the same (as in \cref{wxdist}), but also the eigenvalue distributions of {sums of} these matrices with the matrix $X + (1/2)\hX = (3/2)X^- + (1/2)X^+$.}
\begin{figure}[h]
\includegraphics[width=16cm]{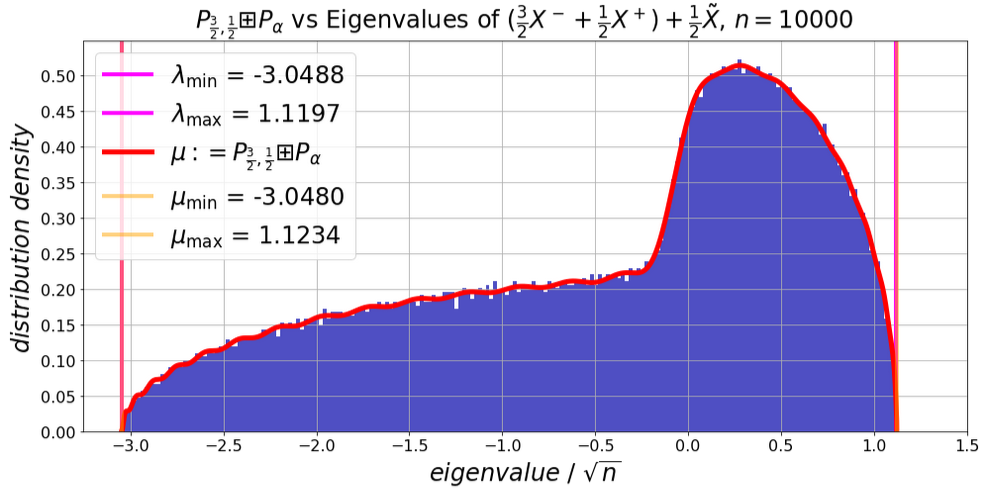}
\caption{Comparison of numerically computed distribution $\Rho_{\frac{3}{2}, \frac{1}{2}} \boxplus \Rho_\alpha$ with experimental distribution of eigenvalues of $(3/2)X^- + (1/2)X^+ + (1/2)\wX$}
\label{fig:freesumx}
\end{figure}

{These observations suggest that matrices $X + (1/2)\hX = (3/2)X^- + (1/2)X^+$ and $(1/2)W$ are indeed asymptotically free, and, as in \cref{freeconvspec}, the spectrum of the sum $X + (1/2)\hX + (1/2)W$ converges to the free convolution $\Rho_{\frac{3}{2}, \frac{1}{2}} \boxplus \Rho_\alpha$.
Furthermore, they suggest that a claim similar to the \cref{freestick} may indeed hold for the sum $X + (1/2)\hX + (1/2)W$, and the largest eigenvalue of the sum matches the top endpoint of the support of the free convolution.
We summarize these claims in the following assumption.
\begin{assumption}\label{freenesswx}
    Let $X \sim \Xx_n$, $\hX = X^- - X^+$, and $W \sim \Ww_n(X)$ be as defined earlier.
    Then
    \begin{enumerate}
    \item matrices $X + (1/2)\hX$ and $(1/2)W$ are asymptotically free;
    \item the empirical spectral distribution measure of $X + (1/2)\hX + (1/2)W$ converges to $\Rho_{\frac{3}{2}, \frac{1}{2}} \boxplus \Rho_\alpha$;
    \item for every $\eps > 0$, for all large $n$, $\spec(X + (1/2)\hX + (1/2)W)\subset \supp(\Rho_{\frac{3}{2}, \frac{1}{2}} \boxplus \Rho_\alpha) + (-\eps, \eps)$ with probability {tending to} $1$ {as $n \to \infty$}; 
\end{enumerate}
\end{assumption}}

{Following the discussion regarding the eigenvalues of the diagonal matrix $\hD$, (\cref{assumdiag} and \cref{diagconcbase})}, we expect that zeroing-out the diagonal of $(3/2)X^- + (1/2)X^+ + (1/2)W$ by adding the matrix $-(1/2)\hD$ will increase all eigenvalues by exactly $\frac{4}{3\pi}\sqrt{n}$, {and the same should happen for $(3/2)X^- + (1/2)X^+ + (1/2)\wX$}.
This does indeed seem to be true, according to experimental results, as one can see in \cref{fig:freesumwd} for $W$ and \cref{fig:freesumxd} for $\wX$.
{The experimental results show that in both cases, the largest eigenvalue of the final sum is roughly $1.55\sqrt{n}$, the number we initially stated in \cref{claim:1.54}.}

\begin{figure}[h]
\includegraphics[width=16cm]{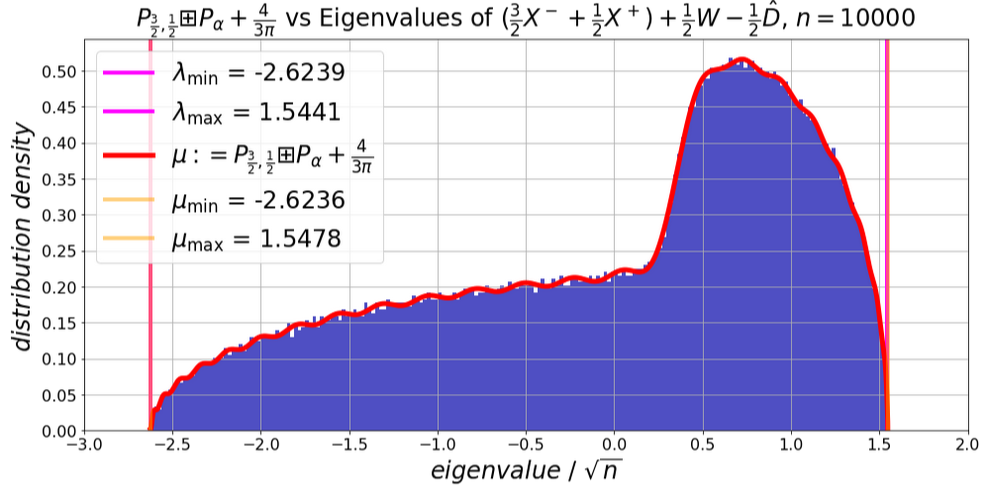}
\caption{Comparison of numerically computed distribution $\Rho_{\frac{3}{2}, \frac{1}{2}} \boxplus \Rho_\alpha$ shifted by $4/(3\pi)$ with experimental distribution of eigenvalues of $(3/2)X^- + (1/2)X^+ + (1/2)W - (1/2)\hD$}
\label{fig:freesumwd}
\end{figure}
\begin{figure}[h]
\includegraphics[width=16cm]{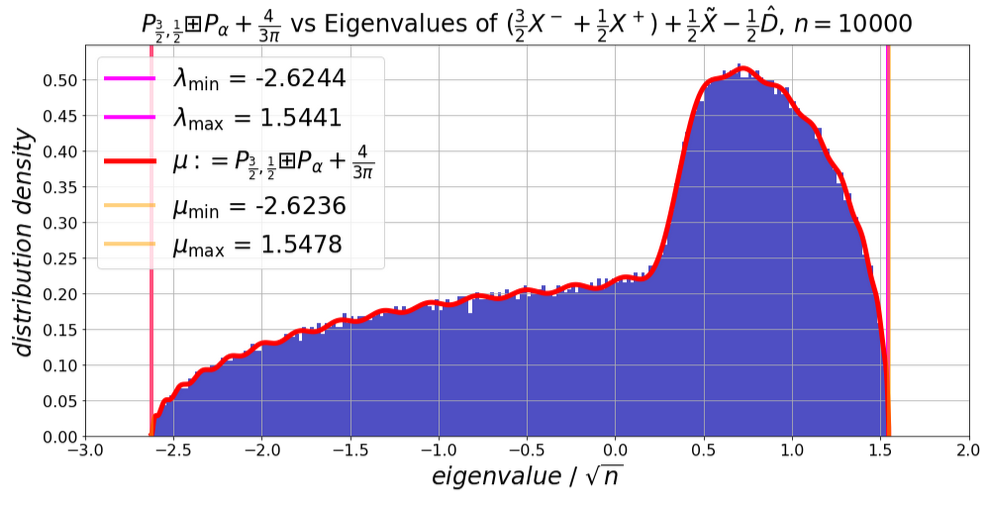}
\caption{Comparison of numerically computed distribution $\Rho_{\frac{3}{2}, \frac{1}{2}} \boxplus \Rho_\alpha$ shifted by $4/(3\pi)$ with experimental distribution of eigenvalues of $(3/2)X^- + (1/2)X^+ + (1/2)\wX - (1/2)\hD$}
\label{fig:freesumxd}
\end{figure}

Unfortunately, it is not known whether generalized Wigner matrices {(which $W\sim \Ww_n(X)$ is)} are free with bounded deterministic matrices, not even if the deterministic matrix is diagonal. 
Furthermore, there are examples of {generalized Wigner matrices that are \textbf{not} free with diagonal deterministic matrices. 
{For example, consider a block-diagonal matrix $W$ of size $2n$ with two blocks of size $n$, where the first $n$-block is a Wigner matrix $W_1$ with entries of variance $1$, while the second $n$-block is a Wigner matrix $W_2$ with entries of variance $2$. 
Then, a diagonal matrix $D$ such that the first $n$ entries are $-1$, while the last $n$ entries are $+1$, will not be free with $W$, as the entries of $D$ depend on the variance of the corresponding block of $W$.
As a result, if we add $D$ to $W$, it will change the eigenvalues of $W$ in a non-homogeneous way depending on which block they correspond to, and the sum will not be a free convolution.}
The latter implies that, for a {generalized Wigner matrix to be free with deterministic matrices}, it needs to satisfy some additional conditions.
One of such possible conditions, for example, may be that the {average variances in all large enough submatrices of the generalized Wigner matrix are roughly the same.}

The result of \cite{AZ06} shows that $W + D$, where $W$ is generalized Wigner and $D$ is deterministic diagonal, converges in distribution to some unique probability measure $\mu_{W + D}$.
However, the obtained bounds on the support of $\mu_{W + D}$ are not strong enough, specifically, the result establishes a bound on the largest singular value of the sum, while in our case (as observed in experiments) $|\lambda_n(W + D)|$ is substantially larger than $|\lambda_1(W + D)|$.
It is known (\cite{MS17}) that generalized Wigner matrices are asymptotically \textit{operator-valued free with amalgamation over diagonals}, however, the notion of operator-valued freeness is quite different from standard freeness, and is usually applied to block-random matrices with different block distributions, and performs poorly in the case when matrix is distributed differently entry-wise.

{Proving the points made in \cref{freenesswx}, i.e asymptotic freeness of $X + (1/2)\hX$ and $(1/2)W$ and the fact that the largest eigenvalue of their sum sticks to the top endpoint of the support of the free convolution, seems to be quite challenging.
The standard approach to proving that two (random) matrices $A_n, B_n$ are asymptotically free is to directly show that \cref{freemat} holds, i.e for every $k \in \N$, all mixed moments of $A_n, B_n$ of order $k$ converge to a corresponding mixed moment of order $k$ of some free random variables.
In our case, from \cref{hxdist} and \cref{genwigbern} we know the limiting distributions of the eigenvalues of $X + (1/2)\hX$ and $(1/2)W$.
The difficulty comes from the fact that the variances of the entries of $W$ depend on matrix $X$, and therefore are correlated with the entries of $X + (1/2)\hX$.
In addition, known proofs of freeness (\cref{goedetfree} and \cref{wigdetfree}) heavily rely on the fact that the variances of entries of considered matrix are all equal to $1$ (or some other fixed value, same for all entries), i.e that the matrix is Wigner, and not generalized Wigner like $W$.
{Finally, to show that $\lambda_1(X + (1/2)\hX + W)$ sticks to the top-end of the support of the free convolution of the corresponding distributions, one may try to extend the second part of \cref{freestick} to the setting of generalized Wigner matrices (from just Wigner).
However, the proof of \cref{freestick} for standard Wigner case already assumes the freeness of the matrices participating in the sum.}

\newpage
\section{Algorithm for spectral radius}\label{sec:specradius}

Recall that $\orho(G)$, where $G = (V, E)$ is a graph on $n$ vertices, can be formulated as the value of the optimal solution to the following SDP (here $\sigma_1(M)$ denotes the largest singular value of matrix $M$):
\[\begin{aligned}
    &\minimize_{M \in \S^n} && \sigma_{1}(M)\\
    &\subto&& M_{ii} = 0,&\forall i \in V;\\
    &&& M_{ij} = 1,&\forall (i, j) \in E.
\end{aligned}\]
We denoted by $\cal{M}_G$ the class of symmetric matrices of order $n$, where $M \in \cal{M}_G$ if and only if $M_{ii} = 1$ for every $i \in [n]$, and $M_{ij} = 1$ for every $(i,j) \in E$.
Then, $\orho(G) = \min_{M \in \Mm(G)}\sigma_{1}(M)$.

We use an algorithm similar to the one described in \cref{sec:approach} and \cref{sec:analysis} to find $M \in \Mm(G)$ such that $\sigma_1(M) < 2\sqrt{n}$.
First, we represent $M$ as $A_G + Y$, where $A_G$ is $\pm 1$ adjacency matrix of $G$, and $Y = M - A_G$.
Then, we express matrix $Y$ as
\[Y = \frac{1}{2}Z - \frac{1}{2}D_Z - \frac{1}{2}Z\circ A_G,\]
where matrix $Z$ has the same eigenvectors as $A_G$ (but different eigenvalues).
That is, if $\lambda_k(A_G)$, $v_k$ are $k$-th eigenvalue and $k$-th eigenvector of $A_G$, $k \in [n]$, then $Z = \sum_{k = 1}^n\alpha_kv_kv_k^T$ for some $\alpha_k$-s.

As described in \cref{sec:results}, in order to minimize the largest eigenvalue of the sum, we chose $\alpha_k = -\lambda_k(A_G)$ for $1 \leq k \leq n/2$ and $\alpha_k = \lambda_k(A_G)$ for $n/2 < k \leq n$.
In case of spectral radius, we will choose eigenvalues $\alpha_k$, $k\in [n]$, quite differently.
Suppose that $n$ is even.
Then, for $k \leq n/2$, we set $\alpha_k = \eta - \lambda_k(A_G)$, and for $k > n / 2$ we set $\alpha_k = -\eta - \lambda_k(A_G)$, {for some particular choice of $\eta$.}
{We will choose $\eta = \frac{3\pi}{8}\sqrt{n}$, for reasons that will become clearer later (in \cref{etaval}).}

Such choice of $\alpha_k$-s for minimization of spectral radius $\orho(G)$ is motivated by the following observation.
Consider the representation $A_G + Y = A_G + \frac{1}{2}Z - \frac{1}{2}D_Z - \frac{1}{2}Z\circ A_G$.
We would like the spectrum of $A_G + Y$ to be symmetric, so that the largest and the smallest eigenvalues of the sum are close in absolute value.
Thus, it would be reasonable to seek for a matrix $Z$ with symmetric spectrum, so that $A_G + \frac{1}{2}Z$ and $Z\circ A_G$ both have symmetric eigenvalues (otherwise the spectrum of the free convolution would not be symmetric).
{In matrix $M = A_G + Y(Z)$, by design, if $(A_G)_{ij} \geq 0$, then $(A_G + Y)_{ij} = (A_G)_{ij}$ --- we keep positive entries the same, and if $(A_G)_{ij} < 0$, then $(A_G + Y)_{ij} = Z_{ij}$ --- we change the entries.
So, in some intuitive (but formally incorrect) sense matrix $M$ can be viewed as half of the sum $A_G + Z$, because half of the entries of $M$ come from $A_G$, and the other half comes from $Z$.
Therefore, in order to minimize the spectral radius of $M$, one can seek to minimize the spectral radius of the sum $A_G + Z$.
So, in order to reduce the spectral radius of the sum to $0$, it makes sense to choose the eigenvalues $\alpha_k$, $k \in [n]$, of $Z$ as $\alpha_k = -\lambda_k(A_G)$ for $k \in [n]$.
And, since originally $A_G$ has semicircular distribution, the spectrum of $Z$ would be symmetric as well.

Of course, since we obtain $M$ not by adding $A_G + Z$, but by replacing free entries of $A_G$ with corresponding entries of $Z$, such choice is not good enough.
In particular, as was shown in the proof of \cref{freeentdistproof}, unless the average entry in the free entries of $M$ is equal to $-1$, the largest eigenvalue (as well as the spectral radius) of $M = A_G + Y(Z)$ will be at least $(2 - o(1))\sqrt{n}$.
For this reason, we symmetrically shift the eigenvalues $\alpha_k$, $k \in [n]$, of $Z$ by a value $\eta$, setting for $1 \leq k \leq n/2$, $\alpha_k = -\eta - \lambda_k(A_G)$, and for $n / 2 < k \leq n$, $\alpha_k = \eta - \lambda_k(A_G)$.
As we show below in \cref{etaval}, if we choose $\eta = \frac{3\pi}{8}\sqrt{n}$ then the average value of the free entry of $Z$ with these eigenvalues is equal to $-1$ with high probability.
With this condition met, there is hope that such choice of $Z$ may push $\sigma_1(M)$ below $2\sqrt{n}$.}

{In addition, the proof of \cref{thm:radius} (\cref{specthm}) implies that for any matrix $B$ with zero diagonal and average non-free entry equal to $1$, the value $\sigma_1(B)$ cannot be smaller than $(\frac{3\pi}{8} - o(1))\sqrt{n}$.
For $A_G + Z$ the average non-free entry is $1$, and the average free entry is $-1$.
So, in a sense, matrix $Z$ is selected in such a way that among all matrices with zero diagonal, average non-free entry equal to $1$ and average free entry equal to $-1$, the matrix $A_G + Z$ would have the smallest possible spectral radius (according to \cref{thm:radius}).}

{
For given expressions $a, b, F$, we will use the notation $(-1)^{a \leq b}F$ (and $(-1)^{a > b})F$) to indicate that when $a \leq b$ (respectively $a > b$), we multiply $F$ by $-1$.}
Consider the $(i, j)$-th entry of $M = A_G + \frac{1}{2}Z - \frac{1}{2}D_Z - \frac{1}{2}Z\circ A_G$, for $i < j \in [n]$.
It is easy to see that, since $A_G = \sum_{k = 1}^n\lambda_k(A_G) v_kv_k^T$ and $Z = \sum_{k = 1}^n\left[(-1)^{k > n/2}\eta - \lambda_k(A_G)\right]v_kv_k^T$,
\[\left(A_G + \frac{1}{2}Z\right)_{ij} = \frac{1}{2}\sum_{k = 1}^n\left[\lambda_k(A_G) + (-1)^{k > n/2}\eta\right]v_{ki}v_{kj}.\]
The off-diagonal entries of $D_Z$ are $0$, and as for $Z\circ A_G$, by definition
\[(Z\circ A_G)_{ij} = \begin{cases}
    \sum_{k = 1}^n\left[(-1)^{k > n/2}\eta - \lambda_k(A_G)\right]v_{ki}v_{kj},& (A_G)_{ij} \geq 0;\\
    -\sum_{k = 1}^n\left[(-1)^{k > n/2}\eta - \lambda_k(A_G)\right]v_{ki}v_{kj},& (A_G)_{ij} < 0.
\end{cases}\]
Then, it is easy to see that 
\[M_{ij} = \left(A_G + \frac{1}{2}Z - \frac{1}{2}D_Z - \frac{1}{2}Z\circ A_G\right)_{ij} = \begin{cases}
    \sum_{k = 1}^n\lambda_k(A_G)v_{ki}v_{kj} = (A_G)_{ij} = 1,& (A_G)_{ij} \geq 0;\\
     \eta\sum_{k = 1}^n(-1)^{k> n/2} v_{ki}v_{kj} ,& (A_G)_{ij} < 0.
\end{cases}\]
Recall that we would like to choose matrix $Z$ in such a way that the absolute value of the average entry of $Z$ in free locations (those with $(i, j) \notin E \iff (A_G)_{ij} < 0$) is small, and likewise for the non-free locations.
This is done in order to ensure that the average entry in the non-free locations remains $+1$, and the average entry in the free locations remains $-1$.
While the non-free entries of $M = A_G + Y$ are all equal $+1$ and automatically satisfy this property, the free locations depend on the value of parameter $\eta$ that we choose.
Hence, we need to choose $\eta$ so that
\[\frac{1}{|\oE|}\sum_{(i, j)\notin E}\eta\sum_{k = 1}^n(-1)^{k> n/2} v_{ki}v_{kj} \simeq -1.\]
\begin{claim}\label{etaval}
    When $\eta = \frac{3\pi}{8}\sqrt{n}$, with probability at least $1 - O(n^{-1}\log n)$, for any $\eps > 0$.
    \[\eta \cdot \frac{1}{|\oE|}\sum_{(i, j)\notin E}\sum_{k = 1}^n(-1)^{k> n/2} v_{ki}v_{kj} = -1 \pm O(n^{\eps - 1/2}).\]
\end{claim}
\begin{proof}
    Let $G \sim G(n, 1/2)$, $G = (V, E)$, be a random graph, let $A_G = \sum_{k = 1}^n\lambda_ku_ku_k^T$ be the spectral decomposition of its $\pm 1$ adjacency matrix $A_G$.
    {Similarly to the proof of \cref{zentries}}, for $k \in [n]$ we consider
    \[f(k) := \sum_{(i, j)\in E}v_{ki}v_{kj};\qquad g(k) := \sum_{(i, j)\notin E}v_{ki}v_{kj}.\]
    {As established in \cref{fgrel}, for every $k \in [n]$ it holds that $f(k) - g(k) = \lambda_k$.
    Then
    \[\sum_{k= 1}^{n/2}(-1)^{k > n/2}[f(k) - g(k)] = \sum_{k= 1}^{n/2}(-1)^{k > n/2}\lambda_k.\]
    As shown in \cref{PlusEigBound}, with probability at least $1 - O(n^{-1}\polylog n)$, 
    {\[\sum_{k = 1}^{n/2} \lambda_k = \frac{4}{3\pi}n^{3/2} \pm O(n^{1/2}\polylog n) \quad\text{and}\quad \sum_{k = n/2 + 1}^{n}\lambda_k = -\frac{4}{3\pi}n^{3/2} \pm O(n^{1/2}\polylog n),\]}
    hence
    \[\sum_{k= 1}^{n/2}(-1)^{k > n/2}[f(k) - g(k)] = \frac{8}{3\pi}n^{3/2} \pm O(n^{1/2}\polylog n).\]
    Again by \cref{fgrel}, for every $k \in [n]$ it holds that for any constant $\eps > 0$ we have $f(k) + g(k) = \pm n^\eps$ with probability at least $1 - O(n^{-10})$.
    It follows that with probability at least $1 - O(n^{-9})$, for every $k \in [n]$ we have $f(k) - g(k) = -2g(k) \pm n^\eps$, so 
    \[\sum_{k= 1}^{n/2}(-1)^{k > n/2}[f(k) - g(k)]  = \sum_{k= 1}^{n/2}(-1)^{k > n/2}[-2g(k) \pm n^\eps] 
        = 2\sum_{k= 1}^{n/2}(-1)^{k \leq n/2}g(k) \pm O(n^{1 + \eps}).\]}
    As a result, with probability at least $1 - O(n^{-1}\polylog n)$,
    \[\sum_{k = 1}^n(-1)^{k\leq n/2}g(k) = \frac{1}{2}\sum_{k= 1}^n(-1)^{k > n/2}[f(k) - g(k)] \pm O(n^{1 + \eps}) = \frac{4}{3\pi}n^{3/2}\pm O(n^{1 + \eps}).\]
    We would like to determine the value of $\eta \in \R$ such that
    \[\eta \cdot \frac{1}{|\oE|}\sum_{(i, j)\notin E}\sum_{k = 1}^n(-1)^{k> n/2}v_{ki}v_{kj} = -1.\]
    We have just shown that, with probability at least $1 - O(n^{-1}\polylog n)$,
    \[\sum_{(i, j)\notin E}\sum_{k = 1}^n(-1)^{k> n/2}v_{ki}v_{kj} = \sum_{k = 1}^n(-1)^{k> n/2}g(k) = -\frac{4}{3\pi}n^{3/2}\pm O(n^{1 + \eps}).\]
    Since $G \sim G(n, 1/2)$, $\E{|\oE|} = \frac{n(n - 1)}{2}$ (the diagonal of adjacency matrix $A_G$ is $0$).
    And, by Chernoff, $\P{|\oE| - \frac{n^2 - n}{2} > cn\log n} = \P{|\oE| - \E{|\oE|} > cn\log n} \leq O(n^{-c})$ for a constant $c > 0$. 
    So, with probability at least $1 - O(n^{-1}\polylog n)$,
    \[\eta \cdot \frac{1}{|\oE|}\sum_{(i, j)\notin E}\sum_{k = 1}^n(-1)^{k> n/2}v_{ki}v_{kj} = \eta \cdot -\frac{8}{3\pi}\left(n^{-1/2}\pm O(n^{\eps - 1})\right).\]
It immediately follows that when $\eta = \frac{3\pi}{8}\sqrt{n}$, the sum above is equal to $-1 \pm O(n^{\eps - 1/2})$ with probability at least $1 - O(n^{-1}\polylog n)$.  
\end{proof}

{We} note that even if we replace the $\pm 1$-Wigner matrix $A_G$ with a Gaussian one, i.e $A_G \sim \GOE(n)$, the correct choice for $\eta$ would still be $-\frac{3\pi}{8}\sqrt{n}$.
The proof is essentially the same as \cref{etaval}, the only difference is that the desired average value in the free locations instead of $-1$ would be $-\sqrt{2/\pi}$, the negative of $\E{|\xi|}$ for $\xi \sim \No(0, 1)$.

To sum up, our pick for matrix $Z$ is $Z = \sum_{k = 1}^n\alpha_kv_kv_k^T$, where for $1 \leq k \leq n/2$, $\alpha_k = \frac{3\pi}{8}\sqrt{n} - \lambda_k(A_G)$, and for $n/2 < k \leq n$, $\alpha_k = -\frac{3\pi}{8}\sqrt{n} - \lambda_k(A_G)$.
Such $Z$ has symmetric spectrum, as well as ensures that the average entry in non-free locations of $A_G + Y(Z)$ is $+1$, and average entry in free locations of $A_G + Y(Z)$ is $-1$.

We formulate a {conjecture} {of a nature} similar to \cref{claim:1.54}.
We do not have a full proof of this conjecture, but we do have supporting evidence that it is true.
\begin{conjecture}\label{claim:1.75}
    For a choice of $M = A_G + Y$, where $Y = \frac{1}{2}Z - \frac{1}{2}D_Z - \frac{1}{2}Z\circ A_G$ and $Z = \sum_{k = 1}^n\alpha_kv_kv_k^T$ with $\alpha_k = \frac{3\pi}{8}\sqrt{n} - \lambda_k(A_G)$ for $k \leq n/2$ and $\alpha_k = -\frac{3\pi}{8}\sqrt{n} - \lambda_k(A_G)$ for $k > n/2$, $\sigma_1(M) \leq 1.75\sqrt{n}$.
\end{conjecture}
For $M$ chosen as above, experimental results for matrices of sizes up to $n = 10000$ indeed show that its largest singular value is approximately $1.75\sqrt{n}$.
The proof approach similar to the one we used {in order to minimize the largest eigenvalue for \cref{claim:1.54}} suggests that \cref{claim:1.75} is true.
We shall give a quick overview of how the approach described {in \cref{sec:approach}} can be extended to upper bounding $\sigma_1(M)$ {for $M$ chosen as in \cref{claim:1.75}}.

\subsection{Analysis for spectral radius}

We present approximate analysis of the spectral radius of the matrix $A_G + \frac{1}{2}Z - \frac{1}{2}D_Z - \frac{1}{2}Z\circ A_G$ for the new choice of $Z$ in a similar fashion to the largest eigenvalue, as was done in \cref{sec:analysis}.
We scale all the matrices by $1/\sqrt{n}$, so that the eigenvalues of $Z$ and $A_G$ are of constant order.
First, consider the spectrum of $A_G + \frac{1}{2}Z$.
Denote $\lambda_i := \lambda_i(A_G)$, $i \in [n]$.
\begin{lemma}\label{specradqc}
    For all $i \in [n]$, $\lambda_i(A_G + \frac{1}{2}Z) = \frac{1}{2}\lambda_i+ (-1)^{i > n/2}\frac{3\pi}{16}\sqrt{n}$.
\end{lemma}
\begin{proof}
    Since matrices $A_G$ and $\frac{1}{2}Z$ have the same eigenvectors, for every $i \in [n]$ we have $\lambda_i(A_G + \frac{1}{2}Z) = \lambda_i + \frac{1}{2}\lambda_i(Z) = \lambda_i + \frac{1}{2}\alpha_i$.
    Consider $i \leq n / 2$, then $\alpha_i = \frac{3\pi}{8}\sqrt{n} - \lambda_i$, so
    \[\lambda_i(A_G + \frac{1}{2}Z) = \lambda_i+ \frac{3\pi}{16}\sqrt{n} - \frac{1}{2}\lambda_i = \frac{1}{2}\lambda_i+ \frac{3\pi}{16}\sqrt{n}.\]
    Similarly, for $i > n / 2$, $\alpha_i = -\frac{3\pi}{8}\sqrt{n} - \lambda_i$, and $\lambda_i(A_G + \frac{1}{2}Z) = \frac{1}{2}\lambda_i- \frac{3\pi}{16}\sqrt{n}$.
\end{proof}
As a result, the eigenvalues of $A_G + \frac{1}{2}Z$ are original eigenvalues of $A_G$ scaled by a factor of $\frac{1}{2}$ and symmetrically shifted away from $0$ by $\frac{3\pi}{16}\sqrt{n}$, i.e we added $\frac{3\pi}{16}\sqrt{n}$ to top $n / 2$ eigenvalues, and subtracted $\frac{3\pi}{16}\sqrt{n}$ from bottom $n/2$ (\cref{fig:qcz} below).
So, for every $i\neq  j \in [n]$,
\[\left(A_G + \frac{1}{2}Z\right)_{ij} = \frac{1}{2}\sum_{k = 1}^n\left[\lambda_k + (-1)^{k > n/2}\frac{3\pi}{8}\sqrt{n}\right]v_{ki}v_{kj}.\]
\begin{figure}[h]
\includegraphics[width=16cm]{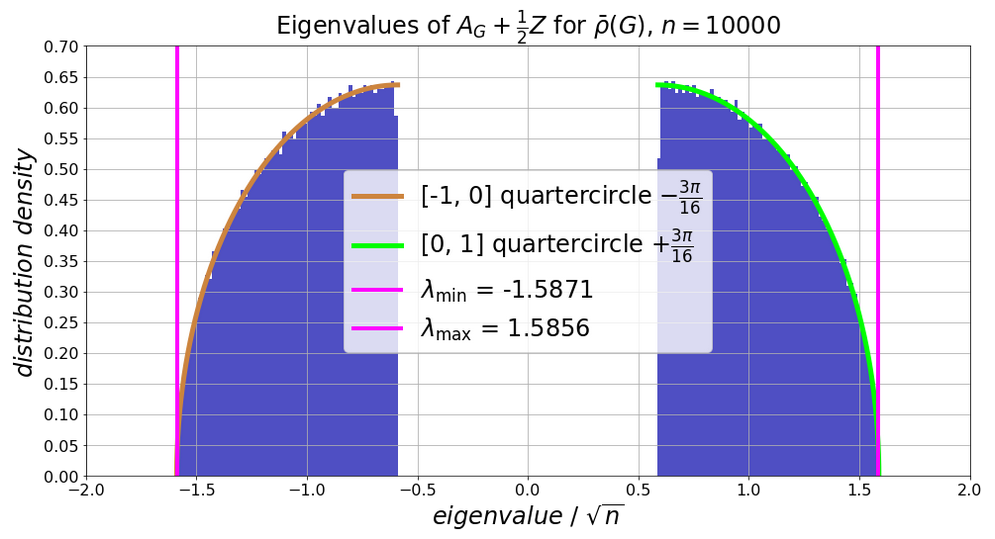}
\caption{Experimental distribution of eigenvalues of $A_G + \frac{1}{2}Z$ for $\bar{\rho}(G)$}
\label{fig:qcz}
\end{figure}
Since the eigenvalues $\lambda_1, \ldots, \lambda_n$ are distributed according to the standard semi-circle (\cref{semicircle}), \cref{specradqc} implies that the eigenvalues of $A_G + \frac{1}{2}Z$ will be distributed as two rescaled quartercircles shifted away from $0$ by $\frac{3\pi}{16}\sqrt{n}$, as with high probability, all the eigenvalues of $A_G + \frac{1}{2}Z$ will be close to their classical locations.

The following theorem is simply an extension of \cref{hxdist}, and is proved analogously. 
\begin{theorem}\label{hxdistrad}
    Let $\alpha, \beta, \gamma > 0$, let $\Rho_{\alpha, \beta, \gamma}$ be a probability distribution on $\R$ with density
    \begin{multline*}
        \rho_{\alpha, \beta, \gamma}(x) = \d\Rho_{\alpha, \beta, \gamma}(x):= \frac{1}{2\alpha^2\pi}\sqrt{4\alpha^2 - x^2}\I{x \in [-2\alpha, -\gamma]}\d x \\+ \frac{1}{2\beta^2\pi}\sqrt{4\beta^2 - x^2}\I{x \in [\gamma, 2\beta]}\d x.
    \end{multline*}
   For matrix $A_G \sim \Xx_n$, where $A_G = \sum_{i = 1}^n\lambda_iv_iv_i^T$ is a spectral decomposition of $A_G$, and $Z = \sum_{i = 1}^n\left[(-1)^{i > n/2}\frac{3\pi}{8}\sqrt{n} - \lambda_i(A_G)\right]v_iv_i^T$, let $L_{A_G + \frac{1}{2}Z}(x)$ be the empirical spectral distribution measure of $(A_G + (1/2)Z)/\sqrt{n}$.
   Then, $L_{A_G + (1/2)Z} \xrightarrow[n\to\infty]{\mathrm{a.s.}} \Rho_{\frac{1}{2}, \frac{1}{2}, \frac{3\pi}{16}}$.
\end{theorem}
As we mentioned earlier, different to \cref{sec:analysis}, our choice of $Z$ makes the spectrum of $A_G + (1/2)Z$ symmetric.
This, together with the symmetricity of {the spectrum of} $A_G\circ Z$, will help us to ensure that the top and the bottom eigenvalues of $A_G + Y$ are close, thus giving us a bound on the spectral radius.

Next, we determine the limiting distribution of the eigenvalues of the diagonal matrix $D_Z = \Diag(Z)$ for our choice of $Z$.
By definition of $Z$, for every $k \in [n]$,
\[Z_{kk} = \sum_{i = 1}^n\alpha_iv_{ki}^2 = \sum_{i = 1}^n(-1)^{i > n/2}\left(\frac{3\pi}{8}\sqrt{n}\right)v_{ki}^2 - \sum_{i = 1}^n\lambda_iv_{ki}^2.\]
The expression $\sum_{i = 1}^n\lambda_iv_{ki}^2$ is exactly $(A_G)_{kk}$, the $k$-th diagonal element of $A_G$, which is equal to $0$ by definition of $A_G$.
So,
\[Z_{kk} = \frac{3\pi}{8}\sqrt{n}\left(\sum_{i = 1}^{n/2}v_{ki}^2 - \sum_{i = n / 2 + 1}^nv_{ki}^2\right).\]
Due to the symmetry of the distribution $\Xx_n$, for all $i, j \in [n]$, $\E{(D_Z)_{ii}} = \E{(D_Z)_{jj}}$, so if we can guarantee that the diagonal entries of $Z$ are concentrated, then matrix $D_Z$ is simply a scalar matrix, and it is easy to analyze its affect on the spectrum of the overall sum in $Y = Y(Z)$.

Due to the symmetry of the distribution $\Xx_n$, for $A_G \sim \Xx_n$ with eigenvectors $v_1,\ldots, v_n$ it holds that $\sum_{i = 1}^{n/2}v_{ki}^2$ and $\sum_{i = n / 2 + 1}^nv_{ki}^2$ have the same distribution, hence $\E{Z_{kk}} = 0$ and it remains to show that with high probability $Z_{kk}$ is concentrated around $0$.
We cannot show simultaneous concentration of all diagonal elements when $\Xx_n$ is a distribution over arbitrary Wigner matrices, but we can prove that the desired simultaneous concentration does hold when $\Xx_n = \GOE(n)$.
\begin{theorem}\label{specdiagconc}
    Let $A_G \sim \GOE(n)$, where $A_G = \sum_{i = 1}^n\lambda_iv_iv_i^T$ is a spectral decomposition of $A_G$, and $Z = \sum_{i = 1}^n\left[(-1)^{i > n/2}\frac{3\pi}{8}\sqrt{n} - \lambda_i\right]v_iv_i^T$.
    For every $k \in [n]$, for any constant $K > 0$,
    \[\P{\left|Z_{kk}\right| \geq n^{1/6}} = O(n^{-K}).\]
\end{theorem}
The proof of \cref{specdiagconc} directly follows from the proof of \cref{diagconcbase} (and \cref{diagconc}), as we consider the sums $V^+_k := \frac{3\pi}{8}\sum_{i = 1}^{n/2}v_{ki}^2$ and $V^-_k := \frac{3\pi}{8}\sum_{i = n / 2 + 1}^nv_{ki}^2$ separately and prove that each one of them individually is concentrated around the same value.
The theorem follows by $Z_{kk} = V^+_k - V^-_k$.
{As was the case with \cref{diagconcbase} and \cref{assumdiag}, based on \cref{specdiagconc} we formulate a similar assumption in the spectral radius case, i.e that the diagonal entries of $Z$ are concentrated around $0$ also when $A_G \sim \Bern(n, 1/2)$.}

\subsection{Spectrum of $Z\circ A_G$}

Denote $W := Z\circ A_G$.
By definition of matrix $Z$, the entries of $W$ are as follows:
\[\forall i \neq j \in [n], \quad W_{ii} = 0; \quad\text{and}\quad W_{ij} = \begin{cases}
    \sum_{k = 1}^n\left[(-1)^{k > n/2}\frac{3\pi}{8}\sqrt{n} - \lambda_k\right]v_{ki}v_{kj},& (A_G)_{ij} \geq 0;\\
    -\sum_{k = 1}^n\left[(-1)^{k > n/2}\frac{3\pi}{8}\sqrt{n} - \lambda_k\right]v_{ki}v_{kj},& (A_G)_{ij} < 0.
\end{cases}\]
It is easy to see that when $(A_G)_{ij} \geq 0$, the $(i, j)$-th entry in $A_G + \frac{1}{2}Z - \frac{1}{2}W$ is equal to $\sum_{k = 1}^n\lambda_kv_{ki}v_{kj} = (A_G)_{ij} = 1$, and when $(A_G)_{ij} < 0$, the $(i, j)$-th entry in $A_G + \frac{1}{2}Z + \frac{1}{2}W$ is equal to $\sum_{k = 1}^n(-1)^{k > n/2}(\frac{3\pi}{8}\sqrt{n})v_{ki}v_{kj}$, which on average is equal to $-1$ by \cref{etaval}.
So, the average entries in free and non-free entries of $W$ correspond to those of $A_G$ respectively.
This gives hope that, as was the case with $\wX$ before in \cref{sec:analysis}, that the eigenvalues of $W$ also have a semicircular distribution, with some different parameter.
{In support of this, as seen in numerical experiments (\cref{fig:wspec}), the spectrum of $W$ resembles a semi-circle law on $[-2\tau, 2\tau]$ for constant $\tau = \sqrt{\frac{9\pi^2}{64} - 1}$ (the exact value of $\tau$ will be explained below).}
\begin{figure}[h]
\includegraphics[width=16cm]{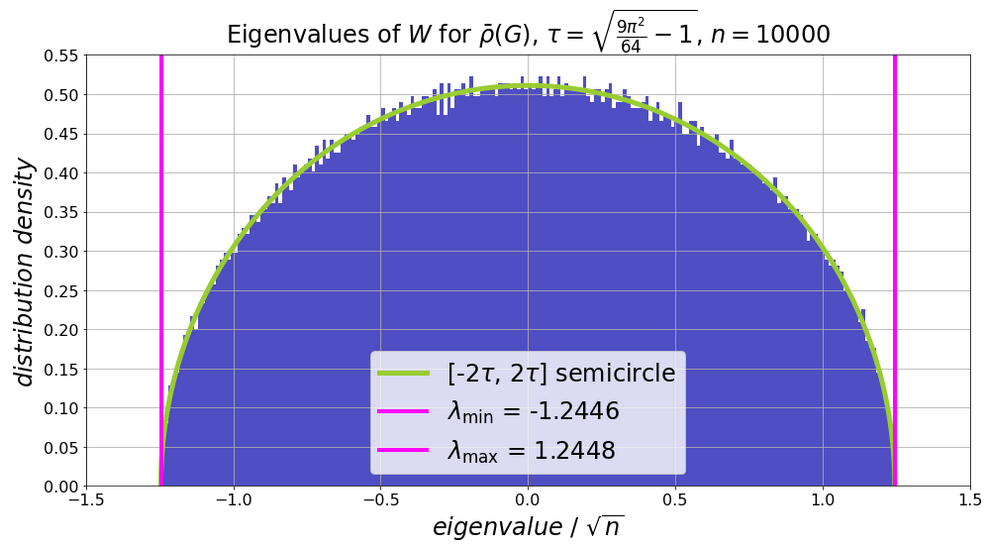}
\caption{Experimental distribution of eigenvalues of $W$ for $\bar{\rho}(G)$}
\label{fig:wspec}
\end{figure}

Similar to what was done in \cref{sec:analysis}, {instead of analyzing the the spectrum of $W = Z\circ A_G$, we} consider a matrix $\wW := Z\circ A_{G'}$ for a new random graph $G' \sim G(n, 1/2)$ sampled independently.
We define a distribution $\wWw_n := \wWw_n(A_G)$ over symmetric matrices as follows.
For $\wW \sim \wWw_n$, for all $i \in [n]$, $\wW_{ii} = 0$, and $\forall i < j \in [n]$:
\[\wW_{ij} = \wW_{ji} = \begin{cases}
    \sum_{k = 1}^n\left[(-1)^{k > n/2}\frac{3\pi}{8}\sqrt{n} - \lambda_k\right]v_{ki}v_{kj},& \wpr 1/2;\\
    -\sum_{k = 1}^n\left[(-1)^{k > n/2}\frac{3\pi}{8}\sqrt{n} - \lambda_k\right]v_{ki}v_{kj},& \wpr 1/2;
\end{cases}\; \textit{independently}.\]
$\wW \sim \wWw_n$ is a generalized Wigner matrix, with additional randomness independent of the randomness in $W = Z\circ A_G$.
{Observe that $\sum_{i, j = 1}^n\E{\wW_{ij}^2}$ is equal to $\sum_{i, j = 1}^n\E{Z_{ij}^2} - \sum_{i = 1}^n\E{Z_{ii}^2}$ where matrix $Z = \sum_{k = 1}^n\left[(-1)^{k > n/2}\frac{3\pi}{8}\sqrt{n} - \lambda_k\right]v_{k}v_{k}^T$ as defined before.
By an earlier assumption (supported by \cref{diagconc2}), $\sum_{i = 1}^n\E{Z_{ii}^2} = n\cdot O(n^{1/6}) = O(n^{7/6})$ with high probability.
Next, $\sum_{i, j = 1}^n\E{Z_{ij}^2}$ is equal to the sum of squares of eigenvalues of $Z$, so $\sum_{i, j = 1}^nZ_{ij}^2 = \sum_{k = 1}^n[(-1)^{k > n/2}\frac{3\pi}{8}\sqrt{n} - \lambda_k]^2 = n\sum_{k = 1}^n\frac{9\pi^2}{64} + \sum_{k = 1}^n\lambda_k^2 - 2\frac{3\pi}{8}\sqrt{n}\sum_{k = 1}^n(-1)^{k > n/2}\lambda_k$.
By \cref{PlusEigBound}, $\sum_{k = 1}^n\lambda_k^2 = (1 + o(1))n^2$ and $\sum_{k = 1}^n(-1)^{k > n/2}\lambda_k = (\frac{8}{3\pi }+ o(1))n^{3/2}$ with high probability, hence the average variance of an non-diagonal entry of $\wW$ is roughly $\frac{1}{n(n - 1)}\sum_{i, j = 1}^nZ_{ij}^2 - \frac{1}{n(n - 1)}\sum_{i = 1}^nZ_{ii}^2 \simeq \frac{9\pi^2}{64} + 1 - 2\frac{3\pi}{8}\cdot \frac{8}{3\pi} = \frac{9\pi^2}{64} - 1$.
This is exactly the value of the scaling parameter $\tau$.
We make an assumption of the same nature as \cref{wxdist} for the matrices $W$ and $\wW\sim \wWw_n$, that the limiting distributions of eigenvalues of $W$ and $\wW$ are the same.}

If matrix $\wW$ satisfies the requirements of \cref{gensemicirc}, i.e the Lindeberg condition holds and the sums of variances in each row are close to each other, then its eigenvalues have semicircular distribution (with parameter $\tau = \sqrt{\frac{9\pi^2}{64} - 1}$).
As before, when $A_G \sim \GOE(n)$ we can indeed show that the conditions of \cref{gensemicirc} are satisfied, and thus $\wW \sim \wWw_n$ has a semicircular spectrum with parameter $\tau$.
Concretely, the following theorem holds, and its proof is completely analogous to the proof of \cref{genwigmain}.
\begin{theorem}\label{genwigspec}
    Let $A_G \sim \GOE(n)$, with eigendecomposition $A_G = \sum_{k = 1}^n\lambda_kv_kv_k^T$, and let $\wW \sim \wWw_n(A_G)$ be as defined above.
    Let $L_{\wW}(x) := \frac{1}{n}\sum_{i = 1}^n\I{\lambda_{i}(\wW) \leq x}$ be the empirical spectral distribution measure of the eigenvalues of the matrix $\wW$.
    With probability at least $1 - O(1/n^8)$ over the choice of $A_G$, $L_{\wW} \xrightarrow[n\to\infty]{\mathrm{a.s.}} \Rho_{\tau}$ with $\tau = \sqrt{\frac{9\pi^2}{64} - 1}$, i.e the distribution of the eigenvalues of $\wW$ converges almost surely to a semi-circle law $\Rho_\tau$ with parameter $\tau= \sqrt{\frac{9\pi^2}{64} - 1}$.
\end{theorem}
When $A_G$ is a $\pm 1$ random matrix, as was the case in \cref{sec:analysis}, it is not clear whether the statement of \cref{genwigspec} holds (the proof relies on the concentration properties of the eigenvectors of $A_G$, which do not always hold for non-Gaussian random matrices).
However, experimental observations, as well as general universality properties of Wigner matrices (\cite{ER11}), suggest that the statement should hold in the $\pm 1$ case as well.
{This motivates us to propose an assumption similar in nature to \cref{genwigbern}, i.e that the eigenvalues of $\wW \sim \wWw_n(A_G)$ are distributed according to $\Rho_\tau$ also when $A_G \sim \Bern(n, 1/2)$.}

\subsection{Free convolution for spectral radius}

The final part of the analysis considers the spectrum of the sum of $(A_G + \frac{1}{2}Z) + \frac{1}{2}W$ where $W = -Z\circ A_G$.
The experimental results suggest that the spectrum of the sum of these two matrices is actually the free convolution of the {spectra} of the summands, i.e a free convolution of two spread scaled quartercircles with a semicircle.
The algorithm also adds the matrix $-\frac{1}{2}D_Z$, zeroing out the diagonal, however, due to the symmetricity of the spectrum of $Z$ we expect the addition of $-\frac{1}{2}D_Z$ to make negligible changes (as $D_Z$ is essentially a $0$-matrix).
As can be seen in the \cref{fig:freesumwspec}, {experimental results suggest that} the largest singular value of the sum (with added diagonal matrix $-D_Z$) is at most $1.75\sqrt{n}$, as conjectured in \cref{claim:1.75}. 
\begin{figure}[h]
\includegraphics[width=16cm]{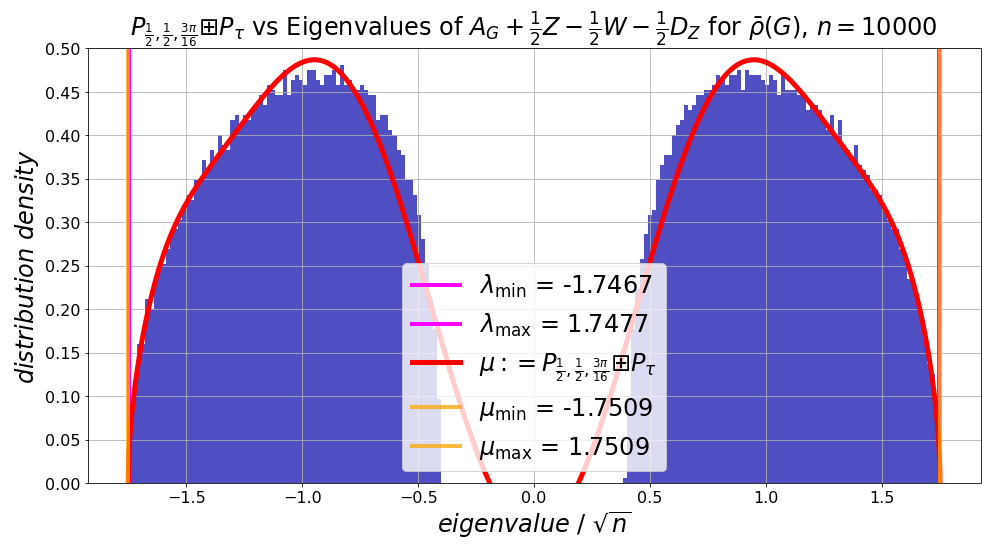}
\caption{Comparison of numerically computed support of $\Rho_{\frac{1}{2}, \frac{1}{2}, \frac{3\pi}{16}} \boxplus \Rho_\tau$ with experimental distribution of eigenvalues of $(A_G + \frac{1}{2}Z) + \frac{1}{2}W - \frac{1}{2}D_Z$}
\label{fig:freesumwspec}
\end{figure}
In addition, we use the algorithm by \cite{CY23} in order to numerically compute the support in \cref{freealgo} of the free convolution of the predicted spectrums of $(A_G + \frac{1}{2}Z)$ and $W = -Z\circ A_G$.
That is, we compute the free convolution of two spread scaled quartercircles and a semicircle with parameter $\tau = \sqrt{\frac{9\pi^2}{64} - 1}$, note that these distributions do not depend on any particular realization of corresponding matrices.
The results are analogous to what has been observed in \cref{sec:free} --- the support bounds computed via \cref{freealgo} corresponds to the largest and smallest eigenvalues of the sum $A_G + \frac{1}{2}Z - D_Z + \frac{1}{2}W$ for our choice of $Z$, and give a bound of $1.75\sqrt{n}$ on the spectral radius (which can also be seen in \cref{fig:freesumwspec}).   
{Having computed the support of the distribution, the algorithm also tries to approximately compute the density of the free convolution and draw its curve.
Due to the nature of the algorithm, the drawn density curve of $\Rho_{\frac{1}{2}, \frac{1}{2}, \frac{3\pi}{16}} \boxplus \Rho_\tau$ does not entirely match with the eigenvalue distribution of the matrix $A_G + \frac{1}{2}Z - \frac{1}{2}D_Z + \frac{1}{2}W$.
As one can observe in \cref{fig:freesumwspec}, the final matrix has no eigenvalues on a segment around $0$ (so the density on this segment in the middle should take value $0$), which is a direct consequence of the distribution of $A_G + \frac{1}{2}Z$ being two spread quartercircles.
{However}, the algorithm represents Cauchy and $R$ transforms as power series and then tries to approximate the free convolution density as a polynomial.
But it is well-known that any non-constant polynomial cannot contain a constant interval, thus the algorithm assumes that the free convolution density takes other, non-$0$ values on the aforementioned segment.
As a result, while the support of $\Rho_{\frac{1}{2}, \frac{1}{2}, \frac{3\pi}{16}} \boxplus \Rho_\tau$ is computed correctly, the middle part of the density curve is not accurate, though we still observe a partial match with the eigenvalues' distribution.}

The same holds when we replace matrix $W$ with a matrix $\wW \sim \wWw_n$, allowing signs to be random.
As was the case with the theta function analysis, it supports the idea that not only $W$ itself behaves like a generalized Wigner matrix, but also that $(A_G + \frac{1}{2}Z)$ and $W = -Z\circ A_G$ are free and satisfy \cref{freeconvspec} (and even \cref{freestick}).
{Based on these observations, we propose two more assumptions regarding matrices $A_G + \frac{1}{2}Z$ and $W = -Z\circ A_G$, of the same nature as \cref{wxsums} and \cref{freenesswx}.
That is, not only matrices $W$ and $\Ww\sim \wWw_n$ have the same limiting distribution of eigenvalues, but also their spectra behave similarly when summed up with matrix $A_G + \frac{1}{2}Z$.
And, on top of that, matrices $A_G + \frac{1}{2}Z$ and $\wW \sim \wWw_n$ are asymptotically free, so the spectrum of their sum is the free convolution of the individual spectra of the summands, and that the top and bottom eigenvalues of $A_G + \frac{1}{2}Z - \frac{1}{2}W$ stick to the corresponding ends of the support of the free convolution.
These assumptions are supported by already mentioned experimental evidence, and also by the same reasoning as we used while discussing \cref{wxsums} and \cref{freenesswx}, since we apply essentially the same algorithm and proof structure as we did for largest eigenvalue minimization in \cref{sec:analysis} and \cref{sec:free}.}

Unfortunately, as was described in the in {\cref{sec:free}}, it is not known whether generalized Wigner matrices are free with deterministic matrices.
So, even if we replace $W$ with $\wW\sim \wWw_n$ and consider the Gaussian case, it still seems difficult to prove that $\wW$ is free with $A_G + \frac{1}{2}Z$.
While the algorithm presented in \cref{sec:approach} and its analysis from \cref{sec:analysis} do extend to the spectral radius problem, at the same time we run into similar issues that prevented us from getting complete proofs when bounding the largest eigenvalue of the sum.

\newpage

\section{Lower Bound on $\orho(G)$ for $G\sim G(n, 1/2)$}
\label{sec:radius}

Recall that $\orho(G)$ for $n$-vertex graph $G = (V, E)$ can be formulated as the value of the optimal solution to the following SDP (here $\sigma_1(M)$ denotes the largest singular value of matrix $M$):
\[\begin{aligned}
    &\minimize_{M \in \S^n} && \sigma_{1}(M)\\
    &\subto&& M_{ii} = 0,&\forall i \in V;\\
    &&& M_{ij} = 1,&\forall (i, j) \in E.
\end{aligned}\]
\begin{theorem}[\cref{thm:radius} restated]\label{specthm}
    Let $G \sim G(n, 1/2)$.
    With probability at least $1 - O(n^{-1}\polylog n)$ it holds that $\orho(G) \geq \left(3\pi/8 - \o{1}\right)\sqrt{n}$.
\end{theorem}

In our proof of \cref{specthm}, we shall use the notion of \textbf{typical} graphs.
\begin{definition}\label{typical}
    For a graph $G = (V, E)$, $V = [n]$, let $A_G$ be the $\pm 1$ adjacency matrix of $G$.
    We say that $G$ is \textbf{typical}, if the following holds:
    \begin{enumerate}
        \item $\left|\sum_{(i, j) \notin E}(A_G)_{ij} + \sum_{(i, j) \in E}(A_G)_{ij}\right| = \left|\sum_{i, j = 1}^n(A_G)_{ij}\right| \leq O(n^{3/2})$;
        \item $\sigma_1(A_G) \leq (2 + o(1))\sqrt{n}$;
        \item if $A_G = \sum_{k = 1}^n\lambda_ku_ku_k^T$ is an eigendecomposition of $A_G$, and $\lambda_1 \geq \ldots \geq \lambda_n$, then
        \[\sum_{k : \lambda_k \geq 0}\lambda_k = \frac{4}{3\pi}n^{3/2} \pm O(n^{1/2}\polylog n) \quad\text{and}\quad \sum_{k : \lambda_k \leq 0}\lambda_k = -\frac{4}{3\pi}n^{3/2} \pm O(n^{1/2}\polylog n).\]
    \end{enumerate}
\end{definition}
\begin{proposition}
    With probability at least $1 - O(n^{-1}\polylog n)$, graph $G \sim G(n, 1/2)$ is typical.
\end{proposition}
\begin{proof}
    For the first condition, note that $\sum_{i, j = 1}^n(A_G)_{ij} = |E| - |\oE|$.
    Since $G \sim G(n, 1/2)$, $\E{|E|} = \E{|\oE|} = \frac{n(n - 1)}{2}$.
    By Chernoff, $\P{|E| - \frac{n^2 - n}{2} > cn\log n} = \P{|E| - \E{|E|} > cn\log n} \leq O(n^{-c})$ for a constant $c > 0$, and the same holds for $|\oE|$.
    Then, with probability at least $1 - O(n^{-c})$, $||E| - |\oE|| \leq O(n\log n)$.
    Next, by \cite{FuKo}, $\sigma_1(A_G) \leq (2 + o(1))\sqrt{n}$ for $G \sim G(n, 1/2)$ with probability at least $1 - \exp(-n^{\Omega(1)})$, proving the second condition.
    And, by \cref{PlusEigBound}, the third condition is satisfied for $G \sim G(n, 1/2)$ with probability at least $1 - O(n^{-1}\polylog n)$.
\end{proof}
\begin{lemma}\label{Mbound}
    Let $G = (V, E)$, $V = [n]$, be typical.
    Let $M \in \S^n$ be a matrix such that for all $i \in [n]$, $M_{ii} = 0$, $\sum_{(i, j)\in E}M_{ij} = |E|$, and $\sigma_1(M)$ is minimal among all such matrices.
    Then $\sigma_1(M) \leq (2 + o(1))\sqrt{n}$ and $\left|\sum_{(i, j)\notin E}M_{ij} - \sum_{(i, j)\notin E}(A_G)_{ij}\right| \leq O(n^{3/2})$.
\end{lemma}
\begin{proof}
    By definition of $M$, $\sum_{(i, j)\in E}M_{ij} = |E| = \sum_{(i, j)\in E}(A_G)_{ij}$.
    On the other hand, observe that
    \[\left|\sum_{(i, j)\in E}M_{ij} + \sum_{(i, j)\notin E}M_{ij}\right| = \left|\sum_{i, j = 1}^nM_{ij}\right| = |\One^TM\One| = n\left|\frac{\One^TM\One}{\|\One\|^2}\right| \leq n \cdot \max_{x\in \R^n}\left|\frac{x^TMx}{\|x\|^2}\right| \leq n\sigma_1(M).\]
    Since for all $i \in [n]$, $(A_G)_{ii} = 0$, and for all $(i, j) \in E$, $(A_G)_{ij} = 1$, it follows by minimality of $M$ that $\sigma_1(M) \leq \sigma_1(A_G)$.
    But then, since $G$ is typical, $\sigma_1(M)\leq \sigma_1(A_G) \leq (2 + o(1))\sqrt{n}$ and 
    \begin{multline*}
        \left|\sum_{(i, j)\notin E}M_{ij} - \sum_{(i, j)\notin E}(A_G)_{ij}\right| 
        \leq \left|\sum_{(i, j)\notin E}M_{ij} + \sum_{(i, j)\in E}M_{ij}\right| + \left|\sum_{(i, j)\in E}M_{ij} + \sum_{(i, j)\notin E}(A_G)_{ij}\right| \\
        \leq n\sigma_1(M) + \left|\sum_{(i, j)\in E}(A_G)_{ij} + \sum_{(i, j)\notin E}(A_G)_{ij}\right| \leq O(n^{3/2}).
    \end{multline*}
    The last summand is at most $O(n^{3/2})$ by the first condition of \cref{typical}.
\end{proof}
{\begin{theorem}
    If graph $G$ is typical, then $\orho(G) \geq \left(3\pi/8 - o(1)\right)\sqrt{n}$.
\end{theorem}
\begin{proof}
    Let $M$ be the matrix from \cref{Mbound}.
    {Most part of the proof will be devoted to showing that for a typical $G$ the following inequality holds:}
    \begin{equation}
    \label{eq:lowerBound}
        \sigma_1(M)\sum_{k = 1}^n|\lambda_k(A_G)| \geq n^2 - n - O(n^{3/2}) 
    \end{equation}
    To see this, consider a matrix $Q := M - A_G$.
    By definition of Frobenius norm,
    \[\|M\|_F^2 = \|A_G + Q\|_F^2 = \tr[(A_G + Q)^2] = \|A_G\|_F^2 + \|Q\|_F^2 + 2\langle A_G, Q\rangle.\]
    Observe that 
    \[\langle A_G, Q\rangle = \tr[A_GQ] = \sum_{i,j}(A_G)_{ij}Q_{ij} = \pm O(n^{3/2}).\]
    Indeed, since $Q = M - A_G$ and $\sum_{(i, j)\in E}M_{ij} =|E| = \sum_{(i, j)\in E}(A_G)_{ij}$ by definition of $M$, it holds $\sum_{(i, j)\in E}Q_{ij} = 0$.
    Therefore, $\sum_{i, j}(A_G)_{ij}Q_{ij} = \sum_{(i, j) \in E}(A_G)_{ij}Q_{ij} + \sum_{(i, j) \notin E}(A_G)_{ij}Q_{ij} = -\sum_{(i, j)\notin E}Q_{ij}$.
    But then, since $G$ is typical, by \cref{Mbound}
    \[\sum_{(i, j)\notin E}Q_{ij} = \sum_{(i, j)\notin E}M_{ij} - \sum_{(i, j)\notin E}(A_G)_{ij} = \pm O(n^{3/2}).\]
    As a corollary, using the property of Frobenius norm $\|H\|_F^2 = \sum_{k = 1}^n\lambda_k(H)^2$,
    \[\sum_{k = 1}^n\lambda_k(M)^2 = \sum_{k = 1}^n\lambda_k(A_G)^2 + \sum_{k = 1}^n\lambda_k(Q)^2 \pm O(n^{3/2}).\]
    We shall use the following classic result on eigenvalues of matrix sums.
    \begin{theorem}[Hoffman-Wielandt, \cite{HW53}]\label{HW}
    If $X$ and $Y$ are size-$n$ symmetric matrices, and $Z = X - Y$, then
    \[\sum_{i = 1}^n\lambda_i(Z)^2 \geq \sum_{i = 1}^n\big(\lambda_i(X) - \lambda_i(Y)\big)^2.\]
    \end{theorem}
    Applying \cref{HW} with $X := M$, $Y := A_G$ and $Z := Q = M - A_G$, we get that 
    \begin{multline*}
        \sum_{k = 1}^n\lambda_k(M)^2 \geq \sum_{k = 1}^n\lambda_k(A_G)^2 + \sum_{k = 1}^n[\lambda_k(M) - \lambda_k(A_G)]^2 - O(n^{3/2})\\
        \implies 
        \sum_{k = 1}^n\lambda_k(M)^2 \geq 2\sum_{k = 1}^n\lambda_k(A_G)^2 + \sum_{k = 1}^n\lambda_k(M)^2 - \sum_{k = 1}^n\lambda_k(M)\lambda_k(A_G) - O(n^{3/2})\\
        \implies \sum_{k = 1}^n\lambda_k(M)\lambda_k(A_G)  \geq \sum_{k = 1}^n\lambda_k(A_G)^2 - O(n^{3/2}).
    \end{multline*}
    Since $A_G$ is a $\pm 1$ matrix with $0$ diagonal, $\sum_{k = 1}^n\lambda_k(A_G)^2 = \tr[A_G^2] = \sum_{i, j = 1}^n(A_G)_{ij}^2 = n^2 - n$, hence $\sum_{k = 1}^n\lambda_k(M)\lambda_k(A_G)  \geq n^2 - n - O(n^{3/2})$.
    On the other hand,
    \[\sum_{k = 1}^n\lambda_k(M)\lambda_k(A_G) \leq \sum_{k = 1}^n|\lambda_k(M)\lambda_k(A_G)| \leq \sum_{k = 1}^n\left(\max_{j \in [n]}|\lambda_j(M)|\right)\cdot |\lambda_k(A_G)| = \sigma_1(M)\sum_{k = 1}^n|\lambda_k(A_G)|.\]
    {This established inequality~(\ref{eq:lowerBound}).}

    Since $G$ is typical, it holds that $\sum_{k = 1}^n|\lambda_k(A_G)| = \sum_{k: \lambda_k(A_G) \geq 0}\lambda_k(A_G) - \sum_{k: \lambda_k(A_G) \leq 0}\lambda_k(A_G)$ is equal to $\frac{8}{3\pi}n^{3/2} \pm O(n^{1/2}\polylog n)$.
    Then, applying the lower bound above,
    \[n^2 - O(n^{3/2}) \leq \sigma_1(M)\sum_{k = 1}^n|\lambda_k(A_G)| \leq \sigma_1(M)\left(\frac{8}{3\pi} + o(1)\right)n^{3/2}.\]
    Dividing both sides by $\left(\frac{8}{3\pi} + o(1)\right)n^{3/2}$ gives us $\sigma_1(M) \geq \left(\frac{3\pi}{8} - o(1)\right)\sqrt{n} = 1.178\sqrt{n}$.
\end{proof}

}
\newpage
\section{Concentration of $\Diag(X^- - X^+)$ for Gaussian $X$}


\begin{theorem}[\cref{diagconcbase} restated]\label{diagconc}
    Let $Y \in \S^n$ be symmetric random matrix with i.i.d entries, such that for every $i \leq j \in [n]$, $Y_{ii} \sim \No(0, \frac{2}{n})$ and $Y_{ij} = Y_{ji} \sim \No(0, \frac{1}{n})$.
    Let $Y = \sum_{i = 1}^n\lambda_iu_iu_i^T$ be the eigendecomposition of $Y$, and let $Y^+ = \sum_{i = 1}^{n/2}\lambda_iu_iu_i^T$, $Y^- = \sum_{i = n/2 + 1}^n\lambda_iu_iu_i^T$.
    For every $k \in [n]$, for every $K > 0$, $\E{Y^+_{kk}} = \frac{4}{3\pi} + o(1)$, $\E{Y^-_{kk}} = -\frac{4}{3\pi} + o(1)$, and 
    \[\P{\left|Y^+_{kk} - \E{Y^+_{kk}}\right| \geq n^{-1/3}} = O(1/n^K);\qquad\qquad\P{\left|Y^-_{kk} - \E{Y_{kk}^-}\right| \geq n^{-1/3}} = O(1/n^K). \]
\end{theorem}
\begin{proof}
    To compute $\E{Y^+_{kk}}$ and $\E{Y^-_{kk}}$, we rely on the fact that $Y \sim \GOE(n)$, so the eigenvalues and eigenvectors of $Y$ are independent.
    More specifically, the following theorem holds.
    \begin{theorem}[\cite{AGZ09}]\label{OrtDist}
    Let $X \sim \GOE(n)$, and let $X = U\Lambda U^T$ be the eigendecomposition of $X$, $\Lambda = \diag(\lambda_1,\ldots, \lambda_n)$, $U = [u_1,\ldots, u_n]$ where for every $i \in [n]$ we have $u_{i1} > 0$.
    Then:
    \begin{enumerate}
        \item The collection $[u_1,\ldots, u_n]$ is independent of the eigenvalues $\lambda_1,\ldots, \lambda_n$, and each of the eigenvectors $u_1,\ldots, u_n$ is distributed uniformly on $S_+^{n - 1} = \{x\in  \R^n : \|x\| = 1, x_1 > 0\}$.
        \item $U = [u_1,\ldots, u_n]$ is distributed according to Haar measure on $\Ort(n)$ (orthogonal matrices of size $n$), with each $u_i$ multiplied by a $\pm 1$-scalar so that all columns of $U$ belong to $S_+^{n - 1}$.
    \end{enumerate}
    \end{theorem}
    Applying \cref{OrtDist} to $Y$, it is easy to see that by independence of $\lambda_i$-s and $u_i$-s
    \[\E{Y_{kk}^+} = \sum_{i = 1}^{n/2}\E[\lambda]{\lambda_i}\E[U]{u_{ki}^2} = \frac{1}{n} \sum_{i = 1}^{n/2}\E[\lambda]{\lambda_i} = \frac{1}{n}\E{\tr(Y^+)}= \frac{4}{3\pi} + o(1),\]
    where the last equality holds by \cref{PlusEigBound}.
    The proof for $\E{Y^-_{kk}} = -\frac{4}{3\pi} + o(1)$ is analogous.
    \begin{definition}
    For $i \in [n]$, denote by $\gamma_i$ the \textbf{classical location} of the $i$-th eigenvalue under the semicircle law.
    Specifically, $\Rho(\gamma_i) = \int_{-\infty}^{\gamma_i}\rho(x)\d x = \frac{i}{n}$.
\end{definition}
\begin{theorem}[\cite{KY13}]\label{rigeig}
    There exist positive constants $A > 1, c, C$ and $\tau < 1$ such that
    \[\P{\exists i : |\lambda_i - \gamma_i| \geq \frac{(\log n)^{A\log\log n}}{n^{2/3}}} \leq \exp\left(-c\left(\log n\right)^{\tau A\log\log n}\right),\]
    and
    \[\P{\exists i, k : |u_{ik}|^2 \geq \frac{(\log n)^{C\log\log n}}{n}} \leq \exp\left(-c\left(\log n\right)^{\tau A\log\log n}\right).\]
\end{theorem}
For $k \in [n]$, consider the $k$-th diagonal entry $Y^+_{kk} = \sum_{i = 1}^{n/2}\lambda_iu_{ik}^2$.
Note that $|Y^+_{kk}| \leq \sum_{i = 1}^{n/2}|\lambda_i|\cdot |u_{ik}^2| \leq (2 + o(1))\sum_{i = 1}^{n/2}\|u_i\|^2 \leq (1 + o(1))n$, as by \cite{BVH16} it holds that $|\lambda_1|\leq 2 + o(1)$ with probability $1 - O(n^{-2K})$ for any $K > 0$.
For $i, j, k \in [n]$, define the following polynomial $P_\gamma^+(i, j, k)$ in $n$ variables: $P_\gamma^+(i, j, k) := \sum_{i = 1}^{n/2}\gamma_iu_{ij}u_{ik}$, where $\gamma_i$ is the classical location of the $i$-th eigenvalue under the semicircle law and $u_i$ is the $i$-th eigenvector of $Y$.
Since by $\gamma_i \in [-2, 2]$ for all $i \in [n]$, $|P_\gamma^+(i, j, k)| \leq \sum_{i = 1}^{n/2}|\gamma_i|\cdot |u_{ij}u_{ik}| \leq 2\sum_{i = 1}^{n/2}\|u_i\|^2 \leq n$.

Note that $Y^+_{kk}$ is $P_\gamma^+(i, k, k) =: P_\gamma^+$ with $\gamma$-s swapped for $\lambda$-s.
To prove concentration of $Y_{kk}^+$ around $\E{Y_{kk}^+}$, we will show that 1) $\E{Y_{kk}^+}$ is close to $\E{P_\gamma^+}$ with high probability, 2) $Y_{kk}^+$ is close to $P_\gamma^+$ with high probability, and 3) $P_\gamma^+$ is concentrated around $\E{P_\gamma^+}$ with high probability.

Let $E$ be the event $E = E_\lambda \cap E_u$, where the events $E_\lambda$ and $E_u$ are defined as follows: 
\[E_\lambda = \left\{\forall i : |\lambda_i - \gamma_i| < \frac{(\log n)^{A\log\log n}}{n^{2/3}}\right\},\qquad E_u = \left\{\forall i, k : u_{ik}^2 < \frac{(\log n)^{C\log\log n}}{n}\right\}.\]
By \cref{rigeig}, $\P{\overline{E}} = \P{\overline{E_\lambda \cap E_u}} = \P{\overline{E_\lambda} \cup \overline{E_u}} \leq 2\exp\left(-c\left(\log n\right)^{\tau A\log\log n}\right)$.
\begin{lemma}\label{expclosediag}
    For every $k \in [n]$ there exist a constant $A > 1$ such that for all $K > 0$, $\left|\E{Y_{kk}^+} - \E{P_\gamma^+}\right|= O(n^{-2/3}(\log n)^{A\log\log n})$ with probability at least $1 - O(1/n^K)$.
\end{lemma}
\begin{proof}
    Observe that by \cref{OrtDist} and since $\E{v_i^2} = 1/n$ for uniform $v_i \sim S_+^{n - 1}$,
    \[\E{Y_{kk}^+} - \E{P_\gamma^+} = \E{\sum_{i = 1}^{n/2}(\lambda_i - \gamma_i)u_{ki}^2} = \sum_{i = 1}^{n/2}\E[\lambda]{(\lambda_i - \gamma_i)}\E[U]{u_{ki}^2}=\frac{1}{n}\sum_{i = 1}^{n/2}\E[\lambda]{(\lambda_i - \gamma_i)}. \]
    Using triangle inequality (twice) and the law of total expectation, $\left|\E{Y_{kk}^+} - \E{P_\gamma^+}\right|$ is at most
    \[\frac{1}{n}\sum_{i = 1}^{n/2}\left|\E[\lambda]{(\lambda_i - \gamma_i)}\right| \leq \frac{1}{n}\sum_{i = 1}^{n/2}\left|\E[\lambda]{(\lambda_i - \gamma_i) \mid E}\right|\P{E} + \frac{1}{n}\sum_{i = 1}^{n/2}\left|\E[\lambda]{(\lambda_i - \gamma_i)\mid \oE}\right|\P{\oE}.\]
    By definition of $E$, $\left|\E[\lambda]{(\lambda_i - \gamma_i) \mid E}\right| = O(n^{-2/3}(\log n)^{A\log\log n})$, and we bound $\P{E} \leq 1$.
    In the event of $\oE$, we can bound $\left|\E[\lambda]{(\lambda_i - \gamma_i)\mid \oE}\right| \leq 5$, as by \cite{BVH16} it holds that $|\lambda_1| \leq 2 + o(1)$ with probability at least $1 - O(n^{-2K})$ for any constant $K$, and $\gamma_i \in [-2, 2]$ for all $i \in [n]$.
    Then
    \[\frac{1}{n}\sum_{i = 1}^{n/2}\left|\E[\lambda]{(\lambda_i - \gamma_i)}\right| \leq O\left(\frac{(\log n)^{A\log\log n}}{n^{2/3}}\right)\cdot 1 + \exp\left(-\Omega\left(\left(\log n\right)^{\tau A\log\log n}\right)\right) =  O\left(\frac{(\log n)^{A\log\log n}}{n^{2/3}}\right)\]
    with probability at least $1 - O(n^{-K})$.
\end{proof}
\begin{lemma}\label{pclosediag}
    For any $t > 0$ and any even $s > 0$, it holds for sufficiently large $n$ that
    \[\P{|Y^+_{kk} - P_\gamma^+| \geq t} \leq \left(\frac{(\log n)^{(A + C)\log\log n}}{2n^{2/3}t}\right)^s. \]
\end{lemma}
\begin{proof}
    For any even $s \geq 2$ we can bound
    \[ \P{\left|Y^+_{kk} - P_\gamma^+\right| \geq t} = \P{\left|Y^+_{kk} - P_\gamma^+\right|^s \geq t^s} \leq t^{-s}\E{\left|Y^+_{kk} - P_\gamma^+\right|^s}.\]
    By the law of total expectation,
    \[\E{\left|Y^+_{kk} - P_\gamma^+\right|^s} = \E{\left|Y^+_{kk} - P_\gamma^+\right|^s\mid E}\P{E} + \E{\left|Y^+_{kk} - P_\gamma^+\right|^s\mid \overline{E}}\P{\overline{E}}.\]
    Since $|Y^+_{kk}| \leq (1 + o(1))n$  with probability at least $1 - O(n^{-2K})$ and $|P_\gamma^+| \leq n$, by \cref{rigeig}
    \[\E{\left|Y^+_{kk} - P_\gamma^+\right|^s\mid \overline{E}}\P{\overline{E}} \leq (3n)^{s}\cdot \P{\overline{E}} \leq (3n)^{s}\cdot \exp\left(-c\left(\log n\right)^{\tau A\log\log n}\right) \xrightarrow{n\to\infty}0.\]
    On the other hand, under the event $E$ we have
    \begin{multline*}
        \left|Y_{kk}^+ - P_\gamma^+\right| = \left|\sum_{i = 1}^{n/2}(\lambda_i - \gamma_i)u_{ik}^2\right| \leq \sum_{i = 1}^{n/2}|\lambda_i - \gamma_i|u_{ik}^2 \\
        \leq \frac{(\log n)^{A\log\log n}}{n^{2/3}}\sum_{i = 1}^{n/2}u_{ik}^2 \leq \frac{(\log n)^{A\log\log n}}{n^{2/3}}\sum_{i = 1}^{n/2}\frac{(\log n)^{C\log\log n}}{n} = \frac{(\log n)^{(A + C)\log\log n}}{2n^{2/3}}.
    \end{multline*}
    As a result, using $\P{E} \leq 1$,
    \[\E{\left|Y^+_{kk} - P_\gamma^+\right|^s\mid E}\P{E} \leq \left(\frac{(\log n)^{(A + C)\log\log n}}{2n^{2/3}}\right)^s.\]
    Combining both bounds for events $E$ and $\oE$ in the law of total expectation finishes the proof.
\end{proof}
\begin{lemma}\label{pconcdiag}
    For any $t > 0$, for any constant $K > 0$,
    \[\P{\left|P_\gamma^+ - \E{P_\gamma^+}\right| > t} \leq 2n^{2 - K} + n\exp\left(-\frac{(K - 1)^2}{4K}n\right) + \exp\left(-\frac{(n - 2)t^2}{96K^3\log n}\right).\]
\end{lemma}
\begin{proof}
    By \cref{OrtDist}, $Y = U\Lambda U^T$ where $U\sim \Ort(n)$ and $U = [u_1, \ldots, u_n]$.
    Note that $P_\gamma^+$ depends only on $U$.
    Define $B_u := \left\{\exists i : \|u_i\|_\infty \geq \sqrt{\frac{2K^3\log n}{n}}\right\}$.
    By the law of total probability
    \[\P{\left|P_\gamma^+ - \E{P_\gamma^+}\right| > t} = \P{\left|P_\gamma^+ - \E{P_\gamma^+}\right| > t \mid B_u}\P{B_u} + \P{\left|P_\gamma^+ - \E{P_\gamma^+}\right| > t\mid \overline{B_u}}\P{\overline{{B_u}}}.\]
    We are going to use the following result on delocalization of columns of orthogonal matrices.
    \begin{theorem}[\cite{Rou16}]\label{ubound}
        Let $U \sim \Ort(n)$, $U = [u_1,\ldots, u_n]$.
        For every $i \in [n]$, for any $K > 0$,
        \[\P{\|u_i\|_{\infty} \geq \sqrt{\frac{2K^3\log n}{n}}} \leq 2n^{1 - K} + \exp\left(-\frac{(K - 1)^2}{4K}n\right).\]
    \end{theorem}
    Applying \cref{ubound} to $\P{B_u}$ and using $\P{\left|P_\gamma^+ - \E{P_\gamma^+}\right| > t \mid B_u} \leq 1$, we have
    \[\P{\left|P_\gamma^+ - \E{P_\gamma^+}\right| > t \mid B_u}\P{B_u} \leq 2n^{2 - K} + n\exp\left(-\frac{(K - 1)^2}{4K}n\right).\]
    Next, we show that $P_\gamma^+$ is Lipschitz under the event $\overline{B_u}$ as a function of $U$.
    \begin{lemma}\label{plipschitz}
        The function $f_\gamma(x) := \sum_{i = 1}^{n/2}\gamma_ix_i^2\I{\|x\|_\infty \leq \sqrt{\frac{2K^3\log n}{n}}}$ is $4\sqrt{K^3\log n}$-Lipschitz.
    \end{lemma}
    \begin{proof}
        It is easy to see that $\nabla f_\gamma(x) = [2\gamma_ix_i]_{i = 1}^{n/2}\I{\|x\|_\infty \leq \sqrt{\frac{2K^3\log n}{n}}}$ and \[|f_\gamma|_L = \|\nabla f_\gamma(x)\| = \sqrt{\sum_{i = 1}^{n / 2}4\gamma_i^2x_i^2}\I{\|x\|_\infty \leq \sqrt{\frac{2K^3\log n}{n}}} \leq \sqrt{\frac{8K^3\log n}{n}\sum_{i = 1}^{n/2}\gamma_i^2} \leq 4\sqrt{
        K^3\log n},\]
        where we used the fact that $\gamma_i \in [-2, 2]$ for $i \in [n]$.
    \end{proof}
    \cref{plipschitz} allows us to show concentration of $P_\gamma^+$ under $\overline{B_u}$, using the following theorem.
    \begin{theorem}[\cite{Meckes19}]\label{OrtConc3}
        Suppose that $F : \Ort(n)\to \R$ is $L$-Lipschitz w.r.t Hilbert-Schmidt metric (i.e induced by matrix inner product) on $\Ort(n)$, and that $U\sim \Ort(n)$.
        For each $t > 0$,
        \[\P{\left|F(U) - \E{F(U)}\right| > t} \leq \exp\left(-\frac{(n - 2)t^2}{24L^2}\right).\]
    \end{theorem}
    By \cref{plipschitz}, $P_\gamma^+$ as a function of $U$ is $4\sqrt{K^3\log n}$-Lipschitz under $\overline{B_u}$, so by \cref{OrtConc3}
    \[\P{\left|P_\gamma^+ - \E{P_\gamma^+}\right| > t\mid \overline{B_u}}\P{\overline{{B_u}}} \leq \P{\left|P_\gamma^+ - \E{P_\gamma^+}\right| > t\mid \overline{B_u}} \leq \exp\left(-\frac{(n - 2)t^2}{96K^3\log n}\right).\]
    Combining the bounds for events $B_u$ and $\overline{B_u}$ gives the desired claim.
\end{proof}
To see the concentration of $Y^+_{kk}$ around $\E{Y^+_{kk}}$, note that by independence from \cref{OrtDist}
\[\P{\left|Y^+_{kk} - \E{Y^+_{kk}}\right| \geq 3t} \leq \P{\left|Y^+_{kk} - P_\gamma^+\right| \geq t} + \P{\left|P_\gamma^+ - \E{P_\gamma^+}\right| \geq t} + \P{\left|\E{P_\gamma^+} - \E{Y^+_{kk}}\right| \geq t}.\]
We use \cref{expclosediag}, \cref{pclosediag} and \cref{pconcdiag} with $s = cK$ for sufficiently large constant $c$ and $t = n^{-1/3}/3$, as well as independence of eigenvectors and eigenvalues from \cref{OrtDist}, to bound the individual summands' probabilities (the choice of $t$ is not tight, but satisfies our needs for \cref{diagconc}).
We obtain that for every $k \in [n]$ and any constant $K > 0$, $\P{\left|Y^+_{kk} - \E{Y^+_{kk}}\right| \geq n^{-1/3}} = O(n^{-K})$.
The proof for $Y^-_{kk}$ is completely analogous.
\end{proof}

\newpage
\section{$W = W(X)$ is generalized Wigner for Gaussian $X$}

\begin{theorem}[\cref{genwig} restated]\label{genwigmain}
    Let $X \in \S^n$ be symmetric random matrix with i.i.d entries, such that for every $i \leq j \in [n]$, $X_{ii} \sim \No(0, 2)$ and $X_{ij} = X_{ji} \sim \No(0, 1)$.
    Let $\frac{1}{\sqrt{n}}X = \sum_{k = 1}^n\lambda_ku_ku_k^T$ be the eigendecomposition of $\frac{1}{\sqrt{n}}X$, and let $\frac{1}{\sqrt{n}}X^+ = \sum_{k = 1}^{n/2}\lambda_ku_ku_k^T$, $\frac{1}{\sqrt{n}}X^- = \sum_{k = n/2 + 1}^n\lambda_ku_ku_k^T$.
    Define a generalized Wigner matrix $W = W(X)\in \S^n$ as follows: for all $i \neq j \in [n]$, $W_{ii} = 0$, and $W_{ij}$ takes values $\{\pm (X^+_{ij} - X^-_{ij})\}$ with probability $1/2$ each.
    
    Define the \textbf{semi-circle law with parameter $\alpha$} as a probability distribution $\Rho_{\alpha}$ with density
    \[\rho_{\alpha}(x) = \d\Rho_{\alpha}(x) :=  \frac{1}{2\pi \alpha^2}\sqrt{4\alpha^2 - x^2}\I{x \in [-2\alpha, 2\alpha]}\d x.\]
    Let $L_{W}(x) := \frac{1}{n}\sum_{i = 1}^n\I{\lambda_{i}(W/\sqrt{n}) \leq x}$ be the empirical spectral distribution measure of $W/\sqrt{n}$.
    With probability at least $1 - O(1/n^3)$ over the choice of $X$, $L_{W} \xrightarrow[n\to\infty]{\mathrm{a.s.}} \Rho_{\alpha}$ with $\alpha = \sqrt{1 - \frac{64}{9\pi^2}}$.
\end{theorem}
The proof will rely on a well-known result of \cite{GNT15} that provides a list of sufficient conditions on matrix' variances so that its spectrum converges to a semicircular distribution.
\begin{theorem}[\cite{GNT15}]\label{wiglist}
    Let $W \in \R^{n\times n}$ be a symmetric random matrix with independent entries, for all $i, j \in [n]$, $\E{W_{ij}^2} = \sigma^2_{ij} < \infty$.
    Let $L_{W}(x) := \frac{1}{n}\sum_{i = 1}^n\I{\lambda_{i}(W/\sqrt{n}) \leq x}$.
    If
    \begin{enumerate}
        \item for any constant $\eta > 0$, $\lim_{n\to\infty}\frac{1}{n^2}\sum_{i, j = 1}^n\E{W_{ij}^2\I{|W_{ij}| > \eta\sqrt{n}}} = 0$;
        \item there exists global constant $C$ such that for every $i \in [n]$, $\frac{1}{n}\sum_{j = 1}^n\sigma_{ij}^2 \leq C$;
        \item $\frac{1}{n}\sum_{i = 1}^n\left|\frac{1}{n}\sum_{j = 1}^n\sigma_{ij}^2 - 1\right|\xrightarrow{n\to\infty}0$;
    \end{enumerate}
    then $\sup_{x\in\R}\left|\E{L_W(x)} - \Rho_1(x)\right| \xrightarrow{n\to\infty}0$, where $\Rho_1$ is the standard semicircle law on $[-2, 2]$.
\end{theorem}
\cref{wiglist} gives sufficient conditions for the pointwise convergence of $\E{L_{W}}$ to $\Rho_\alpha$ for some $\alpha > 0$ (after rescaling).
To obtain almost sure convergence, we use the result of \cite{C23}.
\begin{theorem}[\cite{C23}]\label{easconverge}
    Let $W_n \in \R^{n\times n}$ be a symmetric random matrix with independent entries, and let $L_{W_n}(x) := \frac{1}{n}\sum_{i = 1}^n\I{\lambda_{i}(W_n/\sqrt{n}) \leq x}$.
    Then, for any $\alpha > 0$,
    \[\sup_{x\in\R}\left|\E{L_{W_n}(x)} - \Rho_\alpha(x)\right| \xrightarrow{n\to\infty}0\qquad \iff \qquad L_{W_n}\xrightarrow[n\to\infty]{\mathrm{a.s.}} \Rho_{\alpha}.\]
\end{theorem}
Combining \cref{wiglist} with \cref{easconverge}, we see that to prove \cref{genwigmain} it suffices to show that $W$ satisfies the three conditions of \cref{wiglist} (with a corresponding scale $\alpha$).
Let $Y = \frac{1}{\sqrt{n}}X =  \sum_{k = 1}^n\lambda_ku_ku_k^T$, $Y^+ = \frac{1}{\sqrt{n}}X^+ = \sum_{k = 1}^{n/2}\lambda_ku_ku_k^T$, $Y^- = \frac{1}{\sqrt{n}}X^- = \sum_{k = n/2 + 1}^n\lambda_ku_ku_k^T$.

\begin{theorem}\label{cond1}
    With probability at least $1 - O(1/n^8)$ over the choice of $X$, for any constant $\eta > 0$ it holds $\lim_{n\to\infty}\frac{1}{n^2}\sum_{i, j = 1}^n\E{W_{ij}^2\I{|W_{ij}| > \eta\sqrt{n}}} = 0$.
\end{theorem}
\begin{proof}
    We are going to show that for every $i, j \in [n]$, $W_{ij}$ is well-concentrated around $0$.
    Then, the number of $i, j$ with $|W_{ij}| > \eta\sqrt{n}$ is small, hence the limit of the sum above is $0$.

    For $i < j \in [n]$, consider $\frac{1}{\sqrt{n}}W_{ij} = Y^+_{ij} - Y^-_{ij}$.
    Using independence from \cref{OrtDist},
    \[\E{Y^+_{ij} - Y^-_{ij}} = \sum_{k = 1}^n(-1)^{k > n/2}\E[\lambda]{\lambda_k}\E[U]{u_{ki}u_{kj}} = \sum_{k = 1}^n(-1)^{k > n/2}\E[\lambda]{\lambda_k}\E[U]{\langle u_k, e_i\rangle \langle u_k, e_j\rangle}.\]
    By \cref{OrtDist}, $u_k$ is distributed uniformly on $S^{n - 1}_+$ for every $k$, hence $\E[U]{\langle u_k, e_i\rangle \langle u_k, e_j\rangle } = 0$ for all $k\in[n]$, and $\E{\frac{1}{\sqrt{n}}W_{ij}} = 0$.
    It remains to show that $W_{ij}$ is well-concentrated around $\E{W_{ij}}$.
    We are going to use essentially the same approach as we did in proving \cref{diagconc}.
    
    For $i < j \in [n]$, consider $Y_{ij}^+ - Y_{ij}^- = \sum_{k = 1}^n(-1)^{k > n/2}\lambda_ku_{ki}u_{kj}$.
    For $i, j \in [n]$, define the following polynomial $P_\gamma(i, j) =: P_\gamma$ in $2n$ variables: $P_\gamma(i, j) := \sum_{k = 1}^{n}(-1)^{k > n/2}\gamma_ku_{ki}u_{kj}$, where $\gamma_k$ is the classical location of the $k$-th eigenvalue under the semicircle law and $u_k$ is the $k$-th eigenvector of $Y$.
    Note that $Y^+_{ij} - Y^-_{ij}$ is $P_\gamma$ with $\gamma$-s swapped to $\lambda$-s.
    To prove concentration of $Y^+_{ij} - Y^-_{ij}$ around $0$, we will show that 1) $\E{P_\gamma}$ is equal to $0$, 2) $Y^+_{ij} - Y^-_{ij}$ is close to $P_\gamma$ with high probability, and 3) $P_\gamma$ is concentrated around $\E{P_\gamma}$ with high probability.

    The lemma below follows from \cref{OrtDist}, and is proved similarly to $\E{W_{ij}} = 0$.
    \begin{lemma}\label{expcloseent}
        For every $i < j \in [n]$, $\E{P_\gamma} = \E{P_\gamma(i, j)} = 0$.
    \end{lemma}
    Next, the proof of $Y_{ij}^+ - Y_{ij}^-$ being close to $P_\gamma(i, j)$ is completely analogous to \cref{pclosediag}.
    \begin{lemma}\label{pcloseent}
        For any $i < j \in [n]$, any $t > 0$, any even $s > 0$, it holds for sufficiently large $n$
        \[\P{\left|(Y^+_{ij} - Y_{ij}^-) - P_\gamma(i, j)\right| \geq t} \leq \left(\frac{(\log n)^{(A + C)\log\log n}}{2n^{2/3}t}\right)^s. \]
    \end{lemma}
    To prove concentration of $P_\gamma$ around $\E{P_\gamma}$, we use essentially the same proof as \cref{pconcdiag}. 
    \begin{lemma}\label{pconcent}
        For any $i < j \in [n]$, any $t > 0$, for any constant $K > 0$,
        \[\P{\left|P_\gamma(i, j) - \E{P_\gamma(i, j)}\right| > t} \leq 2n^{2 - K} + n\exp\left(-\frac{(K - 1)^2}{4K}n\right) + \exp\left(-\frac{(n - 2)t^2}{172K^3\log n}\right).\]
    \end{lemma}
    Using independence of eigenvalues and eigenvectors from \cref{OrtDist}, and combining the bounds of \cref{expcloseent}, \cref{pcloseent} and \cref{pconcent} with $s = cK$ for sufficiently large constant $c$ and $t = n^{-1/3}/3$, we get that for any $K > 0$, $\P{|W_{ij}| > n^{1/6}} = \P{\left|Y_{ij}^+ - Y_{ij}^-\right|\geq n^{-1/3}} = O(n^{2 - K})$.
    This holds for every fixed $i < j \in [n]$.
    Taking $K = 10$, we get that simultaneously for all $i < j \in [n]$, $\left|X^+_{ij} - X^-_{ij}\right| \leq n^{1/6}$ with probability at least $1 - O(n^{-8})$.

    Now, $\left|X^+_{ij} - X^-_{ij}\right| = \sqrt{n}\left|Y_{ij}^+ - Y_{ij}^-\right| \leq \sqrt{n}\sum_{k = 1}^{n}|\lambda_k|\cdot |u_{ki}u_{kj}| \leq (2 + o(1))\sqrt{n}\sum_{k = 1}^{n}\|u_k\|^2 \leq (2 + o(1))n^{3/2}$, as by \cite{BVH16} it holds that $|\lambda_1|\leq 2 + o(1)$ with probability $1 - O(n^{-2K})$ for any $K > 0$.
    As a result, with probability at least $1 - O(n^{-8})$ over the choice of $X$,
    $\frac{1}{n^2}\sum_{i, j = 1}^n\E{W_{ij}^2\I{|W_{ij}| > \eta\sqrt{n}}} 
        \leq \frac{1}{n^2}\sum_{i, j = 1}^n\left(X^+_{ij} - X^-_{ij}\right)^2\cdot O\left(n^{-8}\right) = O\left(\frac{4n^{9/4}}{n^{10}}\right) \xrightarrow{n\to\infty}0$. 
\end{proof}
Next, we fix $i \in [n]$ and would like to show concentration of $\sum_{j = 1}^n\E{W_{ij}^2}$.
Since $W_{ii}^2 = 0$, we can express it as $\sum_{j = 1}^n\E{W_{ij}^2} = -(X_{ii}^+ - X_{ii}^-)^2 + \sum_{j = 1}^n(X_{ij}^+ - X_{ij}^-)^2$.
By \cref{diagconc}, $(X_{ii}^+ - X_{ii}^-)^2$ is concentrated, so it is sufficient to show the concentration of $\sum_{j = 1}^n(X_{ij}^+ - X_{ij}^-)^2$.
\begin{theorem}\label{rowconc}
    With probability at least $1 - O(1/n^8)$ over the choice of $X$, for any $K > 0$
    \[\E{\sum_{j = 1}^n(Y_{ij}^+ - Y_{ij}^-)^2 } = 1,\qquad \P{\left|\sum_{j = 1}^n(Y_{ij}^+ - Y_{ij}^-)^2 - \E{\sum_{j = 1}^n(Y_{ij}^+ - Y_{ij}^-)^2 }\right| \geq \frac{\log n}{\sqrt{n}}} = O(1/n^{K - 2}).\]
\end{theorem}
\begin{proof}
The sum $\sum_{j = 1}^n(Y_{ij}^+ - Y_{ij}^-)^2$ can be expressed as
\begin{multline*}
    \sum_{j = 1}^n(Y_{ij}^+ - Y_{ij}^-)^2 = \sum_{j = 1}^n\left(\sum_{k = 1}^n(-1)^{2k > n}\lambda_ku_{ki}u_{kj}\right)\left(\sum_{l = 1}^n(-1)^{2l > n}\lambda_lu_{li}u_{lj}\right) 
    \\
    =\sum_{j = 1}^n\sum_{k, l= 1}^n(-1)^{2k > n \oplus 2l > n}\lambda_k\lambda_lu_{ki}u_{kj}u_{li}u_{lj} = \sum_{k, l}(-1)^{2k > n \oplus 2l > n}\lambda_k\lambda_lu_{ki}u_{li}\sum_{j = 1}^nu_{kj}u_{lj}.
\end{multline*}
Since $\sum_{j = 1}^nu_{kj}u_{lj} = \langle u_k, u_l\rangle = \I{k = l}$, the sum $\sum_{j = 1}^n(Y_{ij}^+ - Y_{ij}^-)^2$ is equal to
\[\sum_{k, l}(-1)^{2k > n \oplus 2l > n}\lambda_k\lambda_lu_{ki}u_{li}\sum_{j = 1}^nu_{kj}u_{lj} = \sum_{k, l}(-1)^{2k > n \oplus 2l > n}\lambda_k\lambda_lu_{ki}u_{li}\I{k = l} = \sum_{k = 1}^n\lambda_k^2u_{ki}^2,\]
and by independence from \cref{OrtDist}, $\E{\sum_{j = 1}^n(Y_{ij}^+ - Y_{ij}^-)^2} = \sum_{k = 1}^n\E[\lambda]{\lambda_k^2}\E[U]{u_{ki}^2}$.
Again by \cref{OrtDist}, vector $u_k$ is distributed uniformly on the unit sphere, hence $\E[U]{u_{ki}^2} = 1/n$ and $\E{\sum_{j = 1}^n(Y_{ij}^+ - Y_{ij}^-)^2} = \frac{1}{n}\sum_{k = 1}^n\E[\lambda]{\lambda_k^2} = \frac{1}{n}\sum_{i, j = 1}^n\E{Y_{ij}^2} = \frac{1}{n^2}\sum_{i, j = 1}^n\E{X_{ij}^2} = 1$.

Next, to show concentration of $\sum_{j = 1}^n(Y_{ij}^+ - Y_{ij}^-)^2$, we use the same approach as we did with \cref{diagconc} and \cref{cond1}.
Recall that for $i \in [n]$, $\sum_{j = 1}^n(Y_{ij}^+ - Y_{ij}^-)^2 = \sum_{k = 1}^n\lambda_k^2u_{ki}^2$, and define the following polynomial $P_\gamma(i) =: P_\gamma^i$ in $n$ variables: $P_\gamma(i) := \sum_{k = 1}^n\gamma_k^2u_{ki}^2$, where $\gamma_k$ is the classical location of the $k$-th eigenvalue under the semicircle law and $u_k$ is the $k$-th eigenvector of $Y$.
Essentially, $\sum_{j = 1}^n(Y_{ij}^+ - Y_{ij}^-)^2$ is $P_\gamma^i$ with $\gamma$-s swapped to $\lambda$-s.
To prove concentration of the sum $\sum_{j = 1}^n(Y_{ij}^+ - Y_{ij}^-)^2$ around 1, we will show that 1) $\E{P_\gamma^i}$ is equal to $1$, 2) $\sum_{j = 1}^n(Y_{ij}^+ - Y_{ij}^-)^2$ is close to $P_\gamma^i$ with high probability, and 3) $P_\gamma^i$ is concentrated around $\E{P_\gamma^i}$ with high probability.

The lemma below follows from \cref{OrtDist}, similarly to proof of $\E{\sum_{j = 1}^n(Y_{ij}^+ - Y_{ij}^-)^2} = 1$.
    \begin{lemma}\label{expcloserow}
        For every $i \in [n]$, $\E{P_\gamma^i} = \E{P_\gamma(i)} = 1$.
    \end{lemma}
    Next, the proof of $\sum_{j = 1}^n(Y_{ij}^+ - Y_{ij}^-)^2$ being close to $P_\gamma(i)$ is analogous to \cref{pclosediag}.
    \begin{lemma}\label{pcloserow}
        For any $i \in [n]$, any $t > 0$, any even $s > 0$, it holds for sufficiently large $n$
        \[\P{\left|\sum_{j = 1}^n(Y_{ij}^+ - Y_{ij}^-)^2 - P_\gamma(i)\right| \geq t} \leq \left(\frac{4(\log n)^{(A + C)\log\log n}}{n^{2/3}t}\right)^s. \]
    \end{lemma}
    To prove concentration of $P_\gamma^i$ around $\E{P_\gamma^i}$, we use essentially the same proof as \cref{pconcdiag}. 
    \begin{lemma}\label{pconcrow}
        For any $i \in [n]$, any $t > 0$, for any constant $K > 0$,
        \[\P{\left|P_\gamma(i) - \E{P_\gamma(i)}\right| > t} \leq 2n^{2 - K} + n\exp\left(-\frac{(K - 1)^2}{4K}n\right) + \exp\left(-\frac{(n - 2)t^2}{384K^3\log n}\right).\]
    \end{lemma}
    Using independence of eigenvalues and eigenvectors from \cref{OrtDist}, and combining the bounds of \cref{expcloserow}, \cref{pcloserow} and \cref{pconcrow} with $s = cK$ for sufficiently large constant $c$ and $t = n^{-1/2}\log n$ (as was done in the proof of \cref{diagconc}), we get that for every fixed $i \in [n]$, for any constant $K > 0$, with probability at least $1 - O(n^{-8})$ over the choice of $X$, $\P{\left|\sum_{j = 1}^n(Y_{ij}^+ - Y_{ij}^-)^2 - \E{\sum_{j = 1}^n(Y_{ij}^+ - Y_{ij}^-)^2 }\right| \geq \frac{\log n}{\sqrt{n}}} = O(1/n^{K - 2})$.
\end{proof}
From \cref{diagconc} we know that with probability at least $1 - O(n^{2 - K})$ for any $K > 0$, $(X_{ii}^+ - X_{ii}^-)^2 = (1 + o(1))\frac{64}{9\pi^2}n$.
Combining with \cref{rowconc}, we get the following corollary.
\begin{corollary}\label{cond23}
    For any $i \in [n]$, any $K > 0$, with probability at least $1 - O(n^{2 - K})$, $\sum_{j = 1}^n\E{W_{ij}^2} = \left(1 - \frac{64}{9\pi^2} + o(1)\right)n$, and for $\alpha = \sqrt{1 - \frac{64}{9\pi^2}}$, $\frac{1}{n}\sum_{i = 1}^n\left|\frac{1}{n}\sum_{j = 1}^n\E{W_{ij}^2} - \alpha^2\right| = O\left(\frac{\log n}{\sqrt{ n}}\right)\xrightarrow{n\to\infty}0$.
\end{corollary}

Together \cref{cond1} and \cref{cond23} give sufficient conditions to apply \cref{wiglist}, showing that $L_W$ converges almost surely to semi-circular law $\P_\alpha$.
As a sanity check for the value $\alpha = \sqrt{1 - \frac{64}{9\pi^2}}$, we bound the largest singular value of $W$ using the result of \cite{BVH16} below.-
\begin{theorem}[\cite{BVH16}]
    Let $\sigma := \max_{i \in [n]}\sqrt{\sum_{j = 1}^n\E{W_{ij}^2}}$, and $\sigma_* = \max_{i, j \in [n]}\E{|W_{ij}|}$.
    Then
    \[\P{\left|\|W\|_{\mathrm{op}}- \left(2\sigma + c\sigma_*\sqrt{\log n}\right)\right| > t} \leq \exp\left(-\frac{t^2}{4\sigma_*^2}\right)\]
\end{theorem}
In our case, as shown in the proofs of \cref{cond1} and \cref{cond23}, with probability at least $1 - O(1/n^8)$ over the randomness of $X$, $\sigma = (1 + o(1))\alpha\sqrt{n}$ and $\sigma_* \leq (1 + o(1))n^{1/6}$, thus choosing $t = n^{1/6}\log^{8}n$ in the theorem above shows that the largest singular value of $W$, with probability at least $1 - O(1/n^3)$, is $(1 + o(1))2\alpha\sqrt{n}$.
But then we immediately get that the support of the limiting semicircle of distribution $L_W$ is $[-2\alpha, 2\alpha]$, implying that $L_W \xrightarrow[n\to\infty]{\mathrm{a.s.}}\Rho_\alpha$ for $\alpha =\sqrt{1 - \frac{64}{9\pi^2}}$ exactly.
This finishes the proof of \cref{genwigmain}.
\newpage
\section{Concentration of $\Diag(X^- - X^+)$ for Wigner $X$}
\begin{theorem}\label{diagconc2}
    Let $X \in \S^n$ be symmetric random matrix with i.i.d entries, such that for every $i \leq j \in [n]$, $X_{ii} = 0$ and $X_{ij} = X_{ji}$ takes values $\{\pm 1\}$ with equal probability.
    Let $X = \sum_{i = 1}^n\lambda_iu_iu_i^T$ be the eigendecomposition of $X$, and $X^+ = \sum_{i = 1}^{n/2}\lambda_iu_iu_i^T$.
    For every $k \in [n]$ and $t > 0$,
    \[\frac{1}{\sqrt{n}}\E{X^+_{kk}} = \frac{4}{3\pi} + o(1)\quad \text{and}\quad \P{\left|\frac{1}{\sqrt{n}}X^+_{kk} - \frac{1}{\sqrt{n}}\E{X^+_{kk}}\right| \geq t} \leq O\left(\frac{1}{t^2\sqrt{n}}\right). \]
\end{theorem}

Assume that matrices $X$ and $X^+$ are already normalized by $\frac{1}{\sqrt{n}}$.
We are going to look at diagonal entries of $X^+$ as linear statistics of matrix $X$ and some deterministic matrix $A$.
More specifically, 
{$X^+_{kk} = \sum_{i = 1}^{n/2}\lambda_i\dpr{u_i}{Au_i}$, where $A$ is such that $A_{kk} = 1$ and all other entries of $A$ are zero, so $u_{ki}^2 = \dpr{u_i}{Au_i}$.
Now, by \cref{rigeig}, with extremely high probability $u_{ki}^2 = O(n^{-1}\polylog n)$ for all $k, i$, and from the proof of \cref{hxdist} we know that with probability at least $1 - O(n^{-1}\polylog n)$ there are at most $O(n^{1/3}\polylog n)$ eigenvalues of $X$ that are within $O(n^{-2/3}\polylog n)$-neighborhood of $0$.
Hence, it holds for every $k \in [n]$ that $X^+_{kk} = \sum_{i = 1}^{n/2}\lambda_i\dpr{u_i}{Au_i} = \sum_{i = 1}^nf(\lambda_i)\dpr{u_i}{Au_i} \pm o(1)$ with probability at least $1 - O(n^{-1}\polylog n)$, where $f(x) = x\I{x\geq 0}$.}
In general, for our matrix $X$ and arbitrary $f$ and $A$, denote by 
\[L_n(f, A) = L_{X, n}(f, A) := \sum_{i = 1}^nf(\lambda_i)\dpr{u_i}{Au_i} - \E{\sum_{i = 1}^nf(\lambda_i)\dpr{u_i}{Au_i}}.\]
For deterministic $A \in \R^{n \times n}$, denote $\ntr{A} := \frac{1}{n}\Tr{A}$.
Then, any matrix $A$ can be decomposed into the sum $A = A_{\dm} + A_{\odm} = \ntr{A}I + \cA_{\dm} + A_{\odm}$, where $A_{\dm}$, $A_{\odm}$ are respectively diagonal and off-diagonal components of $A$, and $\cA_{\dm} := A_{\dm} - \ntr{A}I$ is the diagonal traceless component of $A$ (in a sense $\ntr{\cA_{\dm}} = 0$).
As a corollary, $L_n(f, A) = \ntr{A}L_n(f, I) + L_n(f, \cA_{\dm}) + L_n(f, A_{\odm})$.
\begin{definition}
    For $k \in \N$, let $Z = (Z_1, \ldots, Z_k)$, $Y = (Y_1, \ldots, Y_k)$ be two random vectors, where $Z_i, Y_j$, are dependent on parameter $n$.
    We say that $Z$ and $Y$ are \textbf{close in the sense of moments}, denoted as
    \[Z \eqm Y + O_\mm(n^{-c}),\]
    for some $c > 0$, if for any polynomial $p(x_1,\ldots, x_k)$ it holds that
    \[\E{p(Z_1,\ldots, Z_k)} = \E{p(Y_1, \ldots, Y_k} + O(n^{-c + \xi})\]
    for any small $\xi > 0$, where constant in $O(\cdot)$ depends on $k, \xi, p$ and constant parameters of $Z, Y$.
\end{definition}
We are going to show a vector of linear statistics $\big(L_n(f, I), L_n(f, \cA_\dm), L_n(f, A_\odm)\big)$ is close in the senses of moments to a vector of three independent centered Gaussian random variables with small variances, which would imply that the combined statistic $L_n(f, A) \equiv X_{kk}^+ - \E{X_{kk}^+}$ is well-concentrated around zero.
To do so, we are going to invoke the following result of \cite{CES23} which studies the behavior of statistics $L_n(f, A)$ for various functions $f$ and matrices $A$.
\begin{theorem}[\cite{CES23}, Theorem 2.4]\label{fclt}
    Let $X \in \S^{n}$ be a symmetric random matrix with i.i.d. entries, such that for every $i < j \in [n]$, $\chi_{ii} := \sqrt{n}X_{ii}$, and $\chi_{ij} := \sqrt{n}X_{ij}$, and $\E{\chi_{ii}} = \E{\chi_{ij}} = 0$.
    We denote $\omega_2:= \E{\chi_{ii}^2}$, $\sigma^2 := \E{\chi_{ij}^2}$, $\widetilde{\omega_2} := \omega_2 - 1 - \sigma$ and $\kappa_4 := \E{\chi_{ij}^4} - 2 - \sigma^2$.
    Let $A \in \R^{n\times n}$, $\|A\|\leq 1$, $A = \ntr{A}I + \cA_{\dm}$, be a diagonal deterministic matrix, and let $f : \R \to \R$ be a function with compact support.
    For an arbitrary function $g : \R \to \R$, denote the expectation of $g$ w.r.t the semicircuclar denisty and its inverse by $\ntr{g}_{sc}$ and $\ntr{g}_{1/sc}$ respectively.
    Then, if $\ntr{\cA_{\dm}^2} \gtrsim n^{-1 + \eps}$ for some $\eps > 0$,
    \[\big(L_n(f, I), L_n(f, \cA_\dm)\big) \eqm \big(\xi_{\tr}(f), \xi_{\dm}(f, \cA_{\dm})\big) + O_{\mm}(n^{-1/2}).\]
    Here $\xi_{\tr}(f), \xi_{\dm}(f, \cA_{\dm})$ are two independent centered $n$-dependent Gaussian processes with variances (scaled s.t. $\xi_{\tr}(f), \xi_{\dm}(f, \cA_{\dm})/\ntr{\cA_{\dm}^2}^{1/2}$ are $O(1)$):
    \begin{gather*}
        \E{\xi_{\tr}(f)^2} = V^1_{\tr}(f) + V^2_{\tr}(f, \sigma) + \frac{\kappa_4}{2}\ntr{(2 - x^2)f}^2_{1/sc} + \frac{\widetilde{\omega_2}}{4}\ntr{xf}^2_{1/sc};\\
        \E{\xi_{\dm}(f, \cA_{\dm})^2} = \ntr{\cA_{\dm}^2}\left(V_{\dm}^1(f) + V_{\dm}^2(f, \sigma) + \widetilde{\omega_2}\ntr{fx}^2_{sc} + \kappa_4\ntr{(x^2 - 1)f}^2_{sc}\right);
    \end{gather*}
    where
    \begin{gather*}
        V_{\tr}^1(f) := \frac{1}{4\pi^2}\int_{-2}^2\int_{-2}^2\left(\frac{f(x) - f(y)}{x - y}\right)^2\frac{4 - xy}{\sqrt{(4 - x^2)(4 - y^2)}}\d x \d y;\\
        V_{\tr}^2(f, \sigma) := \frac{1}{4\pi^2}\int_{-2}^2\int_{-2}^2f(x)f(y)\dd_x \dd_y \log\left[\frac{(x - \sigma y)^2 + (\sqrt{4 - x^2} + \sigma\sqrt{4 - y^2})^2}{(x - \sigma y)^2 + (\sqrt{4 - x^2} -\sigma\sqrt{4 - y^2})^2}\right]\d x \d y;\\
        V_{\dm}^1(f) := \ntr{f^2}_{sc} - \ntr{f}_{sc}^2;\\
        V_{\dm}^2(f, \sigma) := \frac{1}{4\pi^2}\int_{-2}^2\int_{-2}^2f(x)f(y)\frac{(1 - \sigma)^2\sqrt{(4 - x^2)(4 - y^2)}}{\sigma^2(x^2 + y^2) + (1 - \sigma^2)^2 - xy\sigma(1 + \sigma^2)}\d x \d y - \ntr{f}_{sc}^2;
    \end{gather*}
    for $|\sigma| < 1$, and $V_{\tr}^2, V_{\dm}^2$ are extended to $\sigma = \pm 1$ by continuity, $V^2_{\tr/\dm}(f, \pm 1) := \lim_{\sigma \to \pm 1}V^2_{\tr/\dm}(f, \sigma)$.
    In addition, for any $\eps > 0$,
    \begin{multline*}
        \E{\sum_{i = 1}^nf(\lambda_i)\dpr{u_i}{Iu_i}} = \E{\sum_{i = 1}^nf(\lambda_i)}=\\
        =n\ntr{f}_{sc} + \frac{\kappa_4}{2}\ntr{(x^4 - 4x^2 + 2)f}_{1/sc} - \frac{\widetilde{\omega_2}}{2}\ntr{(2 - x^2)f}_{1/sc} - \frac{E_{\tr}(f, \sigma)}{2} + O_\eps(n^{-1/2 + \eps});
    \end{multline*}
    and
    \[\left|\E{\sum_{i = 1}^nf(\lambda_i)\dpr{u_i}{\cA_{\dm}u_i}}\right| = O_\eps(n^{-1/2 + \eps}).\]
    Here 
    \[E_{\tr}(f, \sigma) := \ntr{f\left(1 - \frac{1 - \sigma^2}{(1 + \sigma)^2 - \sigma x^2}\right)}_{1/sc},\qquad |\sigma|< 1,\]
    and $E_{\tr}(f, \pm 1) := \lim_{\sigma \to \pm 1}E_{\tr}(f, \sigma)$.
    The constants under $O(\cdot)$ in all error terms depend only on the parameters of the distribution of $X$ and on $\|f\|, |\supp f|$. 
\end{theorem}
\begin{observation}[\cite{CES23}, Remark 2.7]
    When $\sigma = 1$ in \cref{fclt}, the following holds:
    \[V^2_{\tr}(f, 1) = V^1_{\tr}(f),\qquad\qquad V^2_{\dm}(f, 1) = V^1_{\dm}(f), \qquad \qquad E_{\tr}(f, 1) = \ntr{f}_{1/sc} - \frac{f(2) + f(-2)}{2}.\]
\end{observation}
The following lemma is a result of applying \cref{fclt} to our matrix $X$ and function $f$, substituting corresponding parameter values and directly computing the expressions.
\begin{lemma}\label{fcltparams}
    For matrix $X$ with $\E{X_{ii}} = \E{X_{ii}^2} = 0 = \E{X_{ij}}$, $\E{X_{ij}^2} = 1/n$ and $\E{X_{ij}^4} = 1/n^2$, $i \neq j \in [n]$, and $f = x\I{x \geq 0}$, we have $\sigma = 1$, $\kappa_4 = -2$, $\omega_2 = 0$ and $\widetilde{\omega_2} = -2$, and
    \begin{gather*}
        V^2_{\tr}(f, 1) = V^1_{\tr}(f) = \frac{1}{4} + \frac{1}{\pi^2};\quad\quad V^1_{\dm}(f) = V^2_{\dm}(f, 1)= \frac{1}{2} - \frac{16}{9\pi^2};\quad \quad E_{\tr}(f, 1) =  \frac{2}{\pi} - 1;\\
        \ntr{f}_{sc} = \frac{4}{3\pi};\qquad \ntr{(2 - x^2)f}^2_{1/sc} = \frac{16}{9\pi^2};\qquad 
        \ntr{xf}^2_{1/sc} =1;\qquad 
        \ntr{fx}^2_{sc} = \frac{1}{4};\\
        \ntr{(x^2 - 1)f}^2_{sc} = \frac{16}{25\pi^2};\qquad \ntr{(2 - x^2)f}_{1/sc} =  \frac{4}{3\pi};\qquad \ntr{(x^4 - 4x^2 + 2)f}_{1/sc} =-\frac{4}{15\pi}.
    \end{gather*}
    As a result, $\E{\xi_{\tr}(f)^2} = 2/(9\pi^2)$, $\E{\xi_{\dm}(f, \cA_{\dm})^2} = \ntr{\cA_{\dm}^2}\left(1/2 - 1088/(225\pi^2)\right)$, and 
    \[\E{\sum_{i = 1}^nf(\lambda_i)\dpr{u_i}{Iu_i}} = \E{\sum_{i = 1}^nf(\lambda_i)} =\frac{4}{3\pi}\cdot n + \frac{3}{5\pi} + \frac{1}{2} + O_\eps(n^{-1/2 + \eps}).\]
\end{lemma}

    Recall that in our case, for $k$ fixed, $X_{kk}^+ = \sum_{i = 1}^nf(\lambda_i)\dpr{u_i}{Au_i}$ where $f(x) = x\I{x\geq 0}$ and matrix $A \in S^n$ has $A_{kk} = 1$, and all other entries are $0$.
    Even though $\|A\|\leq 1$, unfortunately $\ntr{\cA^2_{\dm}}$ for this $A$ is approximately $1/n$, so \cref{fclt} cannot be applied to $A$ directly.
    To overcome this, we are going to represent $A = A' + A''$, and $L_n(f, A) = L_n(f, A') + L_n(f, A'')$, where $A', A'' \in \S^n$ are defined as follows.
    For some small $\eps \in (0, 1/2)$, matrix $A' \in \S^n$ has for all $j \leq n/2$, $A'_{jj} = n^{\eps - 1/2}$, and for all $j > n/2$, $A'_{jj} = -n^{\eps - 1/2}$.
    Similarly, for the same $\eps \in (0, 1/2)$, matrix $A'' \in S^n$ has for all $j \leq n/2$, $j \neq k$, $A''_{jj} = -n^{\eps - 1/2}$, for all $j > n/2$, $j \neq k$, $A''_{jj} = n^{\eps -1/2}$, and $A''_{kk} = (-1)^{j \leq n/2}n^{\eps - 1/2} + 1$.
    Observe that $\|A'\|\leq 1$, $\ntr{A'} = 0$ and hence $\cA' = A'$, so $\Tr{(\cA')^2} = n\cdot n^{2\eps - 1}$ and $\ntr{(\cA')^2} = n^{2\eps - 1}$.
    Similarly, $\|A''\| \leq 1$, $\ntr{A''} = 1/n$ and $\cA'' = A'' - I/n$, and consequently, 
    \[\Tr{(\cA'')^2} = (n - 1)\cdot \left(n^{\eps - 1/2} - n^{-1}\right)^2 + \left(n^{\eps - 1/2} - n^{-1} + 1\right)^2 \approx n\cdot n^{2\eps - 1},\]
    thus $\ntr{(\cA'')^2}\gtrsim n^{2\eps - 1}$ and we can apply \cref{fclt} to matrices $A'$ and $A''$.

    Using \cref{fclt} on $A' = \ntr{A'}I + \cA'_{\dm} = \cA'_{\dm}$, we get that $(L_n(f, I), L_n(f, \cA'_{\dm}))$ is close in a sense of moments to $(\xi_{\tr}(f), \xi_{\dm}(f, \cA'_{\dm})) + O_{\mm}(n^{-1/2})$, where by \cref{fcltparams}, $\E{\xi_{\dm}(f, \cA'_{\dm})^2} = n^{2\eps - 1}\left(1/2 - 1088/(225\pi^2)\right)$, so $L_n(f, A')$ is a Gaussian random variable with mean $0$ and variance $n^{2\eps - 1}\left(1/2 - 1088/(225\pi^2)\right)$.
    Note that since $A' = \cA'_{\dm}$, $\left|\E{\sum_{i = 1}^nf(\lambda_i)\dpr{u_i}{\cA_{\dm}'u_i}}\right| = O_\eps(n^{\eps-1/2})$, so the main contribution into $\E{X_{kk}^+}$ comes from $A''$.
    Similarly, using \cref{fclt} on $A'' = \ntr{A''}I + \cA''_{\dm} = I/n + \cA''_{\dm}$, we get that $(L_n(f, I), L_n(f, \cA''_{\dm}))$ is close in a sense of moments to $(\xi_{\tr}(f), \xi_{\dm}(f, \cA''_{\dm})) + O_{\mm}(n^{-1/2})$, where by \cref{fcltparams}, $\E{\xi_{\tr}(f)^2} = 2/(9\pi^2)$ and $\E{\xi_{\dm}(f, \cA''_{\dm})^2} \approx n^{2\eps - 1}\left(1/2 - 1088/(225\pi^2)\right)$, so so $L_n(f, A'')$ is a sum of two independent Gaussian random variables, one with mean $0$ and variance $2/(9\pi^2)\ntr{A''}^2$, and another with mean $0$ and variance $\approx n^{2\eps - 1}\left(1/2 - 1088/(225\pi^2)\right)$.
    In addition, \[\E{\sum_{i = 1}^nf(\lambda_i)\dpr{u_i}{A''u_i}} = \ntr{A''}\E{\sum_{i = 1}^nf(\lambda_i)} + \E{\sum_{i = 1}^nf(\lambda_i)\dpr{u_i}{\cA_{\dm}''u_i}} = \frac{4}{3\pi} + O_\eps(n^{\eps - 1/2}).\]
    As a corollary, for every $\eps \in (0, 1/2)$,
    \[\E{X_{kk}^+} = \E{\sum_{i = 1}^nf(\lambda_i)\dpr{u_i}{A'u_i}}+ \E{\sum_{i = 1}^nf(\lambda_i)\dpr{u_i}{A''u_i}} = \frac{4}{3\pi} + O_\eps(n^{\eps - 1/2}),\]
    and the expression $X_{kk}^+ - \E{X_{kk}^+} = L_n(f, A)$, being a linear function of $L_n(f, I)$, $L_n(f, \cA'_{\dm})$, $L_n(f, \cA''_{\dm})$, by \cref{fclt} satisfies that for every $j\in \N$, for any small $\delta > 0$, it holds
    \[\E{L_n(f, A)^j} = \E{\xi(f, A)^j} + O(n^{\delta - 1/2}),\]
    where $\xi(f, A)$ is a Gaussian random variable (of order $O(1)$) with mean $0$ and variance
    \[\alpha^2 := \E{\xi(f, A)^2} = \ntr{A''}^2\E{\xi_{\tr}(f)^2} + \E{\xi_{\dm}(f, \cA'_{\dm})^2} + \E{\xi_{\dm}(f, \cA''_{\dm})^2} \simeq n^{2\eps - 1}\left(1-\frac{2176}{225\pi^2}\right).\]
    Then, by Chebyshev inequality,
    \[\P{\left|X_{kk}^+ - \E{X_{kk}^+}\right| > t} = \P{\left|L_n(f, A)\right| > t} \leq \frac{\E{L_n(f, A)^2}}{t^2} \leq \frac{\E{\xi(f, A)^2}}{t^2} + O\left(\frac{1}{t^2\sqrt{n}}\right).\]
    Since $\E{\xi(f, A)^2}$ is of order $n^{2\eps - 1}$, choosing $\eps = 1/4$ proves \cref{diagconc2}.

\end{document}